\documentclass[acmsmall,authorversion,nonacm]{acmart}
\AtBeginDocument{%
  \providecommand\BibTeX{{%
    \normalfont B\kern-0.5em{\scshape i\kern-0.25em b}\kern-0.8em\TeX}}}

\setcopyright{none}
\usepackage{xcolor}
\usepackage{xspace}
\usepackage{tikz}
\usepackage{cprotect}
\usepackage{algorithmicx}
\usepackage[plain]{algorithm}
\usepackage{algpseudocode}
\usepackage{listings}
\usepackage{multirow}
\usepackage{pifont}
\usepackage{enumitem}
\usepackage{subcaption}
\usepackage{mathtools}
\usepackage{amsmath}
\usepackage{scalerel}
\usepackage{wrapfig}
\usepackage{array}
\usepackage{adjustbox}
\usepackage{wasysym}
\usepackage{diagbox}

\usetikzlibrary{automata,positioning}

\author{Amir Shaikhha}
\affiliation{
  \institution{University of Edinburgh}            %%
  \country{United Kingdom}                    %% \country is recommended
}

\author{Mathieu Huot}
\affiliation{
  \institution{University of Oxford}            %%
  \country{United Kingdom}                    %% \country is recommended
}

\author{Jaclyn Smith}
\affiliation{
  \institution{University of Oxford}            %%
  \country{United Kingdom}                    %% \country is recommended
}

\author{Dan Olteanu}
\affiliation{
  \institution{University of Zurich}            %%
  \country{Switzerland}                    %% \country is recommended
}

\definecolor{forestgreen}{rgb}{0.13, 0.55, 0.13}
\colorlet{myblue}{blue!70!black}
\colorlet{mygreen}{green!70!black}

\newcommand{\revision}[1]{#1}

\lstdefinelanguage{llql}%
{morekeywords={
  if,then,else,let,in,not,%
  build,dom,true,false,key,val,%
  sum,for,iter,range,%
  int,real,dense_int,bool,string,rng,promote,%
  prod,rec,closure
  },%
  sensitive,%
  morecomment=[l]//,%
  morecomment=[s]{/*}{*/},%
  morestring=[b]",%
  morestring=[b]',%
  showstringspaces=false,%
  morecomment=[s][\color{gray}]{@}{\ },%
    breaklines=true,%
  mathescape=true,%
showspaces=false,
showtabs=false,
showstringspaces=false,
breakatwhitespace=true,
  aboveskip=1pt,%\smallskipamount,
  belowskip=1pt,%\negsmallskipamount,
  lineskip=-0.2pt,
   numbersep=5pt,
   numberstyle=\tiny\ttfamily,
   basicstyle=\small\ttfamily,
   keywordstyle=\bfseries\color{myblue},%
   columns=fullflexible,
  frame=single,
  xrightmargin=1ex,
  escapeinside={(*@}{@*)}
}[keywords,comments,strings]%

\lstdefinelanguage{pseudo}%
{morekeywords={
  Dictionary%
  },%
  sensitive,%
  morecomment=[l]//,%
  morecomment=[s]{/*}{*/},%
  morestring=[b]",%
  morestring=[b]',%
  showstringspaces=false,%
  morecomment=[s][\color{gray}]{@}{\ },%
    breaklines=true,%
  mathescape=true,%
showspaces=false,
showtabs=false,
showstringspaces=false,
breakatwhitespace=true,
  aboveskip=1pt,%\smallskipamount,
  belowskip=1pt,%\negsmallskipamount,
  lineskip=-0.2pt,
   numbersep=5pt,
   numberstyle=\tiny\ttfamily,
   basicstyle=\small\ttfamily,
  keywordstyle=\small\bfseries\color{green!50!black},%
      commentstyle     = \color{forestgreen},%
  columns=fullflexible,
  frame=single,
  escapeinside={(*@}{@*)}
}[keywords,comments,strings]%

\lstdefinelanguage{NRC}{
  morekeywords={for, in, union, if, then, else, match, let, groupBy, sumBy},%
  sensitive,%
  morecomment=[l]//,%
  morecomment=[s]{/*}{*/},%
  morestring=[b]",%
  morestring=[b]',%
  showstringspaces=false,%
  breaklines=true,%
  mathescape=true,%
  showspaces=false,
  showtabs=false, 
  showstringspaces=false,
  breakatwhitespace=true,
  xleftmargin=1em,
  aboveskip=1pt,%\smallskipamount,
  belowskip=1pt,%\negsmallskipamount,
  lineskip=-0.2pt,
  basicstyle=\small\ttfamily\color{white!15!black},
  keywordstyle=\small\ttfamily\bfseries\color{gray!80!black},%
  columns=fullflexible,
  escapeinside={(*@}{@*)}
}[keywords,comments,strings]%

\lstdefinelanguage{Scala}%
{morekeywords={abstract,%
  case,catch,char,class,%
  def,else,extends,final,finally,for,%
  if,import,implicit,%
  match,module,%
  new,null,%
  object,override,%
  package,private,protected,public,%
  for,public,return,super,%
  this,throw,trait,try,type,%
  val,var,%
  with,while,%
  yield%
  },%
  sensitive,%
  morecomment=[l]//,%
  morecomment=[s]{/*}{*/},%
  morestring=[b]",%
  morestring=[b]',%
  showstringspaces=false,%
  breaklines=true,%
  mathescape=true,%
  showspaces=false,
  showtabs=false,
  showstringspaces=false,
  breakatwhitespace=true,
  xleftmargin=2em,
  aboveskip=1pt,%\smallskipamount,
  belowskip=1pt,%\negsmallskipamount,
  lineskip=-0.2pt,
   basicstyle=\ttfamily,
  keywordstyle=\ttfamily\color{blue!80!black},
  columns=fullflexible,
}[keywords,comments,strings]%

\lstMakeShortInline[columns=fixed, language=NRC]@

\lstMakeShortInline[columns=fixed, keepspaces=true, language=llql]!

\newcommand{\lang}{SDQL\xspace}
\newcommand{\langext}[1]{SDQL[#1]\xspace}
\newcommand{\langring}{\langext{\code{ring}}}
\newcommand{\langclosure}{\langext{\code{closure}}}
\newcommand{\langproduct}{\langext{\code{prod}}}
\newcommand{\langrec}{\langext{\code{rec}}}
\newcommand{\system}{\lang}
\newcommand{\taco}{taco\xspace}
\newcommand{\nrcplus}{NRC$^+$\xspace}
\newcommand{\nrcagg}{NRC$^{agg}$\xspace}
\newcommand{\grammarcomment}[1]{\textit{\small #1}}
\newcommand{\smartpara}[1]{\noindent \textbf{#1.}}

\newcommand{\code}[1]{\texttt{#1}}

\let\bowtie\relax
\DeclareFontEncoding{LS1}{}{}
\DeclareFontSubstitution{LS1}{stix}{m}{n}
\DeclareSymbolFont{STIXsymbols}{LS1}{stixscr}{m}{n}
\DeclareMathSymbol{\bowtie}{\mathrel}{STIXsymbols}{"0E}

\newcommand{\tab}{\;\;\;}

\newcommand{\zero}[1]{$\textbf{0}_{\texttt{#1}}$}
\newcommand{\one}[1]{$\textbf{1}_{\texttt{#1}}$}
\newcommand{\evalsto}{$\rightarrow$}
\newcommand{\transto}{\text{ $\leadsto$ }}

\newcommand{\myeqq}{\triangleq}
\newcommand{\myeq}{$\myeqq$}
\newcommand{\myotimes}{$\otimes_\texttt{S}$}

\newcommand{\translatebegin}{$\llbracket$}
\newcommand{\translateend}{$\rrbracket$}
\newcommand{\translate}[1]{\translatebegin\text{#1}\translateend}

\newcommand{\dsbegin}{\translatebegin}
\newcommand{\dsgend}[1]{\translateend{}$_{\text{#1}}$}
\newcommand{\dsend}[1]{\dsgend{\dsctx}}
\newcommand{\dsgen}[2]{\dsbegin\text{#1}\dsgend{#2}}
\newcommand{\densem}[1]{\dsgen{#1}{\dsctx}}
\newcommand{\dsnoctx}[1]{\translate{#1}}
\newcommand{\dsctx}{$\gamma$}
\newcommand{\mytimes}{\bullet}

\newenvironment{myexample}[1]
{ 
\noindent \textbf{Example #1.}
}
{
}

\newcommand{\cmark}{\ding{51}}%
\newcommand{\xmark}{\ding{55}}%
\newcommand{\langkw}[1]{\texttt{\small\bfseries\color{myblue}#1}}

\DeclareMathOperator*{\sumplus}{\scalerel*{\oplus}{\textstyle\sum}}

\newcolumntype{R}[2]{%
    >{\adjustbox{angle=#1,lap=\width-(#2)}\bgroup}%
    l%
    <{\egroup}%
}
\newcommand*\rot{\multicolumn{1}{R{90}{0em}|}}
\newcommand{\supfull}{\CIRCLE}
\newcommand{\suphalf}{\LEFTcircle}
\newcommand{\supnone}{\Circle}

\begin{document}

\title{Functional Collection Programming with Semi-ring Dictionaries}

% \author{Authors omitted due to double-blind review process.}


%mandatory

\begin{abstract}
This paper introduces semi-ring dictionaries, a powerful class of compositional and purely functional collections that subsume other collection types such as sets, multisets, arrays, vectors, and matrices. We developed SDQL, a statically typed language that can express relational algebra with aggregations, linear algebra, and functional collections over data such as relations and matrices using semi-ring dictionaries.
Furthermore, thanks to the algebraic structure behind these dictionaries, SDQL unifies a wide range of optimizations commonly used in databases (DB) and linear algebra (LA).
As a result, SDQL enables efficient processing of hybrid DB and LA workloads, by putting together optimizations that are otherwise confined to either DB systems or LA frameworks.
We show experimentally that a handful of DB and LA workloads can take advantage of the SDQL language and optimizations.
SDQL can be competitive with or outperforms a host of systems that are state of the art in their own domain:  in-memory DB systems Typer and Tectorwise for (flat, not nested) relational data;  SciPy for LA workloads; sparse tensor compiler taco; the Trance nested relational engine; and the in-database machine learning engines LMFAO and Morpheus for hybrid DB/LA workloads over relational data.
\end{abstract}

\maketitle

\section{Introduction}
The development of domain-specific languages (DSLs) for data analytics has been an important research topic across many communities for more than 40 years.
The DB community has produced SQL, one of the most successful DSLs based on the relational model of data~\cite{rel_codd}.
For querying complex nested objects, the nested relational algebra~\cite{Buneman:1995:PPC:210500.210501} was introduced, which relaxes the flatness requirement of the relational data model.
The PL community has built language-integrated query languages~\cite{linq} and functional collection DSLs based on monad calculus~\cite{roth1988extended}.
Finally, the HPC community has developed various linear algebra frameworks for tensors~\cite{vasilache2018tensor,Kjolstad:2017:TAC:3152284.3133901}.

The main contribution of this paper is \lang, a purely functional language that is simple, canonical, efficient, and expressive enough for \revision{hybrid} database (DB) and linear algebra (LA) workloads. 
In this language, the data is presented as dictionaries over semi-rings, which subsume collection types such as sets, multisets, arrays, and tensors.
% \lang  subsumes relational, nested relational, and linear algebra.

Furthermore, \lang unifies optimizations with inherent similarities that are otherwise developed in isolation. 
Consider the following relational and linear algebra expressions:

\vspace{-0.5cm}

$$Q(a, d)=\Gamma^{\#}_{a,d}R_1(a, b) \bowtie R_2(b, c) \bowtie R_3(c, d)$$
$$N(i, l)=\Sigma_{j,k}M_1(i, j) \cdot M_2(j, k) \cdot M_3(k, l)$$

% \vspace{-0.5cm}

\noindent The expression $Q$ computes the number of paths between each two nodes $(a, d)$ via the binary relations $R_1$, $R_2$, and $R_3$. The expression $N$ computes the matrix representing the multiplication chain of matrices $M_1$, $M_2$, and $M_3$.
These expressions are optimized as:

\vspace{-0.5cm}

$$Q'(a, c) = \Gamma^{\#}_{a,c}R_1(a, b) \bowtie R_2(b, c) \tab \tab \tab Q(a, d)=\Gamma^{\#}_{a,d}Q'(a, c) \bowtie R_3(c, d)$$
$$N'(i, k)=\Sigma_{j}M_1(i, j) \cdot M_2(j, k) \tab \tab \tab N(i, k)=\Sigma_{k}N'(i, k) \cdot M_3(k, l)$$

\noindent The similarity between these two is not a coincidence;
in both cases, two intermediate results are factored out ($Q'$ and $N'$), thanks to the opportunity provided by the distributivity law. 
This is because of the semi-ring structure behind both relational and linear algebra: natural number and real number semi-rings. 
These optimizations are known as \textit{pushing aggregates past joins}~\cite{groupbybeforejoin} and \textit{matrix chain ordering}~\cite{clrs}, respectively.
% \lang unifies several optimizations common in relational databases/probabilistic graphical models~\cite{abo2016faq}.
% Particularly advantageous for hybrid workloads, \lang can also apply algebraic optimizations inside and across the boundary of different domains.
%Additionally, by expressing hybrid data analytics workloads in this language, one can perform algebraic optimizations inside and across the boundary of different domains. %by expressing hybrid workloads in 

\revision{
\subsection{Contributions}
This paper makes the following contributions.
}

\begin{itemize}[leftmargin=*]
    \item We introduce dictionaries with semi-ring structure (Section~\ref{sec:semiring_const}). \revision{Semi-ring dictionaries realize the well-known connection between relations and tensors~\cite{abo2016faq}.}
    \item \revision{We introduce \lang, a statically typed and functional language over such dictionaries. The kind/type system of \lang keeps track of the semi-ring structure (Section~\ref{sec:lang}).}
    \revision{\lang can be used as an intermediate language for data analytics; 
programs expressed in (nested) relational algebra (Section~\ref{sec:db}) or linear algebra-based languages (Section~\ref{sec:la}) can be translated to \lang.\footnote{\revision{In this paper, by (nested) relational and linear algebra, we mean the corresponding sets of operators presented in Figures~\ref{fig:ra_ndql}-\ref{fig:la_ndql}.}}} 
    \item \revision{The unified formal model provided by \lang allows tighter 
integration of data science pipelines that are otherwise developed in loosely coupled frameworks for different domains. This makes \lang particularly}
advantageous for hybrid workloads such as in-database machine learning and linear algebra over nested biomedical data\revision{;} \lang can uniformly apply loop optimizations \revision{(including vertical and horizontal loop fusion, loop-invariant code motion, loop factorization, and loop memoization)} inside and across the boundary of different domains\revision{. We also show how we can synthesize efficient query processing algorithms (e.g., hash join and group join) based on these optimizations} (Section~\ref{sec:opt}).
    \item \revision{Thanks to the compositional structure of semi-ring dictionaries, \lang unifies alternative representations for relations: row/columnar vs. curried layouts, and tensors: coordinate (COO) vs. compressed formats (Section~\ref{sec:datalayout}).}
    \item We give denotational semantics using $0$-preserving functions between K-semi-modules, and prove the correctness of \lang optimizations (Section~\ref{sec:sem}).
    \item We implemented a prototype compiler and runtime for \lang (Section~\ref{sec:impl}). We show experimentally (Section~\ref{sec:exp}) that \lang can be competitive with or outperforms a host of systems that are state-of-the-art in their own domain and that are not designed for the breadth of workloads and data types supported by \lang. \lang achieves similar performance to the in-memory DB systems Typer and Tectorwise. It is on average $2\times$ faster  than SciPy for sparse LA and has similar performance to taco for sparse tensors. For nested data, it outperforms the Trance nested relational engine by up to an order of magnitude. For hybrid DB/LA workloads over flat relational data, \lang has on average slightly better performance than the in-DB ML engines LMFAO and Morpheus.
\end{itemize}

\vspace{0.2cm}

\smartpara{Motivating Example}
The following setting is used throughout the paper to exemplify \lang.
Biomedical data analysis presents an interesting domain for language development.
Biological data comes in a variety of %domain-specific 
formats that use complex data models~\cite{biodata}.
% In addition, the affordability of genomic sequencing, the advances of image processing, and the improvement of 
% medical data management present many opportunities to integrate complex datasets and develop targeted treatments~\cite{pm}. 
Consider a biomedical analysis focused on the role of mutational burden in cancer.
High tumor mutational burden (TMB) has been shown to be a confidence biomarker for cancer therapy response \cite{tmbref1, tmbref2}. 
A subcalculation of TMB is gene mutational burden (GMB). Given a set of genes and variants for each sample, GMB associates variants to genes and counts the total number of mutations present in a given gene per tumor sample. 
% The !Genes! input is a relational-formatted input containing metadata about a gene and the !Variants! input is a domain-specific, variant call format (VCF) \cite{vcf}, which contains top-level variant information and nested mutational information for each sample. 
This analysis provides a 
basic measurement of how impacted a given gene is by somatic mutations, which can be used directly 
as a likelihood measurement for immunotherapy response \cite{tmbref1}, or can be used as features to predict patient response to therapy or the severity of the patient's cancer. 

The biological community has developed countless DSLs to perform such analyses~\cite{genometric, hail,cromwell}.
% ; these languages can be tightly scoped to a class of transformations \cite{genometric, hail} or describe more generic workflow execution tasks \cite{cromwell}. 
Modern biomedical analyses also leverage SQL-flavoured query languages and machine learning frameworks for classification. An analyst may need to use multiple languages to perform integrative tasks, and additional packages downstream to perform inference.
The development of generic solutions that consolidate and generalize 
complex biomedical workloads is crucial for advancing biomedical 
infrastructure and analyses.

This paper shows the above tasks can be framed in \lang and benefit from optimized execution.

% This paper presents \lang, a statically typed functional language. 
% This language is centered around functional nested dictionaries, and subsumes existing query languages and functional collection DSLs.
% Furthermore, this language can express a wide range of optimizations available in database query engines.
% Finally, basic linear algebra operations can also be expressed in NDQL.
% As a result, NDQL is appropriate for both database and linear algebra workloads.

\section{Language}
\label{sec:lang}
\lang is a purely functional, domain-specific language inspired by efforts from languages developed in both the programming languages (e.g., Haskell, ML, and Scala) and the databases (e.g., AGCA~\cite{dbtoaster} and FAQ~\cite{abo2016faq}) communities.
This language is appropriate for collections with \textit{sparse} structure such as database relations, functional collections, and sparse tensors. Nevertheless, \lang also provides facilities to support dense arrays.

Figure~\ref{fig:lang} shows the grammar of \lang for both expressions (!e!) and types (!T!). 
We first give a background on semi-ring structures.
Then, we introduce the kind and type systems of \lang (cf. Figure~\ref{fig:typesystem}).
Afterwards, we continue by introducing semi-ring and iteration constructs. 
Finally, we show how arrays and sets are encoded in \lang.

\begin{figure*}[t]
\setlength{\tabcolsep}{0.3em}
\centering
\begin{tabular}{|l c l|l|}
\hline
\multicolumn{3}{|c|}{\textbf{Core Grammar}} & \multicolumn{1}{c|}{\textbf{Description}}\\\hline
!e! & \mbox{::=} & !sum(x in e)! !e! \tab $\mid$ \tab !{ e -> e, ... }! & \grammarcomment{Dictionary Aggregation \& Construction} \\
& $\mid$ & !{ }$_{\texttt{T},\texttt{T}}$! \tab  $\mid$ \tab !e(e)! & \grammarcomment{Empty Dictionary, Dictionary Lookup}\\
& $\mid$ & !< a = e, ... >! \tab $\mid$ \tab !e.a! \tab $\mid$ \tab !not e! & \grammarcomment{Record Construction, Field Access, Negation}\\
& $\mid$ & !let x = e in e! \tab $\mid$ \tab !x! \tab $\mid$ \tab !if e then e else e! & \grammarcomment{Variable Binding \& Access, Conditional}\\
& $\mid$ & !e + e! \tab $\mid$ \tab !e * e! \tab $\mid$ \tab !promote$_{\texttt{S},\texttt{S}}$(e)! & \grammarcomment{Addition, Multiplication, Scalar Promotion}\\
& $\mid$ & !n! \tab $\mid$ \tab !r! \tab $\mid$ \tab !false! \tab $\mid$ \tab !true! \tab $\mid$ \tab !c! & \grammarcomment{Numeric, Boolean, and Other Constants}\\
% & $\mid$ & !dn! & \grammarcomment{Dense Integer Value}\\
% & $\mid$ & !"some_text"! & \grammarcomment{String Literal}\\ 
\hline
% !T! & \mbox{::=} & !{ T -> T }! & \grammarcomment{Dictionary Type} \\
% & $\mid$ & !< a:T, ... >! & \grammarcomment{Record Type}\\
% & $\mid$ & !S! & \grammarcomment{Scalar Type}\\
!T! & \mbox{::=} & !{ T -> T }! \;\; $\mid$ \;\; !< a:T, ... >! \;\; $\mid$ \;\; !S! \;\; $\mid$ \;\; !U! & \grammarcomment{Dictionary, Record, Scalar, and Enum Types} \\
!S! & \mbox{::=} & !int! \tab $\mid$ \tab !real! \tab $\mid$ \tab !bool! \tab $\mid$ \tab [cf. Table~\ref{tab:semiring_scalar}] & \grammarcomment{Scalar Semi-Ring Types}\\
!U! & \mbox{::=} & !string! \tab $\mid$ \tab !dense_int! & \grammarcomment{String and Dense Integer Types}\\
\hline
!K! & \mbox{::=} & !Type! \tab $\mid$ \tab !SM(S)! & \grammarcomment{Ordinary \& Semi-Module Kinds}\\
\hline
\end{tabular}
\vspace{-0.3cm}
\caption{Grammar of the core part of \lang. Scalar numeric operations (e.g., \code{sin}) are omitted for brevity.}
\vspace{-0.4cm}
\label{fig:lang}
\end{figure*}

\subsection{Semi-Ring Structures}
\smartpara{Semi-ring}
A semi-ring structure is defined over a data type !S! with two binary operators !+! and !*!. Each binary operator has an identity element; \zero{S} is the identity element for !+! and \one{S} is for !*!. When clear from the context, we use !0! and !1! as identity elements. Furthermore, the following algebraic laws hold for all elements !a!, !b!, and !c!:

% \begin{tabular}{l l l}
% \textit{Associativity} & $a + (b + c) = (a + b) + c$ & $a \times (b \times c) = (a \times b) \times c$ \\
% \end{tabular}

% \begin{itemize}
% \item !a + (b + c) = (a + b) + c!
% \item !0 + a = a + 0 = a!
% \item !a + b = b + a!
% \item !a * (b * c) = (a * b) * c!
% \item !1 * a = a * 1 = a!
% \item !0 * a = a * 0 = 0!
% \item !a * (b + c) = a * b + a * c!
% \item !(a + b) * c = a * c + b * c!
% \end{itemize}

\begin{tabular}{l l l}
!a + (b + c) = (a + b) + c! & !0 + a = a + 0 = a! &  !1 * a = a * 1 = a!\\
!a + b = b + a! & !a * (b * c) = (a * b) * c! &  !0 * a = a * 0 = 0! \\
 !a * (b + c) = a * b + a * c! & !(a + b) * c = a * c + b * c!
\end{tabular}

\noindent The last two rules are distributivity laws, and are the base of many important optimizations for semi-ring structures~\cite{aji2000generalized}. Semi-rings with commutative multiplications (!a*b=b*a!) are called commutative semi-rings.

\smartpara{Semi-module}
The generalization of commutative semi-rings for containers results in a semi-module.
A semi-module over a semi-ring of data type !S! (a !S!-semi-module) is 
defined with an addition operator between two semi-modules, and a multiplication between a semi-ring element and the semi-module. An example is the vector of real numbers with vector addition and scalar-vector multiplication. 
The following laws hold for all the elements !u! and !v! in a !S!-semi-module:

\begin{tabular}{l l l l l}
!a * (u + v) = a * u + a * v! & &
& &
!(u + v) * a = u * a + v * a! 
\\
% !1 * u = u! \\
!(a + b) * u = a * u + b * u! & &
& &
!(a * b) * u = a * (b * u)! 

\end{tabular}

\smartpara{Tensor product}
For two types !T1! and !T2! that are !S!-semi-modules, the tensor product !T1!$\otimes_{\code{S}}$!T2! is another !S!-semi-module. It comes equipped with a canonical map which we also denote using !*: T1!$\times$!T2!$\rightarrow$ !T1!$\otimes_{\code{S}}$!T2! with the following laws for all elements !u1,u2:T1! and !v1,v2:T2!:

\begin{tabular}{l l l}
!u1 * (v1 + v2) = u1 * v1 + u1 * v2! & &
!(u1 + u2) * v1 = u1 * v1 + u2 * v1! \\
!(u1 * a) * v1 = u1 * (a * v1)! & &
!1 * u1 = u1!
\end{tabular}

\begin{figure*}[t]
\begin{tabular}{|c|}
\hline
\begin{tabular}{|c|}
\hline
Kind System:\\
!T :: K!\\
\hline
\end{tabular}
\hfill
\begin{tabular}{c}
\\\hline
!S::SM(S)!
\end{tabular}
\hspace{0.35cm}
\begin{tabular}{c}
$\forall$ !i!. !Ti::SM(S)!
\\\hline
!<a1:T1,...,an:Tn>::SM(S)!
\end{tabular}
\hspace{0.35cm}
\begin{tabular}{c}
!T1::K! $\;$ !T2::SM(S)!
\\\hline
!{T1->T2}::SM(S)!
\end{tabular}\\
\begin{tabular}{c}
!T1::SM(S)! $\;$ !T2::SM(S)!
\\\hline
!T1!$\otimes_{\code{S}}$!T2::SM(S)!
\end{tabular}
\hspace{0.15cm}
\begin{tabular}{c}
\\\hline
!U::Type!
\end{tabular}
\hspace{0.15cm}
\begin{tabular}{c}
$\exists$ !i!. !Ti::Type!
\\\hline
!<a1:T1,...,an:Tn>::Type!
\end{tabular}
\hspace{0.15cm}
\begin{tabular}{c}
!T1::K! $\;$ !T2::Type!
\\\hline
!{T1->T2}::Type!
\end{tabular}
\\\hline
\begin{tabular}{|c|}
\hline
Type System:\\
$\Gamma \vdash$ !e : T!\\
\hline
\end{tabular}
\hfill
\begin{tabular}{c}
!c!: !T! \\\hline
$\Gamma \vdash$ !c!: !T!
\end{tabular}
\hspace{0.25cm}
\begin{tabular}{c}
!x!: !T! $\in \Gamma$ \\\hline
$\Gamma \vdash$ !x!: !T!
\end{tabular}
\hspace{0.25cm}
\begin{tabular}{c}
$\Gamma \vdash$ !e1!: !T1! $\quad$ $\Gamma$, !x!: !T1! $\vdash$ !e2!: !T2! \\\hline
$\Gamma \vdash$ !let x = e1 in e2!: !T2!
\end{tabular}
\hspace{0.25cm}
\begin{tabular}{c}
$\Gamma \vdash$  !e!: !bool! \\\hline
$\Gamma \vdash$  !not e!: !bool!
\end{tabular}
\\
\begin{tabular}{c}
$\Gamma \vdash$  !e1!: !bool! $\quad$ $\Gamma \vdash$ !e2!: !T! $\quad$ $\Gamma \vdash$ !e3!: !T! \\\hline
$\Gamma \vdash$ !if(e1) then e2 else e3!: !T!
\end{tabular}
\hspace{0.5cm}
\begin{tabular}{c}
$\Gamma \vdash$  !e!: !S1!\\\hline
$\Gamma \vdash$  !promote$_{\texttt{S1, S2}}$(e)!: !S2!
\end{tabular}
\\
\begin{tabular}{c}
$\Gamma \vdash$ !e1!: !{T1->T2}! $\quad$ $\Gamma$, !x!: !<key:T1,val:T2>! $\vdash$ !e2!:!T3! $\quad$ !T3::SM(S)! \\\hline
$\Gamma \vdash$ !sum(x in e1) e2!: !T3!
\end{tabular}
\hspace{0.25cm}
\begin{tabular}{c}
\\\hline
$\Gamma \vdash$ !{}$_{\texttt{T1},\texttt{T2}}$! : ! {T1 -> T2}!
\end{tabular}
\\
\begin{tabular}{c}
$\Gamma \vdash$ !k1!: !T1! $\,\,$ $\Gamma \vdash$ !v1!: !T2! $\,\,$ ... $\,\,$ $\Gamma \vdash$ !kn!: !T1! $\,\,$ $\Gamma \vdash$ !vn!: !T2!\\\hline
$\Gamma \vdash$ !{ k1 -> v1, ..., kn -> vn }!: !{ T1 -> T2 }!
\end{tabular}
\hspace{0.25cm}
\begin{tabular}{c}
$\Gamma \vdash$ !e1!: !{ T1 -> T2 }! $\quad$ $\Gamma \vdash$  !e2!: !T1!\\\hline
$\Gamma \vdash$ !e1(e2)!: !T2!
\end{tabular}
\\
\begin{tabular}{c}
$\Gamma \vdash$ !e1!: !T1! $\quad$ ... $\quad$ $\Gamma \vdash$ !en!: !Tn!\\\hline
$\Gamma \vdash$ !<a1=e1,...,an=en>!: !<a1:T1,...,an:Tn>!
\end{tabular}
\hspace{0.5cm}
\begin{tabular}{c}
$\Gamma \vdash$ !e!: !<a1:T1,...,ak:Tk>!\\\hline
$\Gamma \vdash$ !e.ai!: !Ti!
\end{tabular}
\\
\begin{tabular}{c}
$\Gamma \vdash$  !e1!: !T! $\quad$ !e2!: !T! $\quad$ !T::SM(S)!\\\hline
$\Gamma \vdash$  !e1 + e2!: !T!
\end{tabular}
\hspace{0.5cm}
\begin{tabular}{c}
$\Gamma \vdash$  !e1!: !T1! $\quad$ $\Gamma \vdash$ !e2!: !T2! $\quad$ !T1::SM(S)! $\quad$ !T2::SM(S)!\\\hline
$\Gamma \vdash$  !e1 * e2!: !T1$\otimes_{\code{S}}$T2!
\end{tabular}
\\ \hline
\begin{tabular}{|c|}
\hline
Definition of \myotimes:\\
$\forall$ !i!. !Ti :: SM(S)!\\
\hline
\end{tabular}
\hfill
\begin{tabular}{c}
!S! \myotimes{} !T1! \myeq{} !T1! $\quad$ 
% !T1! \myotimes{} !S! \myeq{} !T1! $\quad$ 
!{T1 -> T2}! \myotimes{} !T0! \myeq{} !{T1 -> T2! \myotimes{} !T0}! \\
!<a1:T1, ..., an:Tn>! \myotimes{} !T0! \myeq{} !<a1:T1!\myotimes{}!T0, ...,an:Tn!\myotimes{}!T0>!
\end{tabular}\\
\hline
\end{tabular}
\vspace{-0.4cm}
\caption{Kind System and Type System of \lang.}
\vspace{-0.5cm}
\label{fig:typesystem}
\end{figure*}

\subsection{Kind System and Type System}
Figure~\ref{fig:typesystem} shows the kind/type system of \lang. 
The types with a semi-ring structure have the kind !SM(S)!; 
semi-ring dictionaries with !S!-semi-module value types are also !S!-semi-modules (i.e., they have the kind !SM(S)!). 
However, dictionaries with value types of the ordinary kind !Type! are of kind !Type!.
Similar patterns apply to records. 

\begin{myexample}{1}
Both types !{ string -> int }! and !<c: int>! are of kind !SM(int)!. However, the types !{string -> string}! and !<d: string>! are of kind !Type!.
\end{myexample}

The addition of two expressions requires both operands to have the same type of kind !SM(S)!. This means that the body of summation also needs to have a type of kind !SM(S)!.
The type system rules for the multiplication operator are defined inductively. 
Multiplying a scalar with a dictionary results in a dictionary with the same keys, but with the values multiplied with the scalar value.
Multiplying a dictionary with another term also results in a dictionary with the same keys, and values multiplied with that term. Note that the multiplication operator is not commutative in general.\footnote{To be more precise, the scalar \code{*} is commutative, but the tensor product \texttt{*} is commutative up to reordering.}  The typing rules for the multiplication of record types are defined similarly.

\begin{myexample}{1 (Cont.)}
Assume a dictionary term !d! with type !{ string -> int }!, and a record term !r! with type !<c: int>!. 
The type of the expression !d * r! is !{ string -> int }!$\otimes_{\langkw{int}}$!<c: int>!, which is !{ string -> <c: int> }!, as can be confirmed by the typing rules.
\end{myexample}

\begin{table}[t]
    \caption{Different semi-ring structures for scalar types.}
    \label{tab:semiring_scalar}
    \vspace{-0.4cm}
    \centering
    \begin{tabular}{|l|c|c|c|c|c|c|c|}
    \hline
        \textbf{Name} & \textbf{Type} & \textbf{Domain} & \textbf{Addition} & \textbf{Multiplication} & \textbf{Zero} & \textbf{One} & \textbf{Ring} \\ \hline
        Real Sum-Product & !real! & $\mathbb{R}$ & $+$ & $\times$ & !0! & !1! &  \cmark \\ \hline
        Integer Sum-Product & !int! & $\mathbb{Z}$  & $+$ & $\times$ & !0! & !1! &  \cmark \\ \hline
        Natural Sum-Product & \langkw{nat} & $\mathbb{N}$ & $+$ & $\times$ & !0! & !1! &  \xmark \\ \hline
        Min-Product & \langkw{mnpr} & $(0, \infty]$ & min & $\times$ & $\infty$ & !1! & \xmark \\ \hline
        Max-Product & \langkw{mxpr} & $[0, \infty)$ & max & $\times$ & !0! & !1! & \xmark \\ \hline
        Min-Sum & \langkw{mnsm} & $(-\infty, \infty]$ & min & $+$ & $\infty$ & !0! & \xmark \\ \hline
        Max-Sum & \langkw{mxsm} & $[-\infty, \infty)$ & max & $+$ & $-\infty$ & !0! & \xmark \\ \hline
        Max-Min & \langkw{mxmn} & $[-\infty, \infty]$ & max & min & $-\infty$ & $+\infty$ & \xmark \\ \hline
        Boolean & !bool! & $\{T, F\}$ & $\vee$ & $\land$ & !false! & !true! &  \xmark \\ \hline
    \end{tabular}
    \vspace{-0.5cm}
\end{table}

\subsection{Semi-Ring Constructs}
\label{sec:semiring_const}
\smartpara{Scalars}
Values of type !bool! form the \textit{Boolean Semi-Ring}, with disjunction and conjunction as binary operators, and !false! and !true! as identity elements.
Values of type !int! and !real! form \textit{Integer Semi-Ring} ($\mathbb{Z}$) and \textit{Real Semi-Ring} ($\mathbb{R}$), respectively. Table~\ref{tab:semiring_scalar} shows an extended set of semi-rings for scalar values. Both addition and multiplication only support elements of the same scalar type.

\smartpara{Promotion}
Performing multiplications between elements of different scalar data types requires explicitly \textit{promoting} the operands to the same scalar type. Promoting a scalar term !s! of type !S1! to type !S2! is achieved by !promote$_{\texttt{S1}, \texttt{S2}}$(s)!.

\smartpara{Dictionaries}
\label{sec:sd}
A dictionary with keys of type !K!, and values of type !V! is represented by the data type !{ K -> V }!.
The expression !{ k_1 -> v_1, ..., k_n -> v_n }!, constructs a dictionary of !n! elements with keys !k_1, ..., k_n! and values !v_1, ..., v_n!. The expression !{}!$_{\code{K,V}}$ constructs an empty dictionary of type !{ K -> V }!, and we might drop the type subscript when it is clear from the context.
The expression !dict(k)! performs a lookup for key !k! in the dictionary !dict!.
% The elements of a dictionary are key-value pairs, which can be seen as records with 
% field names !key! and !val!.

If the value elements with type !V! form a semi-ring structure, then the dictionary also forms a semi-ring structure, referred to as a semi-ring dictionary (SD) where the addition is point-wise, that is the values of elements with the same key are added.
The elements of an SD with \zero{V} as values are made implicit and can be removed from the dictionary. 
This means that two SDs with the same set of !k_i! and !v_i! pairings are equivalent regardless of their \zero{V}-valued !k_j!s.

% The multiplication operator for SDs is a generalization of outer products, and is exactly outer-products  when it is restricted to types defined by scalar types and records. 
% In more detail, the 
The multiplication !dict * s!, where !dict! is an SD with !k_i! and !v_i! as keys and values, results in an SD with !k_i! as the keys, and !v_i * s! as the values. For the expression !s * dict!, where !s! is not an SD and !dict! is an SD with keys !k_i! and values !v_i!, the result is an SD with !k_i! as keys and !s * v_i! as values. 
Note that the multiplication operator is not commutative by default.

\begin{myexample}{2}
Consider the following two SDs: !{ "a"->2, "b"->3 }! named as !dict1! and !{ "a"->4, "c"->5 }! named as !dict2!.
The result of !dict1+dict2! is !{ "a"->6, "b"->3, "c"->5 }!.
This is because !dict1! is equivalent to !{ "a"->2, "b"->3, "c"->0 }! and !dict2! is equivalent to !{ "a"->4, "b"->0, "c"->5 }!, and element-wise addition of them results in !{ "a"->2+4, "b"->3+0, "c"->0+5 }!.

The result of !dict1 * dict2! is !{ "a"->2 * dict2, "b"->3 * dict2 }!. The expression !2 * dict2! is evaluated to !{ "a"->2*4, "c"->2*5 }!. By performing similar computations, !dict1 * dict2! is evaluated to !{ "a"->{ "a"->8, "c"->10 }, "b"->{ "a"->12, "c"->15 } }!.
On the other hand, !dict2 * dict1! is !{ "a"->4 * dict1, "c"->5 * dict1 }!. After performing similar computations, the expression is evaluated to !{ "a"->{ "a"->8, "b"->12 }, "c"->{ "a"->10, "b"->15 } }!.
\end{myexample}

\smartpara{Records}
Records are constructed using !< a_1 = e_1, ..., a_n = e_n >! and
the field !a_i! of record !rec! can be accessed using !rec.a_i!.
When all the fields of a record are !S!-semi-modules, the record also forms an !S!-semi-module.

% The multiplication operator !*! of records is defined similarly to the one for SDs. The multiplication of a record with another semi-ring element results in a record with the same field names, but with the elements multiplied by the other operand.
% The multiplication of a scalar element with a record results in a record with the same fields and the elements left-multiplied by the other operand.
% % This means that in the body of the loop !sum(r in R)!, one has to use !r.key! and !r.val! to access the key and value of !r!, respectively.

\begin{myexample}{1 (Cont.)}
Assume the dictionary !d! with the value !{ "a"->2, "b"->3 }!, and the record !r! with the value !< c=4 >!.
The expression !d * r! is evaluated as !{ "a" -> <c=8>, "b" -> <c=12> }!.
\end{myexample}

\subsection{Dictionary Summation}

The expression !sum(x in d) e! specifies iteration over the elements of dictionary !d!, where each element !x! is a record with the attribute !x.key! specifying the key and !x.val! specifying the value. One can alternatively use the syntactic sugar !sum(<k,v> in d) e! that binds !k! to !x.key! and !v! to !x.val! (cf. Figure~\ref{fig:langext}).
This iteration computes the summation of the result of the expression !e! using the corresponding addition operator, and by starting from an appropriate additive identity element.
In the case that !e! has a scalar type, this expression computes the summation using the corresponding scalar addition operator.
If the expression !e! is an SD, then the SD addition is used.
% in the case of !e! with a scalar type, the scalar addition operator is used, and in the case of !e! with an SD type, the SD addition is used.

\begin{myexample}{1 (Cont.)}
Consider the expression !sum(x in d) x.val! where !d! is a dictionary with value of !{ "a" -> 2, "b" -> 3 }!. This expression is evaluated to !5!, which is the result of adding the values (!2 + 3!) in dictionary !d!.
Let us consider the expression !sum(<k,v> in d) { k -> v * 2 }!, with the same value as before for !d!. This expression is evaluated to !{ "a" -> 4, "b" -> 6 }!, which is the result of the addition of !{ "a" -> 2*2 }! and !{ "b" -> 3*2 }!.
\end{myexample}

% \todo{explain more on example 1 and have another example similar to lemma 3 of loop fusion where key also forms a semi-ring}

% \begin{myexample}{1}
% As our working example, we consider the following biological application.
% \end{myexample}

\subsection{Set and Array}
Collection types other than dictionaries, such as arrays and sets, can be defined in terms of dictionaries (cf. Figure~\ref{fig:langext}).
Arrays can be obtained by using \textit{dense integers} (!dense_int!), which are continuous integers ranging from $0$ to $k$
% \footnote{We do not include data types such as \langkw{dense\_int} and other enumerated types (cf. Section~\ref{sec:semiring_ext}) in the grammar and kind/type system, to make the language minimalistic for presentation purposes.}
, as keys and the elements of the array as values.
Sets can be obtained by using the elements of the set as keys and Booleans as values.
Arrays and sets of elements of type !T! are represented as ![| T |]! and  !{ T }!, respectively.

\begin{figure*}[t]
\setlength{\tabcolsep}{0.4em}
\begin{center}
\begin{tabular}{|l|l|l|}
\hline
\multicolumn{1}{|c|}{\textbf{Extension}} & \multicolumn{1}{c|}{\textbf{Definition}} & \multicolumn{1}{c|}{\textbf{Description}}\\ \hline
!if e_0 then e_1 !& !if e_0 then e_1 else !\zero{T} \tab \textit{where} !e_1: T!& \grammarcomment{One-Branch Conditional} \\ \hline
!{ e_0, ..., e_k }! & !{ e_0 -> true,..., e_k -> true }! & \grammarcomment{Set Construction}\\ \hline
!dom(e)! & !sum(x in e) { x.key }! & \grammarcomment{Key Set of Dictionary}\\ \hline
!sum(<k,v> in e)e_1! & !sum(x in e) let k=x.key in let v=x.val in e_1! & \grammarcomment{Sum Paired Iteration}\\ \hline
!range(dn)! & !{ 0 -> true, ..., dn-1 -> true }! & \grammarcomment{Range Construction}\\ \hline
![| e_0,...,e_k |]! & !{ 0 -> e_0, ..., k -> e_k }! & \grammarcomment{Array Construction}\\ \hline \hline
!{ T }! & !{ T -> bool } ! & \grammarcomment{Set Type}\\ \hline 
% !rng! & !{{ dense_int -> bool }} ! & \grammarcomment{Range Type}\\ \hline 
![| T |]! & !{ dense_int -> T } ! & \grammarcomment{Array Type}\\ \hline 
\end{tabular}
\end{center}
\vspace{-0.4cm}
\caption{Extended constructs of \lang.}
\vspace{-0.4cm}
\label{fig:langext}
\end{figure*}

\section{Expressiveness for Databases}
\label{sec:db}
This section analyzes the expressive power of \lang for database workloads.
We start by showing the translation of relational algebra to \lang (Section~\ref{sec:ra}).
Then we show the translation of nested relational calculus to \lang (Section~\ref{sec:nrc}), followed by the translation of 
aggregations (Section~\ref{sec:agg}).
% Finally, we show how \lang can express the hash join operator, an important operator for query processing engines (Section~\ref{sec:hashjoin}). 

\subsection{Relational Algebra}
\label{sec:ra}
% In his seminal work, Codd~\cite{rel_codd} introduced relational model of data and a corresponding algebra, named as relational algebra. 
% Relational algebra became the foundation of many query languages used in database management systems, including SQL.

Relational algebra~\cite{rel_codd} is the foundation of many query languages used in database management systems, including SQL.
In general, a relation $R(a_1, ..., a_n)$ (with set semantics) is represented as a dictionary of type !{ <$a_1$: $A_1$, ..., $a_n$: $A_n$> -> bool }! in \lang.
% By using the variable order of $[a_1, ..., a_n]$, the factorized representation of this relation in \lang is a nested dictionary of type !{ $A_1$ -> { ... -> { $A_n$ -> int } ... } }!.
Figure~\ref{fig:ra_ndql} shows the translation rules for the relational algebra operators. 
\lang can also express different variants of joins including outer/semi/anti-joins.
The explanation of the relational algebra and various join operators can be found in the supplementary materials.

\begin{myexample}{3}
Consider the following data for the !Genes! input, which is a flat relation providing positional information of genes on the genome:

\begin{center}
\begin{footnotesize}
\setlength\tabcolsep{4pt}
\begin{tabular}{c | c | c | c | c | c | c |}
\textbf{Genes} & \textbf{name} & \textbf{desc} & \textbf{contig} & \textbf{start} & \textbf{end} & \textbf{gid}  \\ \hline
& NOTCH2 & notch receptor 2 & 1 & 119911553 & 120100779 & ENSG00000134250 \\ \cline{2-7}
& BRCA1 & DNA repair associate & 17 & 43044295 & 43170245 & ENSG00000012048 \\ \cline{2-7}
& TP53 & tumor protein p53 & 17 & 7565097 & 7590856 & ENSG00000141510 \\ \hline
\end{tabular} 
\end{footnotesize}
\end{center}

\noindent This relation is represented as follows in \lang:

\begin{lstlisting}[frame=none,basicstyle=\scriptsize\ttfamily]
{ <name="NOTCH2",desc="notch receptor 2", contig=1, start=119911553, end=120100779, gid="ENSG00000134250">,
  <name="BRCA1",desc="DNA repair associate", contig=17, start=43044295, end=43170245, gid="ENSG00000012048">, 
  <name="TP53",desc="tumor protein p53", contig=17, start=7565097, end=7590856, gid="ENSG00000141510"> }
\end{lstlisting}

\noindent Only a subset of the attributes in the !Genes! relation are 
commonly used in a biomedical analysis. 
This can be achieved using the following expression:

\noindent
\begin{lstlisting}
sum(<g,v> in Genes) { <gene=g.name,contig=g.contig,start=g.start,end=g.end> }
\end{lstlisting}

\end{myexample}

\smartpara{Inefficiency of Joins}
The presented translation for the join operator is inefficient. 
This is because one has to consider all combinations of elements of the input relations.
In the case of equality joins, this situation can be improved by leveraging data locality as will be shown in Section~\ref{sec:hashjoin}.

\begin{figure*}[t]
\begin{tabular}{|l|r c l|}
\hline
\textbf{Name} & \multicolumn{3}{|l|}{\textbf{Translation}}  \\ \hline
Selection&\translate{$\sigma_p$(R)} &=& !sum(x in! \translate{R} !)  if p(x.key)! !then { x.key }!\\ \hline 
Projection&\translate{$\pi_f$(R)} &=& !sum(x in! \translate{R} !) { f(x.key)  }!\\ \hline 
Union&\translate{R $\cup$ S} &=& \translate{R}! + !\translate{S}\\ \hline
Intersection&\translate{R $\cap$ S} &=& !sum(x in! \translate{R} !)!~!if !\translate{S}!(x.key)! !then { x.key }! \\ \hline
Difference&\translate{R $-$ S} &=& !sum(x in! \translate{R} !)!~!if not! \translate{S}!(x.key)! !then { x.key }! \\ \hline
Cartesian Product &\translate{R $\times$ S} &=& !sum(x in! \translate{R} !)!~!sum(y in! \translate{S} !) { concat(x.key, y.key) }! \\ \hline
Join&\translate{R $\bowtie_\theta$ S} &=& \translate{$\sigma_\theta$(R $\times$ S)} \\ \hline
\end{tabular}
\vspace{-0.4cm}
\caption{Translation from relational algebra (with set semantics) to \lang.}
\label{fig:ra_ndql}
\vspace{-0.4cm}
\end{figure*}

\subsection{Nested Relational Calculus}
\label{sec:nrc}
Relational algebra does not allow nested relations; a relation in the first normal form (1NF) when none of the attributes is a set of elements~\cite{rel_codd}. 
Nested relational calculus allows attributes to be relations as well. 
In order to make the case more interesting, we consider \nrcplus~\cite{incnrc}, a variant of nested relational calculus with \textit{bag semantics} and without difference operator.

Nested relations are represented as dictionaries mapping each row to their multiplicities. 
As the rows can contain other relations, the keys of the outer dictionary can also contain dictionaries.
Figure~\ref{fig:nrc_ndql} shows the translation from positive   nested relational calculus (without difference) to \lang. The explanation on the translation of its constructs can be found in the supplementary material.

\begin{myexample}{4}
Consider the !Variants! input, which contains top-level metadata for genomic variants and nested genotype information for every sample. Genotype calls
denoting the number of alternate alleles in a sample. An example of the nested !Variants! input is as follows: 

\begin{center}
\begin{footnotesize}
\begin{tabular}{c | c | c | c | c | c|}
\textbf{Variants} & \textbf{contig} & \textbf{start} & \textbf{reference} & \textbf{alternate} & \textbf{genotypes} \\ \hline
& 17 & 43093817 & C & A & \begin{tabular}{|c | c|} %rs730881473
\textbf{sample} & \textbf{call} \\ \hline  
TCGA-AN-A046 & 0 \\ \hline
TCGA-BH-A0B6 & 1 \\ \hline
\end{tabular} \\ \cline{2-6}
& 1 & 119967501 & G & C & \begin{tabular}{|c | c|} %COSM527767
\textbf{sample} & \textbf{call} \\ \hline  
TCGA-AN-A046 & 1 \\ \hline
TCGA-BH-A0B6 & 2 \\ \hline
\end{tabular} \\ \hline
\end{tabular} 
\end{footnotesize}
\end{center}

\noindent This nested relation is represented as follows in \lang:

\begin{lstlisting}[frame=none,basicstyle=\footnotesize\ttfamily]
{ <contig=17, start=43093817, reference="C", alternate="A", genotypes=
    { <sample="TCGA-AN-A046", call=0> -> 1, <sample="TCGA-BH-A0B6", call=1> -> 1 } > -> 1,
  <contig=1, start=119967501, reference="G", alternate="C", genotypes= 
    { <sample="TCGA-AN-A046", call=1> -> 1, <sample="TCGA-BH-A0B6", call=2> -> 1 } > -> 1 }
\end{lstlisting}

\end{myexample}

\begin{figure*}[t]
\begin{tabular}{|l|r c l|}
\hline
\textbf{Name} & \multicolumn{3}{|l|}{\textbf{Translation}}  \\ \hline
Let Binding&\translatebegin{}@let X = $e_1$ in $e_2$@\translateend{}
&=& !let X = !\translate{$e_1$}! in !\translate{$e_2$}\\ \hline
Empty Bag&\translate{$\emptyset_T$} &=& !{ }!$_{T,\langkw{int}}$ \\ \hline
% \translate{\nrckw{sng}($e$)} 
Singleton Bag&\translatebegin{}@sng($e$)@\translateend{}
&=& !{ !\translate{$e$} !-> 1 }!\\ \hline
% \translate{\nrckw{flatten}($e$)} 
Flattening&\translatebegin{}@flatten($e$)@\translateend{}
&=& !sum(<k,v> in!\translate{$e$}!) v * k! \\ \hline
% \translate{\nrckw{for} x \nrckw{in} $e_1$ \nrckw{union} $e_2$} 
Monadic Bind&\translatebegin{}@for x in $e_1$ union $e_2$@\translateend{}
&=& !sum(<x,x_v> in!\translate{$e_1$}!) x_v * !\translate{$e_2$} \\ \hline
Union&\translate{$e_1 \uplus e_2$} &=& \translate{$e_1$} !+! \translate{$e_2$} \\ \hline
Cartesian&\translate{$e_1 \times e_2$} &=& !sum(<x,x_v> in! \translate{$e_1$} !) sum(<y,y_v> in! \translate{$e_2$} !)! \\
Product&&&!  { <fst=x,snd=y> -> x_v * y_v }!\\ \hline
\end{tabular}
\vspace{-0.4cm}
\caption{Translation from \nrcplus (positive NRC with bag semantics)~\cite{incnrc} to \lang.}
\vspace{-0.4cm}
\label{fig:nrc_ndql}
\end{figure*}

\begin{myexample}{5}
The gene burden analysis uses data from !Variants! to calculate the mutational burden for every gene within every sample. The program first iterates over the top-level of !Variants!, iterates over 
the top-level of !Genes!, then assigning a variant to a gene if the variant lies within the mapped position on the genome. The program then iterates into the nested \textbf{genotypes} information of !Variants! to return sample, gene, and burden information; here, the \textbf{call} attribute provides the count of mutated alleles in that sample. 
This expression is represented as follows in \nrcplus:

\begin{lstlisting}[language=NRC, frame=none]
for v in vcf union   for g in genes union 
    if (v.contig = g.contig && g.start <= v.start && g.end >= v.start)
    then for c in v.genotypes union 
      {sample := c.sample, gene := g.name, burden := c.call}
\end{lstlisting}

\noindent This expression is equivalent to the following \lang expression (after pushing the multiplication of multiplicities of !Variants! and !Genes! inside the inner singleton dictionary construction):

\begin{lstlisting}
sum(<v,v_v> in Variants)  sum(<g,g_v> in Genes)
    if(g.contig==v.contig&&g.start<=v.start&&g.end>=v.start) 
     then sum(<c,c_v> in v.genotypes)
        { <sample = c.sample, gene = g.name, burden = c.call> -> v_v * g_v * c_v }
\end{lstlisting}

\noindent The type of this output is !{ <sample:string, gene:string, burden:real> -> int }!.

\end{myexample}

\subsection{Aggregation}
\label{sec:agg}
An essential operator used in query processing workloads is aggregation. Both relational algebra and nested relational calculus need to be extended in order to support this operator. 
The former is extended with the group-by aggregate operator $\Gamma_{g;f}$, where $g$ specifies the set of keys that are partitioned by, and $f$ specifies the aggregation function.
\nrcagg is an extended version of the latter with support for two aggregation operators; @sumBy$_g^f$@ is similar to group-by aggregates in relational algebra, whereas @groupBy$_g$@ only performs partitioning without performing any aggregation.

Figure~\ref{fig:agg_ndql} shows the translation of aggregations in relational algebra and \nrcagg to \lang.
The explanation of these operators can be found in the supplementary materials.

\smartpara{Generalized Aggregates} Both scalar and group-by aggregate operators can be generalized to support other forms of aggregates such as minimum and maximum by supplying appropriate semi-ring structure (i.e., addition, multiplication, zero, and one). 
% In such cases, the multiplication with !r.val! should be modified appropriately. 
% For example, in the case of computing the minimum, this term should be completely removed.
For example, in the case of maximum, the maximum function is supplied as the addition operator, and the numerical addition needs to be supplied as the multiplication operator~\cite{mohri2002semiring}. An extended set of semi-rings for scalar values are presented in Table~\ref{tab:semiring_scalar}.
To compute aggregates such as average, one has to compute both summation and count using two aggregates.
The performance of this expression can be improved as discussed later in Section~\ref{sec:multiagg}.

\smartpara{Inefficiency of Group-by} The translated group-by aggregates are inefficient. 
This is because relational algebra and NRC need to have an internal implementation utilizing dictionaries for the grouping phase (i.e., the creation of the variable !tmp! in the second, fourth, fifth rules of Figure~\ref{fig:agg_ndql}).
Nevertheless, as there is no first-class support for dictionaries, the grouped structure is thrown away when the final aggregate result is produced.
This additional phase involves an additional iteration over the elements, as illustrated in the next example.

\begin{myexample}{6}
As the final step for computing gene burden, one has to perform sum-aggregate of the genotype call (now denoted burden) for each sample corresponding to that gene. By naming the previous NRC expression as \texttt{gv}, the following \nrcagg expression specifies the full burden analysis:

% \begin{tabular}{l}
% \nrckw{for} x \nrckw{in} \textbf{gmb0\_nrc} \nrckw{union} \\
% \tab\{ gene=x.gene, burdens=\nrckw{sumBy}$_\text{sid}^\text{burden}$(x.burdens) \}
% \end{tabular}

% \begin{lstlisting}[language=NRC, frame=none]
% for x in $\textbf{gmb0\_nrc}$ union 
%   {sample := x.sample, burdens := sumBy$_{sample}^{burden}$(x.burdens)}
% \end{lstlisting}

\begin{lstlisting}[language=NRC, frame=none]
let gmb = groupBy$_{sample}$(gv)
for x in gmb union 
  {sample := x.key, burdens := sumBy$_{gene}$(x.val)}
\end{lstlisting}

% An alternative (and more efficient) \nrcagg expression is:

% \begin{lstlisting}[language=NRC, frame=none]
% groupBy$_{sample}^{burdens}$(sumBy$_{sample,gene}^{burden}$(gv))
% \end{lstlisting}

\noindent This expression is translated as the following \lang expression:

\begin{lstlisting}
let tmp = sum(<x,x_v> in gv) { x.sample -> { x -> x_v } } in
let gmb = sum(<x,x_v> in tmp) { <key=x, val=x_v> -> 1 } in
sum(<x,x_v> in gmb) { <sample = x.key, burdens = 
  let tmp1 = sum(<b,b_v> in x.val) { b.gene -> x_v * b_v * b.burden } in
  sum(<t,t_v> in tmp1) { <key=t, val=t_v> -> 1 }    > -> 1 }
\end{lstlisting}

\noindent This expression is of type !{ <sample:string,burdens:{<key:string,val:real> -> int}> -> int }!. 
% However, one can write the following \lang expression:

% \begin{lstlisting}
% sum(x in gmb0)
%   { x.key.sample ->
%      sum(b in x.key.burdens)
%       { b.key.gene -> x.val * b.val * b.key.burden }
%   }
% \end{lstlisting}

% \noindent This expression is of type !{ string -> { string -> real } }!. As this expression removes one iteration over the inner dictionary, it is more efficient than the previous version. 
% Furthermore, for the cases where one needs to access the inner dictionary using !sample! values, this expression can provide constant time access using dictionary lookup, as opposed to iterating over the entire dictionary.
\end{myexample}

\begin{figure*}[t]
\begin{tabular}{|l|r c l|}
\hline
\textbf{Name} & \multicolumn{3}{|l|}{\textbf{Translation}}  \\ \hline \hline
\multicolumn{4}{|l|}{\textit{Relational Algebra:}}\\\hline
Scalar Agg.&\translate{$\Gamma_{\emptyset;f}$($e$)} &=& !sum(<x,x_v> in! \translate{$e$}!) x_v * !\translate{$f$}!(x)! \\ \hline
Group-by&\translate{$\Gamma_{g;f}$($e$)} &=& !let tmp=sum(<x,x_v> in! \translate{$e$}!) {!\translate{$g$}!(x)->x_v*!\translate{$f$}!(x)}! \\
Aggregate&&&!in sum(<x,x_v> in tmp) { <key=x, val=x_v> -> 1 }!\\\hline \hline
\multicolumn{4}{|l|}{\textit{\nrcagg:}}\\\hline
% \translate{\nrckw{sumBy}$_\emptyset^f$($e$)} 
Scalar Agg.&\translatebegin{}@sumBy$_\emptyset^f$($e$)@\translateend{}
&=& !sum(<x,x_v> in! \translate{$e$}!) x_v * !\translate{$f$}!(x)! \\\hline
% \translate{\nrckw{sumBy}$_g^f$($e$)} 
Group-by&\translatebegin{}@sumBy$_g^f$($e$)@\translateend{}
&=& !let tmp=sum(<x,x_v> in! \translate{$e$}!) { !\translate{$g$}!(x) -> x_v*!\translate{$f$}!(x) }! \\
Aggregate&&&!in sum(<x,x_v> in tmp) { <key=x, val=x_v> -> 1 }!\\\hline
% \translate{\nrckw{groupBy}$_g$($e$)} 
Nest&\translatebegin{}@groupBy$_g$($e$)@\translateend{}
&=& !let tmp=sum(<x,x_v> in! \translate{$e$}!) { !\translate{$g$}!(x) -> {x -> x_v} }! \\
&&&!in sum(<x,x_v> in tmp) { <key=x, val=x_v> -> 1 }!\\\hline
\end{tabular}
\vspace{-0.4cm}
\caption{Translation of aggregate operators of relational algebra and \nrcagg~\cite{trance_vldb} to \lang.}
\vspace{-0.4cm}
\label{fig:agg_ndql}
\end{figure*}

\section{Expressiveness for Linear Algebra}
\label{sec:la}
In this section, we show the power of \lang for expressing linear algebra workloads. 
We first show the representation of vectors in \lang, followed by the representation of matrices in \lang. 
We also show the translation of linear algebra operators to \lang expressions together with their Einstein summation notation, referred to as \code{einsum} in libraries such as numpy.

\begin{figure*}[t]
\setlength\tabcolsep{0.05cm}
\begin{tabular}{|l|r c l|c|}
\hline
\textbf{Name} & \multicolumn{3}{|l|}{\textbf{Translation}} & \textbf{Einsum} \\ \hline \hline
\multicolumn{5}{|l|}{\textit{Vector Operations:}}\\\hline
Addition &\translate{$V_1 + V_2$} &=& \translate{$V_1$} !+!  \translate{$V_2$}& - \\ \hline 
Scal-Vec. Mul.&\translate{$a \cdot V$} &=& \translate{$a$} !*! \translate{$V$}& !,i->i! \\ \hline 
Hadamard Prod.&\translate{$V_1 \circ V_2$} &=& !sum(x in! \translate{$V_1$} !) { x.key->x.val*!\translate{$V_2$}!(x.key) }! & !i,i->i! \\ \hline
Dot Prod.&\translate{$V_1 \cdot V_2$} &=& !sum(x in! \translate{$V_1$} !) x.val * !\translate{$V_2$}!(x.key)!& !i,i->! \\ \hline 
Summation&\translate{$\sum_{a\in V}a$} &=& !sum(x in! \translate{$V$} !) x.val!&!i->! \\ \hline \hline
\multicolumn{5}{|l|}{\textit{Matrix Operations:}}\\\hline
Transpose&\translate{$M^T$} &=& !sum(x in! \translate{$M$} !)! & !ij->ji!\\
&&&\tab!{ <row=x.key.col, col=x.key.row> -> x.val }!&\\ \hline
Addition&\translate{$M_1 + M_2$} &=& \translate{$M_1$} !+!  \translate{$M_2$} & -\\ \hline
Scal-Mat. Mul.&\translate{$a \cdot M$} &=& \translate{$a$} !*!  \translate{$M$} & !,ij->ij!\\ \hline
Hadamard Prod.&\translate{$M_1 \circ M_2$} &=& !sum(x in! \translate{$M_1$} !) { x.key -> x.val * !\translate{$M_2$}!(x.key) }! & !ij,ij->ij!\\ \hline
Matrix-Matrix&\translate{$M_1 \times M_2$} &=& !sum(x in! \translate{$M_1$} !) sum(y in! \translate{$M_2$} !)! & !ij,jk->ik! \\
Multiplication&&&~~!if(x.key.col == y.key.row) then ! &\\ 
&&&~~~~~~!{ <row=x.key.row,col=y.key.col> -> x.val*y.val }! &\\ \hline
Mat-Vec. Mul.&\translate{$M \cdot V$} &=& !sum(x in! \translate{$M$}!) {x.key.row->x.val*!\translate{$V$}!(x.key.col)}! & !ij,j->i!\\ \hline
Trace&\translate{$Tr(M)$} &=& !sum(<k,v> in! \translate{$M$}!) if(k.row==k.col) then v!& !ii->! \\ \hline
\end{tabular}
\vspace{-0.4cm}
\caption{Translation of linear algebra operations to \lang.}
\label{fig:la_ndql}
\vspace{-0.4cm}
\end{figure*}

\subsection{Vectors}
\lang represents vectors as dictionaries mapping indices to the element values; thus, vectors with elements of type !S! are \lang expressions of type !{ int -> S }!. 
This representation is similar to functional pull arrays in array processing languages~\cite{repa}.
The key difference is that the size of the array is not stored separately. 
% Figure~\ref{fig:vec_ndql} shows the translation of vector operations to \lang. The detailed explanation for each construct can be found in the supplementary material.

\begin{myexample}{7}
Consider two vectors defined as  
$V=\begin{bmatrix}
a_{0}&0&a_{1}&a_{2}\\
\end{bmatrix}$ and $U=\begin{bmatrix}
b_0&b_1&b_2&0\\
\end{bmatrix}$.
These vectors are represented in \lang as
!{ 0 -> $a_{0}$, 2 -> $a_{1}$, 3 -> $a_{2}$ }! and !{ 0 -> $b_{0}$, 1 -> $b_{1}$, 2 -> $b_{2}$ }!.
\noindent The expression $V\circ U$ is evaluated to
!{ 0 -> $a_{0}$*$b_0$, 2 -> $a_{1}$*$b_2$, 3 -> $a_{2}$*0 }!.
As the value associated with the key !3! is zero, this dictionary is equivalent to 
!{ 0 -> $a_{0}$*$b_0$, 2 -> $a_{1}$*$b_2$ }!.
This value corresponds to the result of evaluating $V\circ U$, that is the vector
$\begin{bmatrix}
a_0b_0&0&a_1b_2&0\\
\end{bmatrix}$.
% \[
% \begin{bmatrix}
% a_0b_0&0&a_1b_2&0\\
% \end{bmatrix}
% \]
\end{myexample}

\subsection{Matrices}
\label{sec:matrix}
Matrices are considered as dictionaries mapping the row and column indices to the element value. 
This means that matrices with elements of type !S! are \lang expressions with the type !{ <row: int, col: int> -> S }!.
Figure~\ref{fig:la_ndql} shows the translation of vector and matrix operations to \lang. 
We give a detailed explanation of these operators in the supplementary material.

\begin{myexample}{8}
Consider the following matrix $M$ of size $2 \times 4$:
$\begin{bmatrix}
c_0&0&0&c_1\\
0&c_2&0&0\\
\end{bmatrix}$.
This matrix is in \lang as 
!{<row=0,col=0> -> $c_0$,<row=0,col=3> -> $c_1$,<row=1,col=1> -> $c_2$}!.
The expression $M\cdot V$ is evaluated to the following dictionary after translating to \lang:
!{ 0 -> $c_0$*$a_0$+$c_1$*$a_2$, 1 -> $c_2$*0 }!.
This expression is the dictionary representation of the following vector, which is the result of the matrix-vector multiplication:
$\begin{bmatrix}
c_0a_0+c_1a_2&0\\
\end{bmatrix}$.
% \[
% \begin{bmatrix}
% c_0a_0+c_1a_2&0\\
% \end{bmatrix}
% \]
\end{myexample}

\begin{myexample}{9}
Computing the covariance matrix is an essential technique in machine learning, and is useful for training various models~\cite{10.1145/3196959.3196960}. The covariance matrix of a matrix $A$ is computed as $A^TA$. In our biomedical example, computing the covariance matrix enables us to train different machine learning models such as linear regression on top of the !Variant! dataset. 

\end{myexample}

\smartpara{Point-wise Operations} In many machine learning applications, it is necessary to support point-wise application of functions such as $cos$, $sin$, and $tan$ on matrices. 
\lang can easily support these operators by adding the corresponding scalar functions and using !sum! to apply them at each point.

\smartpara{Inefficiency of Operators} Note that the presented operators are highly inefficient. For example, matrix-matrix multiplication requires iterating over every combination of elements, whereas with a more efficient representation, this can be significantly improved. 
This improved representation is shown later in Section~\ref{sec:curried}.

\section{Efficiency}
\label{sec:opt}
In this section, we present loop optimizations of \lang.
Figure~\ref{fig:opt_rules} summarizes the transformation rules required for such optimizations.

\begin{figure}
\begin{tabular}{|l c l|}
\hline
\textit{Vertical Loop Fusion:} & &\\ \hline
% !let y=sum(x in e1) {f1(x.key)->x.val}! & \multirow{2}{*}{\transto} & !sum(x in e1)! \\
% !sum(x in y) f3(x)! & & !  f3(<key = f1(x.key), val = x.val>)! \\ \hline \hline
!let y=sum(<x,x_v> in e1){f1(x)->x_v}! & \multirow{2}{*}{\transto} & !sum(<x,x_v> in e1)! \\
!in sum(<x,x_v> in y){f2(x)->x_v}! & & !  { f2(f1(x)) -> x_v }!\\ \hline
!let y=sum(<x,x_v> in e1){x->f1(x_v)}! & \multirow{2}{*}{\transto} & !sum(<x,x_v> in e1)! \\
!in sum(<x,x_v> in y){x->f2(x_v)}! & & !  { x -> f2(f1(x_v)) }!\\
\hline \hline
\textit{Horizontal Loop Fusion:} & &\\ \hline
!let y1=sum(x in e1) f1(x) in! & \multirow{3}{*}{\transto} & !let tmp = sum(x in e1)! \\
!let y2=sum(x in e1) f2(x) in! & & !  <y1 = f1(x), y2 = f2(x) >! \\ 
!f3(y1, y2)! & & !in f3(tmp.y1, tmp.y2)! \\ \hline \hline
\textit{Loop Factorization:} & &\\ \hline
!sum(x in e1) e2 * f(x)! & \transto & !e2 * sum(x in e1) f(x)! \\ \hline
!sum(x in e1) f(x) * e2! & \transto & !(sum(x in e1) f(x)) * e2! \\ \hline \hline
\textit{Loop-Invariant Code Motion:} & &\\ \hline
!sum(x in e1) let y = e2 in f(x, y)! & \multirow{1}{*}{\transto} & !let y = e2 in sum(x in e1) f(x, y)! \\ \hline \hline
\textit{Loop Memoization:} & &\\ \hline
!sum(x in e1)! & \multirow{2}{*}{\transto} & !let tmp=sum(x in e1){p(x)->{x.key->x.val}}! \\
!  if(p(x) == e2) then g(x, e3)! & & !in sum(x in tmp(e2)) g(x, e3)!\\ \hline
!sum(x in e1)! & \multirow{2}{*}{\transto} & !let tmp=sum(x in e1) {p(x)->f(x)}! \\
!  if(p(x) == e2) then f(x)! & & !in tmp(e2)!\\ \hline
\end{tabular}
\vspace{-0.4cm}
\caption{Transformation rules for loop optimizations.}
\vspace{-0.4cm}
\label{fig:opt_rules}
\end{figure}

\subsection{Loop Fusion}
\label{sec:loopfusion}

% What is the general rule for vertical loop fusion?
% does this work?
% let y = sum(x in e1) {{ f1(x.key) -> f2(x.val) }} in sum(x1 in y) f3(x1.key, x1.val) 
% --->
% sum(x in e1) f3(f1(x.key), f2(x.val))

\subsubsection{Vertical Loop Fusion} One of the essential optimizations for collection programs is deforestation~\cite{deforestation,foldr-fusion-1,Svenningsson:2002:SFA:581478.581491,Coutts07streamfusion}. 
This optimization can remove an unnecessary intermediate collection in a vertical pipeline of operators, and is thus named as vertical loop fusion. The benefits of this optimization are manifold. The memory usage is improved thanks to the removal of intermediate memory, and the run time is improved because the removal of the corresponding loop.
In query processing engines, pull and push-based \textit{pipelining}~\cite{Neumann11,cowbook} has the same role as vertical loop fusion~\cite{jfppushpull}.
Similarly, in functional array processing languages, pull arrays and push arrays~\cite{edsl-push,Claessen:2012:EAC:2103736.2103740,Svensson:2014:DPA:2636228.2636231} are responsible for fusion of arrays.
However, none of the existing approaches support fusion for dictionaries. 
Next, we show how vertical fusion in \lang subsumes the existing techniques. 

\begin{figure*}[t]
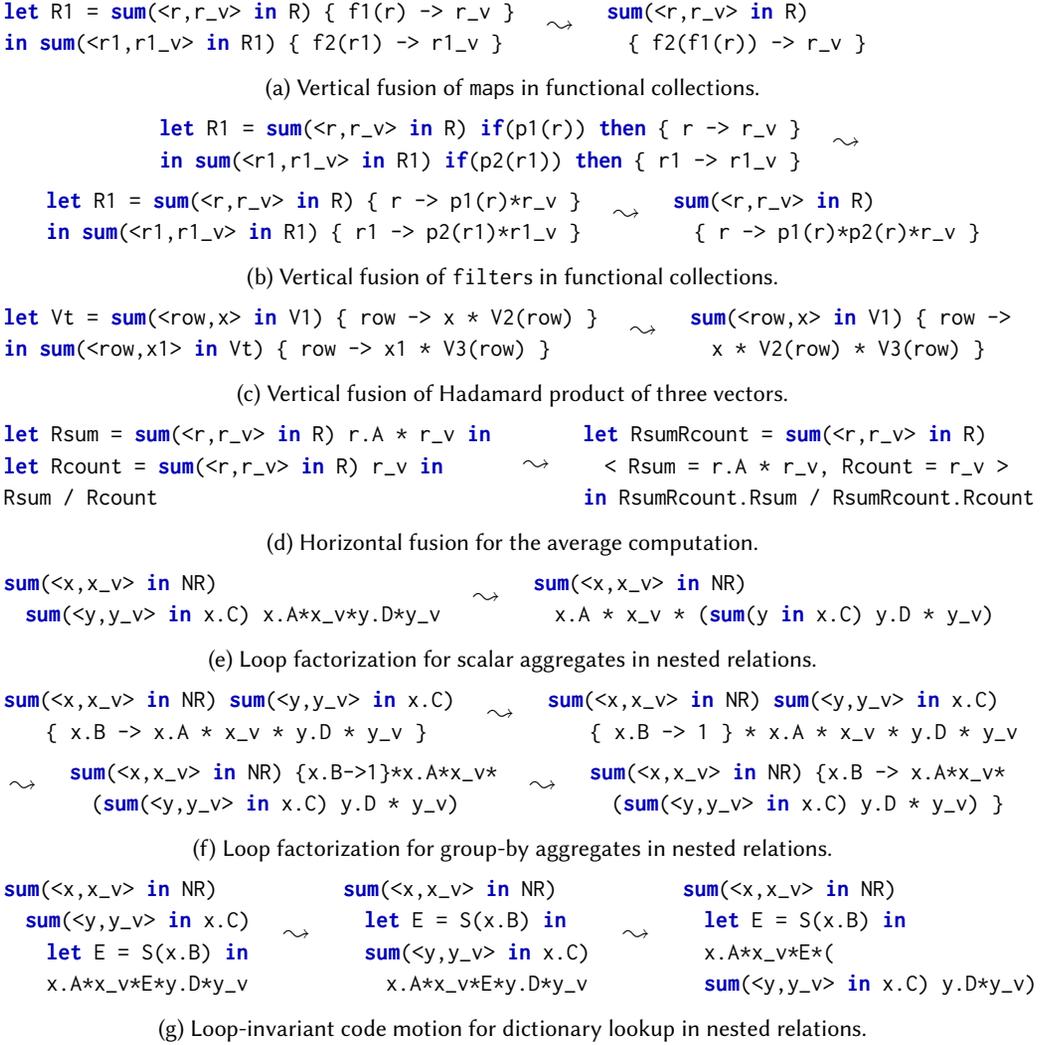

\begin{center}
\begin{subfigure}{\textwidth}
\begin{tabular}{l c l}
\begin{lstlisting}[frame=none]
let R1 = sum(<r,r_v> in R) { f1(r) -> r_v }
in sum(<r1,r1_v> in R1) { f2(r1) -> r1_v }
\end{lstlisting}&\transto&
\begin{lstlisting}[frame=none]
sum(<r,r_v> in R) 
  { f2(f1(r)) -> r_v }
\end{lstlisting}
\end{tabular}
\caption{Vertical fusion of \texttt{map}s in functional collections.}
\label{fig:opt:ex:map}
\end{subfigure}
\begin{subfigure}{\textwidth}
\centering
\begin{tabular}{l c}
\begin{lstlisting}[frame=none]
let R1 = sum(<r,r_v> in R) if(p1(r)) then { r -> r_v }  
in sum(<r1,r1_v> in R1)       if(p2(r1)) then { r1 -> r1_v } 
\end{lstlisting} & \transto
\end{tabular} \\
\begin{tabular}{l c l}
\begin{lstlisting}[frame=none]
let R1 = sum(<r,r_v> in R) { r -> p1(r)*r_v } 
in sum(<r1,r1_v> in R1)       { r1 -> p2(r1)*r1_v }
\end{lstlisting}&\transto&
\begin{lstlisting}[frame=none]
sum(<r,r_v> in R) 
  { r -> p1(r)*p2(r)*r_v }
\end{lstlisting}
\end{tabular}
\caption{Vertical fusion of \texttt{filter}s in functional collections.}
\label{fig:opt:ex:filter}
\end{subfigure}
\begin{subfigure}{\textwidth}
\begin{tabular}{l c l}
\begin{lstlisting}[frame=none]
let Vt = sum(<row,x> in V1) { row -> x * V2(row) } 
in sum(<row,x1> in Vt) { row -> x1 * V3(row) }
\end{lstlisting}&\transto&
\begin{lstlisting}[frame=none]
sum(<row,x> in V1) { row -> 
  x * V2(row) * V3(row) }
\end{lstlisting}
\end{tabular}
\caption{Vertical fusion of Hadamard product of three vectors.}
\label{fig:opt:ex:hadamard}
\end{subfigure}
\begin{subfigure}{\textwidth}
\begin{tabular}{l c l}
\begin{lstlisting}[frame=none]
let Rsum   = sum(<r,r_v> in R) r.A * r_v in
let Rcount = sum(<r,r_v> in R) r_v       in
Rsum / Rcount 
\end{lstlisting}&\transto&
\begin{lstlisting}[frame=none]
let RsumRcount = sum(<r,r_v> in R) 
  < Rsum = r.A * r_v, Rcount = r_v >
in RsumRcount.Rsum / RsumRcount.Rcount 
\end{lstlisting}
\end{tabular}
\caption{Horizontal fusion for the average computation.}
\label{fig:opt:ex:avg}
\end{subfigure}
\begin{subfigure}{\textwidth}
\begin{tabular}{l c l}
\begin{lstlisting}[frame=none]
sum(<x,x_v> in NR)  
  sum(<y,y_v> in x.C) x.A*x_v*y.D*y_v
\end{lstlisting}&\transto&
\begin{lstlisting}[frame=none]
sum(<x,x_v> in NR) 
  x.A * x_v * (sum(y in x.C) y.D * y_v)
\end{lstlisting}
\end{tabular}
\caption{Loop factorization for scalar aggregates in nested relations.}
\label{fig:opt:ex:fact}
\end{subfigure}
\begin{subfigure}{\textwidth}
\begin{tabular}{l c l}
\begin{lstlisting}[frame=none]
sum(<x,x_v> in NR) sum(<y,y_v> in x.C) 
    { x.B -> x.A * x_v * y.D * y_v }
\end{lstlisting}&\transto&
\begin{lstlisting}[frame=none]
sum(<x,x_v> in NR) sum(<y,y_v> in x.C) 
    { x.B -> 1 } * x.A * x_v * y.D * y_v
\end{lstlisting}
\end{tabular}
\begin{tabular}{c l c l}
\transto&
\begin{lstlisting}[frame=none]
sum(<x,x_v> in NR) {x.B->1}*x.A*x_v* 
  (sum(<y,y_v> in x.C) y.D * y_v)
\end{lstlisting}&\transto&
\begin{lstlisting}[frame=none]
sum(<x,x_v> in NR) {x.B -> x.A*x_v* 
  (sum(<y,y_v> in x.C) y.D * y_v) } 
\end{lstlisting}
\end{tabular}
\caption{Loop factorization for group-by aggregates in nested relations.}
\label{fig:opt:ex:fact2}
\end{subfigure}
\begin{subfigure}{\textwidth}
\begin{tabular}{l c l c l}
\begin{lstlisting}[frame=none]
sum(<x,x_v> in NR) 
  sum(<y,y_v> in x.C) 
    let E = S(x.B) in 
    x.A*x_v*E*y.D*y_v
\end{lstlisting}&\transto&
\begin{lstlisting}[frame=none]
sum(<x,x_v> in NR) 
  let E = S(x.B) in 
  sum(<y,y_v> in x.C) 
    x.A*x_v*E*y.D*y_v
\end{lstlisting}&\transto&
\begin{lstlisting}[frame=none]
sum(<x,x_v> in NR) 
  let E = S(x.B) in 
  x.A*x_v*E*(
  sum(<y,y_v> in x.C) y.D*y_v)
\end{lstlisting}
\end{tabular}
\caption{Loop-invariant code motion for dictionary lookup in nested relations.}
\label{fig:opt:ex:motion}
\end{subfigure}
\end{center}
\vspace{-0.3cm}
\caption{Examples for loop fusion (vertical and horizontal) and loop hoisting in \lang.}
\label{fig:opt:ex}
\vspace{-0.3cm}
\end{figure*}

\smartpara{Fusion in Functional Collections}
As a classic example in functional programming, a sequence of two !map! operators can be na\"ively expressed as the left expression in Figure~\ref{fig:opt:ex:map}. There is no need to materialize the results of the first !map! into !R1!.
Instead, by applying the first vertical loop fusion rule from Figure~\ref{fig:opt_rules} one can fuse these two operators and remove the intermediate collection as depicted in the right expression of Figure~\ref{fig:opt:ex:map}.
Another interesting example is the fusion of two !filter! operators. The pipeline of these operators is expressed as the first \lang expression in Figure~\ref{fig:opt:ex:filter}. The conditional construct in both summations can be pushed to the value of dictionary resulting in the second expressions. Finally, by applying the second rule of vertical fusion, the last expression is derived, which uses a single iteration over the elements of !R!, and the result collection has a zero multiplicity for elements where !p1! or !p2! is !false!. 

\smartpara{Fusion in Linear Algebra}
Similarly, in linear algebra programs there are cases where the materialization of intermediate vectors can be avoided. 
As an example, consider the Hadamard product of three vectors, which is na\"ively translated as the first \lang expression in Figure~\ref{fig:opt:ex:hadamard}.
Again, the intermediate vector !Vt! is not necessary. By applying the second vertical loop fusion rule from Figure~\ref{fig:opt_rules}, one can avoid the materialization of !Vt!, as shown in the right expression in Figure~\ref{fig:opt:ex:hadamard}. This expression performs a single iteration over the elements of the vector !V1!.

\subsubsection{Horizontal Loop Fusion}
\label{sec:multiagg} 
Another form of loop fusion involves simultaneous iterations over the same collection, referred to as horizontal loop fusion. 
More specifically, in query processing workloads, there could be several aggregate computations  over the same relation. In such cases,
one can share the scan over the same relation and compute all the aggregates simultaneously.
For example, in order to compute the average, one can use the following two aggregates over the same relation !R!, as shown in the left expression in Figure~\ref{fig:opt:ex:avg}.
In such a case, one can iterate over the input relation only once, and 
compute both aggregates as a tuple. In this optimized expression (cf. right expression in Figure~\ref{fig:opt:ex:avg}), the average is computed by dividing the element of the tuple storing summation over the count. This optimization corresponds to \textit{merging a batch of aggregates} over the same relation in databases.

\subsection{Loop Hoisting}
\subsubsection{Loop Factorization}
\noindent One of the most important algebraic properties of the semi-ring structure is the distributive law, which enables factoring out a common factor in addition of two expressions. 
This algebraic law can be generalized to the case of summation over a collection (cf. Figure~\ref{fig:opt_rules}).

Consider a nested relation !NR! with type !{<A:real,B:int,C:{<D:real> -> int}> -> int}! where we are interested in computing the multiplication of the attributes !A! and !D!. This can be represented as the left expression in Figure~\ref{fig:opt:ex:fact}.
The subexpression !x.A*x_v! is independent of the inner loop, and can be factored out, resulting in the right expression in the same figure.

This optimization can also benefit expressions involving dictionary construction, such as group by expressions. As an example, consider the same aggregation as before grouped by attribute !B!, represented in the first expression of Figure~\ref{fig:opt:ex:fact2}. According to the semantics of \lang (cf. Section~\ref{sec:sem}), we can rewrite the dictionary construction resulting in the second expression.
Again, we can factor out the terms independent of the inner loop (cf. the third expression). By using the semantics of dictionaries, this expression can be translated to the last expression in Figure~\ref{fig:opt:ex:fact2}.
In this expression the intermediate dictionaries corresponding to each group are only constructed for each element of the outer relation, instead of each element of the inner relation. 

\subsubsection{Loop-Invariant Code Motion}
In addition to multiplication operands, one can hoist let-bindings invariant to the loop. Consider the following example, where one computes the aggregate !A * E * D! where !E! comes from looking up (using hash join) for another relation !S!, represented as the first expression in Figure~\ref{fig:opt:ex:motion}.
In this case, the computation of !E! of is independent of the inner loop and thus can be hoisted outside following the last rule of Figure~\ref{fig:opt_rules}, resulting in the middle expression.
Additionally, this optimization enables further loop factorization, which results in the last expression in Figure~\ref{fig:opt:ex:motion}.

\subsection{Loop Memoization}
\label{sec:loopmem}
In many cases, the body of loops cannot be easily hoisted. 
Such cases require further memoization-based transformations on the loop body to make them independent of the loop variable, referred to as loop memoization. 

\begin{figure*}[t]
\begin{center}
\begin{subfigure}{\textwidth}
\setlength{\tabcolsep}{0em}
\begin{tabular}{l c l c l}
\begin{lstlisting}[frame=none]
sum(<r,r_v> in R) 
 sum(<s,s_v> in S) 
  if(jkR(r)==jkS(s)) then 
  { concat(r,s)->r_v*s_v }
\end{lstlisting}&\transto&
\begin{lstlisting}[frame=none]
sum(<r,r_v> in R)
 let Sp = sum(<s,s_v> in S) 
  { jkS(s) -> {s->s_v} } in
 sum(<s,s_v> in Sp(jkR(r))) 
  { concat(r,s)->r_v*s_v }
\end{lstlisting}&\transto&
\begin{lstlisting}[frame=none]
let Sp = sum(<s,s_v> in S) 
 { jkS(s) -> {s->s_v} } in
sum(<r,r_v> in R)
 sum(<s,s_v> in Sp(jkR(r))) 
  { concat(r,s)->r_v*s_v }
\end{lstlisting}
\end{tabular}
\caption{Synthesizing hash join operator from nested loop join.}
\label{fig:opt:synth:hashjoin}
\end{subfigure}
\begin{subfigure}{\textwidth}
\setlength{\tabcolsep}{0em}
\begin{tabular}{l c l c l}
\begin{lstlisting}[frame=none]
sum(<r,r_v> in R) 
 sum(<s,s_v> in S) 
  if(jkR(r)==jkS(s)) then 
  { jkR(r)->f(r)*g(s) }
\end{lstlisting}&\transto&
\begin{lstlisting}[frame=none]
sum(<r,r_v> in R) { jkR(r)->
 f(r)*(sum(<s,s_v> in S) 
  if(jkR(r)==jkS(s)) then 
    g(s) ) }
\end{lstlisting}&\transto&
\begin{lstlisting}[frame=none]
let Sp = sum(<s,s_v> in S) 
 { jkS(s) -> g(s) } in
sum(<r,r_v> in R)
 { jkR(r)->f(r)*Sp(jkR(r)) }
\end{lstlisting}
\end{tabular}
\caption{Synthesizing groupjoin operator from nested loop join and group-by aggregation.}
\label{fig:opt:synth:groupjoin}
\end{subfigure}
\end{center}
\vspace{-0.3cm}
\caption{Synthesizing hash join and groupjoin operators by loop memoization.}
\label{fig:opt:synth}
\vspace{-0.5cm}
\end{figure*}

\subsubsection{Synthesizing Hash Join}
\label{sec:hashjoin}
In general, we can produce a nested dictionary by memoizing the inner loop. 
Then, instead of iterating the entire range of inner loop, only iterate over its relevant partition.
Consider again the case of equality join between two relations !R! and !S! (cf. Section~\ref{sec:ra}) based on the join keys !jkR(r)! and !jkS(s)!, represented as the first expression in Figure~\ref{fig:opt:synth:hashjoin}.
This expression is inefficient, due to iterating 
over every combination of the elements of the two input relations. 
The body of the conditional is however dependent on the outer loop and thus cannot be hoisted outside. 
Applying the first loop memoization rule results in the middle expression; in order to join the two relations, it is sufficient to iterate over relation !R! and find the corresponding partition from relation !S! by using !Sp(jkR(r))!.
In this expression, the dictionary !Sp! is no longer dependent on !r!. Thus, we can perform loop-invariant code motion, which results in the last expression.

In the specific case of implementing a dictionary using
a hash-table, this join algorithm corresponds to a hash join operator;
The first loop corresponds to the \textit{build phase} and the second
loop corresponds to the \textit{probe phase}~\cite{cowbook}. 
This expression is basically the same expression as the one for the hash join operator. 
This means that the first rewrite rule of loop memoization when combined with loop hoisting \textit{synthesizes hash join operator}.

\begin{myexample}{5 (Cont.)}
Let us consider again the join between !Gene! and !Variants!.
The previous expression used nested loops in order to handle join, which is inefficient. The following expression uses hash join instead:

% \todo{add grouping by sample}
\begin{lstlisting}
let Vp = sum(<v,v_v> in Variants) 
  { v.contig -> {<start=v.start,genotypes=v.genotypes> -> v_v} } in
sum(<g,g_v> in Genes) sum(<v,v_v> in Vp(g.contig)) sum(<m,m_v> in v.genotypes)
  if(g.start<=v.start&&g.end>=v.start) then 
    { <sample=m.sample,gene=m.gene,burden=m.call> -> g_v*v_v*m_v }
\end{lstlisting}
\end{myexample}

\subsubsection{Synthesizing Groupjoin}

There are special cases, where the loop memoization can perform even better.
This achieved by performing a portion of computation while partitioning the data.
This situation arises when computing an aggregation over the
result of join between two relations.
As an example, consider the 
summation of !f(r) * g(s)! on the elements !r! and !s! that successfully join, grouped by the join key, represented as the last expression of Figure~\ref{fig:opt:synth:groupjoin}.
In this case, the inner !sum! contains the terms !f(r)! and !jkR(r)! which are dependent on !r! and thus makes it impossible to be hoisted. The terms !jkR(r)! and !f(r)! inside the conditional body can be factored outside using the loop factorization rule, resulting in the middle expression. Afterwards, by applying the second rule of loop memoization, the dictionary bound to variable !Sp! is constructed. 
As this dictionary is no longer dependent on !r!, we can apply loop-invariant code motion, resulting in the last expression.

In fact, the result expression corresponds to the implementation of a groupjoin operator~\cite{groupjoin}.
In essence, the loop memoization and loop hoisting optimizations have the effect of \textit{pushing aggregations past joins}~\cite{groupbybeforejoin}.

\subsubsection{Memoization Beyond Databases}

In the case of using max-product semi-ring (cf. Figure~\ref{tab:semiring_scalar}) these optimization can \textit{synthesize variable elimination} for maximum a priority (MAP) inference in Bayesian networks~\cite{abo2016faq,aji2000generalized}. 
Furthermore, \textit{loop normalization}~\cite{shaikhha2019efficient} can also be thought of as a special case of this rule.

\begin{figure}
\begin{small}
\begin{tabular}{|l|c|c|c|c|c|}
\hline
\multirow{2}{*}{\diagbox[width=3.5cm]{Optimization}{Feature}}
& Purely & Dictionary & Dictionary & \multirow{2}{*}{Semi-ring} & \multirow{2}{*}{Compositional} \\
 &  functional &  lookup &  summation &  &  \\
\hline \hline
Vertical loop fusion & \cmark & \cmark & \cmark & & \\ \hline
Horizontal loop fusion & \cmark & & \cmark & & \\ \hline
Memoization & \cmark & \cmark & \cmark & & \\ \hline
Loop factorization & \cmark &  & \cmark & \cmark & \\ \hline
Code motion & \cmark & & \cmark & & \\ \hline
Data layouts &  &  &  & & \cmark \\ \hline
\end{tabular}
\end{small}
\vspace{-0.4cm}
\caption{\revision{The features of \lang leveraged by each transformation.}}
\vspace{-0.4cm}
\label{fig:opt_category}
\end{figure}

\revision{
\subsection{Putting all Together}
In this section, we investigate the design decisions behind \lang that enables the optimizations presented before.
The features of \lang can be categorized as follows:
\begin{itemize}[leftmargin=*]
\item \textbf{Purely functional:} \lang does not allow any mutation and global side effect.
\item \textbf{Dictionary lookup:} the dictionaries support a constant-time look up operation.
\item \textbf{Dictionary summation:} iteration over dictionaries allows for both scalar aggregates and dictionary construction in the style of monoid comprehensions~\cite{monoid-comprehension}.
\item \textbf{Semi-ring:} \lang has constructs with such structure including semi-ring dictionaries.
\item \textbf{Compositional:} semi-ring dictionaries accept semi-ring dictionaries as both keys and values.
\end{itemize}
Figure~\ref{fig:opt_category} shows the features that are leveraged by each loop optimization.
The compositionality feature is essential for expressing various data layout representations, which is presented next.
}

\section{Data Layout Representations}
\label{sec:datalayout}
In this section, we investigate various data representations supported by \lang, and show their correspondence to existing data formats used in query engines and linear algebra frameworks.

\subsection{Flat vs. Curried Representation}
\label{sec:curried}
Currying a function of type !T1$\times$T2 => T3! results in a function of type !T1 => (T2 => T3)!.
Similarly, dictionaries with a pair key can be curried into a nested dictionary. 
More specifically, a dictionary of type !{ <a: T1, b: T2> -> T3 }! can be
curried into a dictionary of type !{ T1 -> { T2 -> T3 } }!.

\subsubsection{Factorized Relations}
Relations can be curried following a specified order for their attributes. 
In the database community, this representation is referred to as \textit{factorized representation}~\cite{fdb} using a \textit{variable order}. 
In practice, a trie data structure can be used for factorized representation, and has proved useful for computational complexity improvements for joins, resulting into a class of join algorithms referred to as worst-case optimal joins~\cite{leapfrog}.

Consider a relation $R(a_1, ..., a_n)$ (with bag semantics), the representation of which is a dictionary of type !{ <$a_1$:$A_1$,...,$a_n$:$A_n$> -> int }! in \lang.
By using the variable order of $[a_1, ..., a_n]$, the factorized representation of this relation in \lang is a nested dictionary of type !{$A_1$->{$...$->{$A_n$->int}$...$}}!.

% \begin{myexample}{1 (Cont.)} Consider relation R, which was presented in Example 1. The factorized representation of this relation in \lang is as follows:

% !{{ $a_1$ -> {{ $b_1$ -> 1, $b_2$ -> 1 }}, $a_2$ -> {{ $b_3$ -> 1 }} }}!

% \end{myexample}

\subsubsection{Curried Matrices}
Matrices can also be curried as a dictionary with row as key, and another dictionary as value. The inner dictionary has column as key, and the element as value. 
Thus, a curried matrix with elements of type !S! is an \lang expression of type !{ int -> { int -> S } }!.

\begin{myexample}{8 (Cont.)}
Consider matrix $M$ from Example 8. The curried representation of this matrix in \lang is 
!{ 0 -> { 0 -> $c_0$, 3 -> $c_1$ }, 1 -> { 1 -> $c_2$ } }!.
\end{myexample}

\noindent The flat encoding of matrices presented in Section~\ref{sec:matrix} results in inefficient implementation for various matrix operations, as explained before. 
By using a curried representation instead, one can provide more efficient implementations
for matrix operations.
% (cf. Figure~\ref{fig:cla_ndql}).

As an example, Figure~\ref{fig:cla_ndql} shows the translation of curried matrix-matrix multiplication. 
% \smartpara{Matrix-Matrix Multiplication} 
Instead of iterating over every combination of elements of two matrices, 
the curried representation allows a direct lookup on the elements of a particular row of the second matrix.
Assuming that the dimension of the first matrix is $m\times n$, and the second matrix is of dimension $n\times k$, this improvement reduces the complexity from $O(mn^2k)$ down to $O(mnk)$.

% \smartpara{Matrix Trace} The trace of a matrix can be computed by iterating over the rows of a matrix and a constant look up in the inner dictionary, instead of iterating over all elements of the matrix. 
% Thus, the curried representation converts a quadratic computation to a linear one.

\begin{myexample}{9 (Cont.)}
The computation of the covariance by curried matrices can be optimized as:

\begin{lstlisting}
let At = sum(row in A) sum(x in row.val) { x.key -> {row.key -> x.val } } in
sum(row in At){ row.key -> sum(x in row.val) sum(y in A(x.key)){y.key->x.val*y.val} }
\end{lstlisting}

\noindent Furthermore, performing vertical loop fusion results in the following optimized program:

\begin{lstlisting}
sum(row in A) sum(x in row.val) { x.key -> sum(y in row.val){y.key->x.val*y.val} }
\end{lstlisting}
\end{myexample}

\begin{figure*}[t]
\begin{tabular}{|r c l|}
\hline
\translate{$M_1 \times M_2$} &=& !sum(row in! \translate{$M_1$} !) { row.key -> ! \\
&&\tab!sum(x in row.val) sum(y in! \translate{$M_2$}!(x.key)) { y.key -> x.val * y.val } }!\\ 
\hline
\end{tabular}
\vspace{-0.4cm}
% \caption{Translation of curried matrix operations to \lang.}
\caption{Translation of matrix-matrix multiplication for curried matrices to \lang.}
\vspace{-0.2cm}
\label{fig:cla_ndql}
\end{figure*}

\smartpara{Correspondence to Tensor Formats} The flat representation corresponds to the COO format of sparse tensors, whereas the curried one corresponds to CSF using hash tables~\cite{taco_format}.

\subsection{Sparse vs. Dense Layouts}

\subsubsection{Sparse Layout}
So far, all collections were encoded as dictionaries with hash table as their underlying implementations. 
This representation is appropriate for sparse structures, but it is suboptimal for dense ones;
typically linear algebra frameworks use arrays to store dense tensors.

\subsubsection{Dense Layout} 
\lang can leverage !dense_int! type in order to use array for implementing collections.
As explained in Section~\ref{sec:lang}, arrays are the special case of dictionaries with !dense_int! keys.
The runtime environment of \lang uses native array implementations for such dictionaries instead of hash-table data-structures.
Thus, by using !dense_int! as the index for tensors, \lang can have a more efficient layout for 
dense vectors and matrices.
In this way, a vector is encoded as an array of elements and a matrix as a nested array of elements.

Next, we see how dense layout and in particular arrays can be used to implement row and columnar layout for query engines.

\subsection{Row vs. Columnar Layouts}

\subsubsection{Row Layout}
In cases where input relations do not have duplicates, there is no need to keep the boolean multiplicity information in the corresponding dictionaries.
Instead, relations can be stored as dictionaries where the key is an index, 
and the value is the corresponding row.
This means that the relation $R(a_1, ..., a_n)$ can be represented as a dictionary of type !{ idx_type -> {$a_1$: $A_1$, ..., $a_n$: $A_n$} }!.
The key (of type !idx_type!) can be an arbitrary \textit{candidate key}, as it can uniquely specify a row.
By using !dense_int! type as the key of this dictionary, the keys are consecutive integer values starting from zero; thus, we encode relations using an array representation.
This means that the previously mentioned relation becomes an array of type ![|<$a_1$: $A_1$, ..., $a_n$: $A_n$>|]!.

\subsubsection{Columnar Layout}
Column store~\cite{idreos2012monetdb} databases represent relations using vertical fragmentation.
Instead of storing all fields of a record together as in row layout, 
columnar layout representation stores the values of each field in separate collections.

In \lang, columnar layout is encoded as a record where each field stores the array of its values.
This representation corresponds to the array of struct representation that is used in many high performance computing applications.
Generally, the columnar layout representation of the relation $R(a_1, ..., a_n)$ is encoded as a record of type !<$a_1$: [|$A_1$|], ..., $a_n$: [|$A_n$|]>! in \lang.

\begin{figure}
    \centering
    \begin{tabular}{c || c || c || c}
    Dictionary & Factorized & Row & Columnar \\ \hline
\begin{tabular}{|c|c|}
\hline
!<A=$a_1$, B=$b_1$>!&!1!\\\hline
!<A=$a_1$, B=$b_2$>!&!1!\\\hline
!<A=$a_2$, B=$b_3$>!&!1!\\\hline
\end{tabular}
    & 
\begin{tabular}{|c|c|}
\hline
!$a_1$!&
\begin{tabular}{|c|c|}
\hline
!$b_1$!&!1!\\\hline
!$b_2$!&!1!\\\hline
\end{tabular}
\\\hline
!$a_2$!&
\begin{tabular}{|c|c|}
\hline
!$b_3$!&!1!\\\hline
\end{tabular}
\\\hline
\end{tabular}
     &
\begin{tabular}{|c|c|}
\hline
!0!&!<A=$a_1$, B=$b_1$>!\\\hline
!1!&!<A=$a_1$, B=$b_2$>!\\\hline
!2!&!<A=$a_2$, B=$b_3$>!\\\hline
\end{tabular}
    &
!<A=!
\begin{tabular}{|c|c|}
\hline
!0!&!$a_1$!\\\hline
!1!&!$a_1$!\\\hline
!2!&!$a_2$!\\\hline
\end{tabular}
!, B=!
\begin{tabular}{|c|c|}
\hline
!0!&!$b_1$!\\\hline
!1!&!$b_2$!\\\hline
!2!&!$b_3$!\\\hline
\end{tabular}
!>!
     \\
    \end{tabular}
    \vspace{-0.3cm}
    \caption{Different data layouts for relations.}
    \label{fig:my_label}
    \vspace{-0.3cm}
\end{figure}

\section{Semantics}
\label{sec:sem}
\revision{\lang is mainly a standard functional programming language, but we study its specificity in this section. First, we show its typing/kinding properties. We then introduce a denotational semantics for \lang that sheds another light on the language and helps us prove the correctness of the transformation rules presented in Section~\ref{sec:opt}. The operational semantics and type safety proofs can be found in the supplementary materials.}
% In this section, we first prove some simple typing properties for \lang. Then we give it an denotational semantics. Finally, we prove the correctness of the transformation rules presented in Section~\ref{sec:opt}.
% The operational semantics and type safety proofs are in the supplementary materials.

\subsection{Typing}
% Next, we present the following typing properties that are essential for the type safety of \lang.
\revision{\lang satisfies the following essential typing properties.}

\begin{lemma}
Let $\mathbf{T}$ denote the set of all types of \lang.
$\otimes$ is a well-defined partial operation $\mathbf{T}\times\mathbf{T}\to\mathbf{T}$.
\end{lemma}

\begin{proposition}
\label{prop:first}
	Every type/term defined using the inference rules of Figure~\ref{fig:typesystem} has a unique kind/type.
\end{proposition}

\begin{proof}[Proof Sketch]
By induction on the structure of types/terms and case analysis on each kinding/typing rule. It is straightforward for most rules using the induction hypothesis. For \revision{the typing rules of} dictionaries there are two cases on whether the dictionary is empty or not, and the type annotation ensures the property for the empty dictionary. 
\revision{As for}
%For typing 
!sum! and !let! which have a bound variable, we use the induction hypothesis on !e1! first.
\end{proof}

\subsection{Denotational Semantics}

\revision{The kind system acts as a type refinement machinery. Roughly, a type is to be considered by default of kind }!Type!\revision{. Otherwise, the kind indicates that the type carries more structure, more precisely that of a semi-module.}
% We start with the denotational semantics of types. 
% A type !T! is interpreted as a set |T|. 
% The types of kind \code{SM(S)} are interpreted with an additional monoid structure, as follows.
% The kind \dsnoctx{\code{T :: SM(S)}} is interpreted as the category of S-semi-modules and all functions between them.
\noindent
\revision{More formally, the interpretation of types is given by induction on the kinding rules, and is shown in Figure~\ref{fig:denot_sem}. A type of kind} !Type! \revision{is interpreted as a set, while a type of kind} !SM(S)! \revision{is interpreted as a S-semi-module. A scalar type} !S! \revision{represents a semi-ring and is therefore canonically a S-semi-module. A product of S-semi-modules is a semi-module, and so is the tensor product $\otimes_S$ of two S-semi-modules. One way to describe $\otimes_S$ is as the bifunctor on the category of S-semi-modules and S-module homomorphisms that classifies S-bilinear maps. It is an analogue for semi-modules to the tensor product of vector spaces. For more details on tensor products see e.g. \cite{conrad2018tensor}. The interpretation for a dictionary type is analogous to a free vector space on $|T1|$, in which every element is a finite formal sum of elements of \dsnoctx{\code{T2}}. One can show by induction that all our types of kind SM(S) are free S-semi-modules. Hence \dsnoctx{\code{T2}} is a free S-semi-module and this implies that the interpretation for a dictionary type can itself be seen as a free S-semi-module.}

% \begin{center}
% \begin{tabular}{l c l}
% \dsnoctx{\code{S}} & = & ($S$, $+$, $0$) \\
% \dsnoctx{\code{<a1:T1, ..., an:Tn>}} & = &\dsnoctx{\code{T1}} $\times$ ... $\times$ \dsnoctx{\code{Tn}} \\
% \dsnoctx{\code{\{T1 -> T2\}}} & = & $\sumplus\limits_{a \in |T1|}$\dsnoctx{\code{T2}}\\
% \dsnoctx{\code{T1} $\otimes_S$ \code{T2}} & = & \dsnoctx{\code{T1}} $\otimes_S$ \dsnoctx{\code{T2}} \\
% \end{tabular}
% \end{center}

% \revision{\input{figures/types_denot_semantics}}

\noindent
% Scalar types are interpreted as the monoid of the underlying semi-ring ($S$, $+$, $\times$, $0$, $1$). 
% The denotational semantics of a record type is the product of the underlying monoids. 
% The interpretation for a dictionary type is an analogue to a free vector space on $|T1|$, in which an element is a finite formal sum of elements of \dsnoctx{\code{T2}}. This vector space is a monoid with a component-wise addition.
% Finally, the type \dsnoctx{\code{T1} $\otimes_S$ \code{T2}} is interpreted as \dsnoctx{\code{T1}} $\otimes_S$ \dsnoctx{\code{T2}}, where $\otimes_S$ is the bifunctor on the category of S-semi-modules and S-module homomorphisms that classifies bilinear maps. 

For the semantics of environments $\Gamma=$!x1:T1, ..., xn:Tn!, we use:

\begin{center}
\begin{tabular}{l c l}
\dsnoctx{$\Gamma$} & = & \dsnoctx{\code{T1}} $\times$ ... $\times$ \dsnoctx{\code{Tn}} \\
\end{tabular}
\end{center}

\noindent
% The denotational semantics of a term \dsnoctx{$\Gamma\vdash~$\code{e: T}} is a function from \dsnoctx{$\Gamma$} to \dsnoctx{\code{T}}. 
\revision{A term \dsnoctx{$\Gamma\vdash~$\code{e: T}} is interpreted as a function from \dsnoctx{$\Gamma$} to \dsnoctx{\code{T}}.}
When it is clear from the context, we use \dsnoctx{\code{e}} instead of \dsnoctx{$\Gamma\vdash~$\code{e: T}}. 
We use the notation !v!$\mytimes$!k! to mean the vector whose only non-zero component !v! is at position !k! in $\sumplus\limits_{a \in |T1|}$\dsnoctx{\code{T2}}.
% To range over environments we use $\gamma$.
\revision{We denote by $\gamma$ any assignment of the variables of a context $\Gamma$.}
The denotational semantics for terms is shown in Figure~\ref{fig:denot_sem}.
$Prom_{\code{S1}\rightarrow\code{S2}}$ maps the elements of the scalar semi-ring !S1! to !S2!.
Every scalar type !S! \revision{is a semi-ring and as such admits distinguished elements} !0! and !1!. 
The action of !S! on a type !T!::\code{SM(S)} \revision{thus} restricts to an action !*! of the booleans on !T!. This gives the presented description to the semantics of conditionals which we use in the next section. 
% It is given by \dsbegin!if e1 then e2 else e3!\dsend{} \myeq \dsbegin !e1! \dsend{} * \dsbegin !e2! \dsend{} + \dsbegin !not(e1)! \dsend{} *  \dsbegin !e3! \dsend{}.
\revision{For the semantics for dictionaries, we use a formal infinite sum, but similarly to standard polynomials this sum actually has a finite support and thus behaves like a finite sum in all contexts. For the semantics of }!sum!\revision{, we apply the semantics of }!e2! \revision{ component-wise to the formal sum that is the semantics of }!e1!\revision{. The resulting real sum is thus over a finite support, and is therefore well-defined.}

\begin{figure*}[t]
\setlength{\tabcolsep}{0.3em}
\begin{tabular}{|c|}
\hline
\begin{tabular}{l c l c l c l}
\dsnoctx{\code{S}} & \myeq & ($S$, $+$, $0$) & \hspace{0.2cm}  &
\dsnoctx{\code{<a1:T1, ..., an:Tn>}} & \myeq  &\dsnoctx{\code{T1}} $\times$ ... $\times$ \dsnoctx{\code{Tn}} \hspace{0.7cm} \\
\dsnoctx{\code{T1} $\otimes_S$ \code{T2}} & \myeq  & \dsnoctx{\code{T1}} $\otimes_S$ \dsnoctx{\code{T2}} & &
\dsnoctx{\code{\{T1 -> T2\}}} & \myeq  & $\sumplus\limits_{a \in |T1|}$\dsnoctx{\code{T2}}\\
\end{tabular} \\ \hline
\begin{tabular}{l c}
% Types: & \\\hline
% Terms: & \\\hline
\begin{tabular}{l c l}
\dsbegin!x!\dsend{} & \myeq & \dsctx{}!(x)! \\
\dsbegin!c!\dsend{} & \myeq & !c! \\
\dsbegin!true!\dsend{} & \myeq & 1 \\
\dsbegin!false!\dsend{} & \myeq & 0 \\
\dsbegin!not(e)!\dsend{} & \myeq & 1 - \dsbegin!e!\dsend{} \\
\dsbegin!e.ai!\dsend{} & \myeq & $\pi_i$(\dsbegin!e!\dsend{}) \\
\dsbegin!op(e)!\dsend{} & \myeq & !op!(\dsbegin!e!\dsend{}) \\
\dsbegin!e1 + e2!\dsend{} & \myeq & \densem{\code{e1}} $+$ \densem{\code{e2}} \\
\dsbegin!e1 * e2!\dsend{} & \myeq & \densem{\code{e1}} $*$ \densem{\code{e2}} \\
\end{tabular}&
\begin{tabular}{l c l}
\dsbegin!<a1=e1,...,an=en>!\dsend{} & \myeq & <\dsbegin!e1!\dsend{}, ..., \dsbegin!en!\dsend{}> \\
\dsbegin!let x = e1 in e2!\dsend{} & \myeq & \dsbegin!e2!\dsgend{\dsctx[\densem{\code{e1}}/\code{x}]} \\
\dsbegin!promote!$_{\code{S1},\code{S2}}$!(e)!\dsend{} & \myeq & $Prom_{\code{S1}\rightarrow\code{S2}}$(\dsbegin!e!\dsend{}) \\
\dsbegin!if e1 then e2 else e3!\dsend{} & \myeq & 
\begin{tabular}{l l} 
% \densem{\code{e2}} & (\densem{\code{e1}} \myeq{} tt) \\
% \densem{\code{e3}} & (\densem{\code{e1}} \myeq{} ff)
\dsbegin !e1! \dsend{} $*$ \dsbegin !e2! \dsend{} $+$ \\
(1 - \dsbegin !e1! \dsend{}) $*$  \dsbegin !e3! \dsend{}
\end{tabular}
 \\
\dsbegin!e1(e2)!\dsend{} & \myeq & $\pi_{\text{\densem{\code{e2}}}}$(\dsbegin!e1!\dsend{}) \\
\dsbegin!{}!$_{\code{T1,T2}}$\dsend{} & \myeq & $0_{\code{\{T1->T2\}}}$ \\
\dsbegin!{k1->v1,...,kn->vn}!\dsend{} & \myeq & $\sum\limits_{i\in[1..n]}$\densem{\code{vi}}$\mytimes$\densem{\code{ki}} \\
\end{tabular}\\
\end{tabular}\\
\dsbegin!sum(x in e1) e2!\dsend{} \myeq{} $\sum\limits_{k\in X}$\dsgen{\code{e2}}{\dsctx[<k,$a_k$>/\code{x}]} \hspace{0.5cm} (\densem{\code{e1}} \myeq{} $\sum\limits_{k\in X}a_k\mytimes k$) \\ \hline
\end{tabular}
\vspace{-0.2cm}
\caption{Denotational Semantics for \revision{types and terms of} \lang.}
% \vspace{-0.2cm}
\label{fig:denot_sem}
\end{figure*}

\revision{\begin{proposition}[Substitution lemma]
For all $\Gamma \vdash~ \code{e1: T1}$ and $\Gamma,\code{x: T1}\vdash~ \code{e2: T2}$, the following holds: 
\dsnoctx{\code{e2}}\code{[}\dsnoctx{\code{e1}}/\code{x]} = \dsnoctx{\code{e2[e1/x]}}.
\end{proposition}}

\revision{\begin{theorem}[Soundness]
For all closed terms $\vdash \code{e: T}$ and $\vdash \code{v: T}$ where \code{v} is a value, if \code{e} reduces to \code{v} in the operational semantics, then \dsnoctx{\code{e}} = \dsnoctx{\code{v}}.
\end{theorem}}

\revision{\begin{proof}[Proof sketch]
For both Proposition 7.3 and Theorem 7.4, the proof is by induction on the structure of terms and case analysis on the structure of terms in the first case, and on the last rule used of the operational semantics in the other case. The only non-standard cases are the ones involving a dictionary or \code{sum}. More details can be found in the supplementary materials.
\end{proof}}

\subsection{Correctness of Optimizations}
% In this section, we prove correct the optimizations of Figure~\ref{fig:opt_rules}. The formal $\sum$ notation in the semantics automatically provides an efficient and sound calculus that is reminiscent of the algebra of polynomials. We make use of this in the following proofs.
\revision{The denotational semantics allows us to easily prove correctness of the optimizations of Figure~\ref{fig:opt_rules}. In particular, the formal $\sum$ notation in the semantics mechanically  provides an efficient and sound calculus that is reminiscent of the algebra of polynomials. We make use of this calculus in the following proofs.}

\begin{proposition}
The vertical loop fusion rules of Figure~\ref{fig:opt_rules} are sound.
\end{proposition}

\begin{proof}[Proof] We prove the first rule. The second rule is proved similarly.
\\
\dsbegin!let y = sum(x in e1) {f1(x.key)->x.val} in sum(x in y){f2(x.key)->x.val}!\dsend{} \tab = \\
\dsbegin!sum(x in y){f2(x.key)->x.val}!\dsgend{$\gamma'$} \tab ($\gamma'$ = $\gamma$[\dsbegin!sum(x in e1) {f1(x.key)->x.val}!\dsend{}/ y]) \tab = \\
\dsbegin!sum(x in y){f2(x.key)->x.val}!\dsgend{$\gamma'$} \tab ($\gamma'$ = $\gamma$[$\sum\limits_{k \in X}a_k\mytimes$\densem{\code{f1}}($k$)/ y], \densem{\code{e1}}=$\sum\limits_{k \in X}a_k\mytimes k$) \tab = \\
$\sum\limits_{k \in X}a_k\mytimes$\densem{\code{f2}}(\densem{\code{f1}}($k$)) \tab (\densem{\code{e1}}=$\sum\limits_{k \in X}a_k\mytimes k$) \tab = $\sum\limits_{k \in X}a_k\mytimes$\densem{\code{f2}$\circ$\code{f1}}($k$)) \tab (\densem{\code{e1}}=$\sum\limits_{k \in X}a_k\mytimes k$) \tab = \\
\dsbegin!sum(x in e1){f2(f1(x.key))->x.val}!\dsend{}
\end{proof}

\begin{proposition}
\label{theorem:loopfact}
The loop factorization rules of Figure~\ref{fig:opt_rules} are sound.
\end{proposition}

\begin{proof}[Proof]
We prove the first rule, and the second rule is proved similarly.
\\
\dsbegin!sum(x in e1) e2 * f(x)!\dsend{} \tab = \tab
$\sum\limits_{k \in X}$\dsbegin!e2 * f(x)!\dsgend{$\gamma'$} \tab ($\gamma'$ = $\gamma$[<$k,a_k$>/ x], \densem{\code{e1}}=$\sum\limits_{k \in X}a_k\mytimes k$) \tab = \\
$\sum\limits_{k \in X}$\dsbegin!e2!\dsend{} $*$ \dsbegin!f!\dsend{}<$k,a_k$> \tab (\densem{\code{e1}}=$\sum\limits_{k \in X}a_k\mytimes k$) \tab = (bilinearity) \\
\dsbegin!e2!\dsend{} $*$ $\sum\limits_{k \in X}$\dsbegin!f!\dsend{}<$k,a_k$> \tab (\densem{\code{e1}}=$\sum\limits_{k \in X}a_k\mytimes k$) \tab =\\
\dsbegin!e2!\dsend{} $*$ $\sum\limits_{k \in X}$\dsbegin!f(x)!\dsgend{$\gamma'$} \tab ($\gamma'$ = $\gamma$[<$k,a_k$>/ x], \densem{\code{e1}}=$\sum\limits_{k \in X}a_k\mytimes k$) \tab =\\
\dsbegin!e2!\dsend{} $*$ \dsbegin!sum(x in e1) f(x)!\dsend{} \tab = \tab \dsbegin!e2 * sum(x in e1) f(x)!\dsend{}\\
\end{proof}

% The correctness of horiztonal fusion, loop-invariant code motion, and 
% loop memoization \revision{based on both operational and denotational semantics} can be found in the supplementary materials.
\revision{The correctness proofs of the remaining optimizations, horizontal fusion, loop-invariant code motion, and loop memoization, based on both operational and denotational arguments can be found in the supplementary materials.}

\section{Implementation}
\label{sec:impl}
\lang is implemented as an external domain-specific language. 
The entire compiler tool-chain is written in Scala.
The order of rewrite rules are applied as follows until a fix-point is reached: 
1) loop fusion, 2) loop-invariant code motion, 3) loop factorization, and 4) loop memoization. 
After each optimization, generic optimization such as DCE, CSE, and partial evaluation are also applied.
Note that we currently expect the loop order to be specified correctly by the user.
Finally, the optimized program is translated into C++.

\subsection{C++ Code Generation}
The code generation for \lang is mostly straightforward, thanks to the first-order nature of
most of its constructs.
Thus, we do not face the technical challenges of compiling polymorphic higher-order functional languages (e.g., all objects are stack-allocated, hence there is no need for GC).
The key challenging construct is !sum! which is translated into \code{for}-loops.
Furthermore, for the case of summations that produce dictionaries, 
the generated loop performs destructive updates to the collection, to improve the performance~\cite{henriksen2017futhark}.
% This is thanks to the tail-recursive nature of the !sum! construct.

\subsection{C++ Runtime}
The C++ runtime employs an efficient hash table implementation based on closed hashing for dictionaries.\footnote{https://github.com/greg7mdp/parallel-hashmap}
For dictionaries with !dense_int! keys, the runtime either uses \code{std::array} or \code{std::vector} depending on whether the size is statically known during compilation time.
Finally, for implementing records, \lang uses \code{std::tuple}.

\subsection{Semi-Ring Extensions}
\label{sec:semiring_ext}

\smartpara{Scalar Semi-Rings}
Throughout the paper, we only focused on three important scalar semi-rings, and the corresponding record and dictionary semi-rings. 
FAQ~\cite{abo2016faq} introduced several semi-ring structures with applications on graphical models, coding theory, and logic.
Also, semi-rings were used for language recognition, reachability, and shortest path problems~\cite{dolan2013fun,pilatus19ecoop}.
\lang can support such applications by including additional scalar semi-rings, a subset of which are presented in Table~\ref{tab:semiring_scalar}. 
The !promote! construct can be used to annotate numeric values with the type of the appropriate types in such cases.

\smartpara{Non-scalar Semi-Rings}
The support for semi-ring extensions in \lang is beyond scalar types.
As an example, \lang supports the (semi-)ring of the covariance matrix~\cite{Nikolic:2018:IVM:3183713.3183758}.
For each $n\in\mathbb{Z}$, the domain $\mathbb{D}$ of this semi-ring is a triple $<\mathbb{R}, \mathbb{R}^n, \mathbb{R}^{n\times n}>$. The additive and multiplicative identities are defined as $0^{\mathbb{D}}\myeqq<0, 0^n, 0^{n\times n}>$ and $1^{\mathbb{D}}\myeqq <1, 0^n, 0^{n\times n}>$. For each $a\myeqq<s_a, v_a, m_a>$ and $b\myeqq<s_b, v_b, m_b>$, the addition and multiplication are defined as:

\begin{center}
\begin{tabular}{r c l}
$a +^{\mathbb{D}} b$ &\myeq& $<s_a + s_b, v_a + v_b, m_a + m_b>$\\
$a \times^{\mathbb{D}} b$ &\myeq& $<s_a * s_b, s_a * v_b + v_a * s_b, s_b * m_a + s_a * m_b + v_a * v_b + v_b * v_a >$
\end{tabular}
\end{center}

\noindent We use this semi-ring to compute covariance matrix as aggregates over relations (cf. Section~\ref{sec:exp:indbml}).

\subsection{\revision{Language Extensions}}
\label{sec:langext}

\begin{figure}
\begin{tabular}{|l c l|}
\hline
\langring & &\\ \hline
% !let y=sum(x in e1) {f1(x.key)->x.val}! & \multirow{2}{*}{\transto} & !sum(x in e1)! \\
% !sum(x in y) f3(x)! & & !  f3(<key = f1(x.key), val = x.val>)! \\ \hline \hline
!-(-e)! & \transto & !e! \\ \hline
!e + (-e)! & \transto & !0!\\ 
\hline \hline
\langclosure & &\\ \hline
!1 + e * closure(e)! & \transto & !closure(e)! \\ \hline
!1 + closure(e) * e! & \transto & !closure(e)! \\ 
\hline \hline
\langproduct & &\\ \hline
!(prod(x in e1) f1(x)) * (prod(x in e1) f2(x))! & \transto & !prod(x in e1) f1(x) * f2(x)! \\ \hline \hline
\langrec & &\\ \hline
!rec(x => let y=e1 in f(x,y))(e2)! & \multirow{1}{*}{\transto} & !let y=e1 in rec(x => f(x,y))(e2)! \\ \hline
\end{tabular}
\vspace{-0.4cm}
\caption{\revision{Additional transformation rules for language extensions of \lang.}}
\vspace{-0.4cm}
\label{fig:opt_rules_ext}
\end{figure}

\revision{
In this section, we define possible language extensions over \lang. 
Apart from an additional expressive power, each extension enables further optimizations, which are demonstrated in Figure~\ref{fig:opt_rules_ext}.
We use \langext{\code{X}} to denote \lang extended with \code{X}.
}

\revision{\smartpara{\langring: \lang + Ring Dictionaries}} We have consistently talked about semi-ring structures, and how semi-ring dictionaries can be formed using value elements with such structures.
There is another important structure, referred to as \textit{ring}, \revision{for the cases that the addition operator admits an inverse. The transformation rules enabled by the ring structure are shown in Figure~\ref{fig:opt_rules_ext}.} 
As it can be observed in Table~\ref{tab:semiring_scalar}, real and integer sum-products form ring structures. Similarly to semi-ring dictionaries, one can obtain ring dictionaries by using values that form a ring. In this case, the additive inverse of a particular ring dictionary is a ring dictionary with the same keys but with inverse value elements.

\smartpara{\langclosure: \lang + Closed Semi-Rings} Orthogonally, one can extend the semi-ring structure with a closure operator~\cite{dolan2013fun}.
In this way, transitive closure algorithms can also be expressed by generalizing semi-rings to closed semi-rings~\cite{lehmann1977algebraic}.
In many cases, the semi-ring structures involve an additional idempotence axiom (\code{a + a = a}) resulting in dioids.
The closure operator for dioids is called a Kleene star and the extended structure is referred to as Kleena algebra, which is useful for expressing path problems in graphs among other use-cases~\cite{gondran2008graphs}.
This structure can be reflected in our kind-system; the product of dioids/Kleene algebras forms a dioid/Kleene algebra. In future work, we would like to investigate how to express the standard algorithm that computes \code{closure}($A$) for a matrix $A$ over a Kleene algebra in terms of a program involving semi-ring dictionaries over a Kleene algebra.

\smartpara{\langproduct: \lang + Product} 
We have only considered the summation over semi-ring dictionaries.
One can use \code{prod} instead of \code{sum}.
This would allow to elegantly express universal quantification over the possible assignments of that variable (like in FAQ~\cite{abo2016faq} to express quantified Boolean queries). As an example, checking if the predicate \code{p} is satisfied by all elements of relation \code{R} is phrased as: \code{prod(r <- R) p(r)}. 
The commutative monoid structure of multiplication allows for optimizations with a similar impact as horizontal loop fusion (cf. Figure~\ref{fig:opt_rules_ext}).

\smartpara{\langrec: \lang + Recursion} Apart from supporting the closure and product constructs, it is possible to support more general forms of recursion. 
As shown for matrix query languages~\cite{matlang_sigrec}, an additional for-loop-style construct can express summation, product, transitive closure, as well as matrix inversion. 
This general form of recursion also allows for iterations, similarly to the \code{while} construct in IFAQ~\cite{ifaq-cgo} that enables iterative computations required for optimization producures such as batch gradient decent (BGD).
The additional expressive power of this construct comes with limited optimization opportunies; loop fusion and factorization are no longer applicable to them, however, code motion can still be leveraged (cf. Figure~\ref{fig:opt_rules_ext}).

\section{Experimental Results}
\label{sec:exp}

\subsection{Experimental Setup} 
We run our experiments on a iMac equipped with
an Intel Core i5 CPU running at 2.7GHz, 32GB of DDR3 RAM with OS X 10.13.6. We use CLang 1000.10.44.4 for compiling the generated C++ code using the O3 flag. \revision{Our competitor systems} use Scala 2.12.2, Spark 3.0.1, Python 3.7.4 \revision{(Python 2.7.12 for MorpheusPy)}, NumPy 1.16.2, and SciPy 1.2.1.
\revision{All experiments are run on one CPU core.}\footnote{Prior work on parallelism for database query engines~\cite{Volcano}, nested data processing (flattening and shredding~\cite{trance_vldb}), and sparse linear algebra~\cite{Kjolstad:2017:TAC:3152284.3133901} can be transferred to \lang, which we leave as future work.}
\revision{We} measure the average run time execution \revision{of five runs excluding} the loading time.

\subsection{Database Workloads} 
\label{sec:exp:tpch}

\begin{wrapfigure}{r}{0.5\columnwidth}
\vspace{-0.4cm}
\includegraphics[width=0.48\columnwidth]{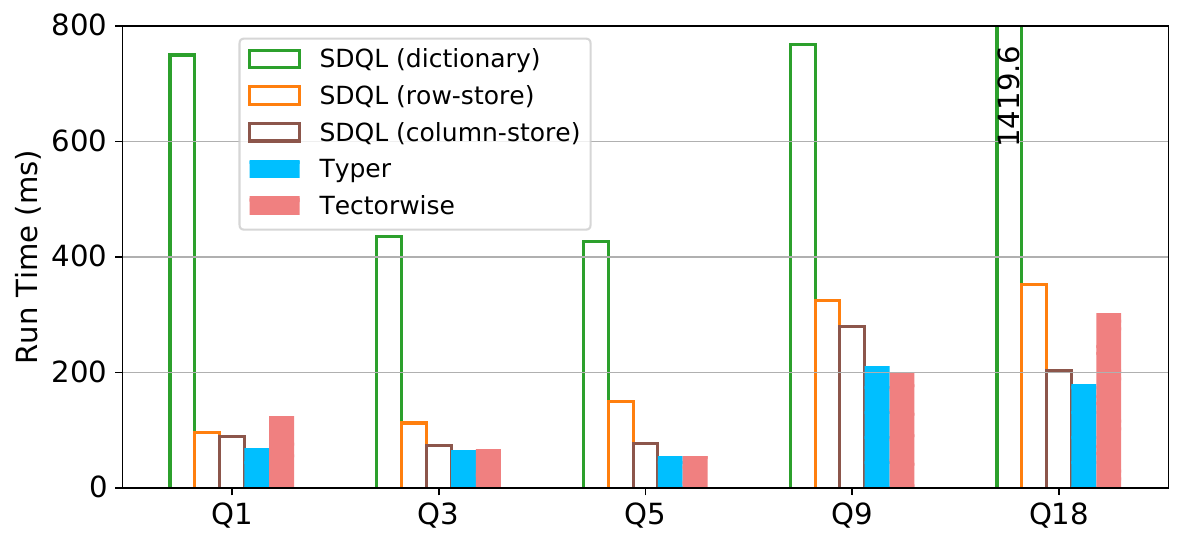}
\vspace{-0.4cm}
\caption{Run time results for TPCH queries comparing different data layouts in \lang, Typer, and Tectorwise.}
\label{fig:exp:tpch}
\vspace{-0.4cm}
\end{wrapfigure}

In this section, we investigate the performance of \lang for online analytical processing (OLAP) workloads used in the databases.
For this purpose, 
\revision{we} compare the performance of generated optimized code for the dictionary layout, row layout, and columnar layout of \lang with the open source implementation\footnote{https://github.com/TimoKersten/db-engine-paradigms}~\cite{kersten2018everything} of two state-of-the-art analytical query processing engines: 1) Typer for HyPer~\cite{Neumann11}, and 2) Tectorwise for Vectorwise~\cite{monetdb-handwritten}.

For these experiments, we use TPCH, the main benchmark for such workloads in databases. 
Instead of running all 22 TPCH queries, we only use a representative subset of them for the following reasons.
First, previous research~\cite{Boncz2014,kersten2018everything} identified that this subset has the ``choke points'' of all TPCH queries. 
Second, the open source implementations of Typer and Tectorwise only support this subset.
We further restricted this subset to the queries that construct intermediate dictionaries; we excluded Q6 as it does not have any joins or group-by aggregates.

\revision{Figure~\ref{fig:exp:tpch} shows that the} row layout for input relations \revision{leads to a} $4.2\times$ speedup \revision{over} the standard dictionary layout.
\revision{The columnar layout further improves the performance by $1.5\times$.
This is due to} improved cache locality\revision{, as unused columns are not read into cache in case of the columnar layout.}
The columnar layout 
\revision{
leads to performance on par with Tectorwise, but \lang remains  
 about $20\%$ slower than Typer.
 The performance can be further improved} by better memory management and string processing techniques, as used in Typer and Tectorwise.

\begin{figure}[t]
\includegraphics[width=0.48\columnwidth]{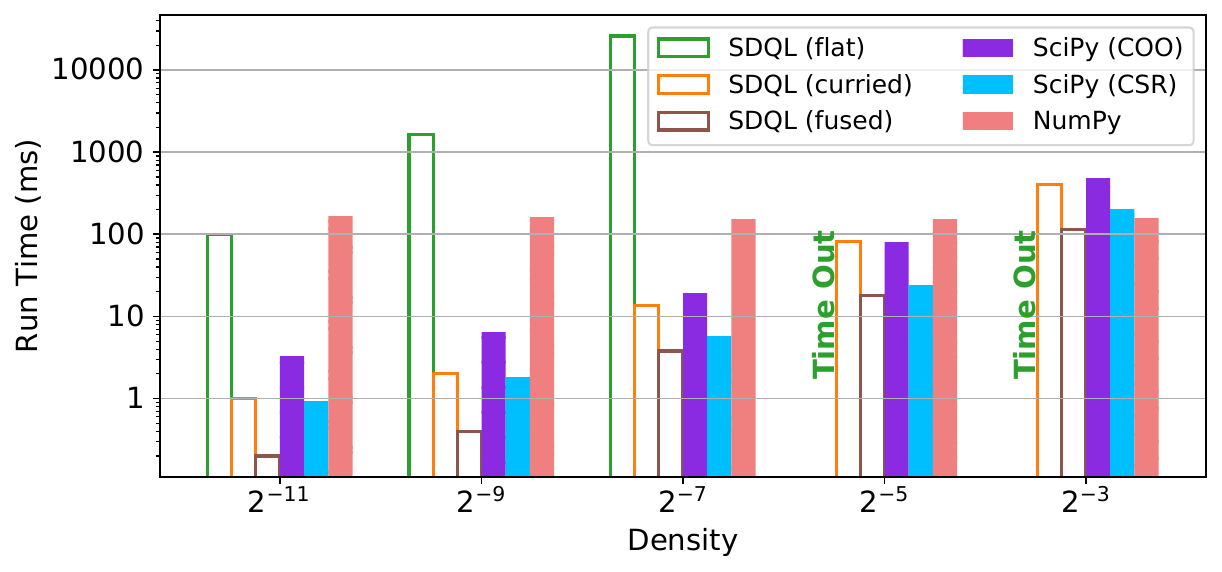}~\includegraphics[width=0.48\columnwidth]{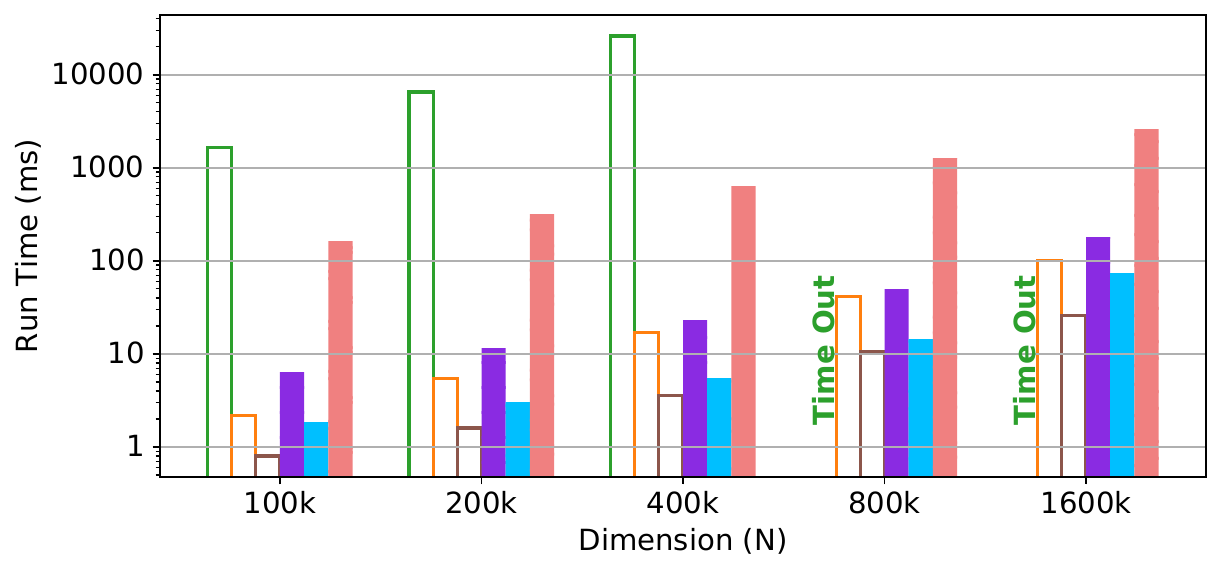}
\vspace{-0.4cm}
\caption{Run time results for computing the covariance matrix comparing different optimizations and representations in \lang, SciPy, and NumPy. The dimension for the input matrix of the left figure is $100000\times 100$, and the dimension of the input matrix of the right figure is $N \times 100$ with the density of $2^{-7}$.}
\label{fig:exp:la}
\vspace{-0.5cm}
\end{figure}

\subsection{Linear Algebra Workloads}
In this section, we investigate the performance of \lang for linear algebra workloads. 
We consider both matrix and higher-order tensor workloads. 
For the matrix processing workload, we use NumPy and SciPy as competitors, which use dense and sparse representations for matrices. This workload involves matrix transpose, which is not supported by systems such as \taco~\cite{Kjolstad:2017:TAC:3152284.3133901}.
For the tensor processing workloads, 
we use \taco~\cite{Kjolstad:2017:TAC:3152284.3133901} as the only competitor. SciPy does not support higher-order tensors, and it was shown before~\cite{Kjolstad:2017:TAC:3152284.3133901,taco_format} that on these workloads, \taco is faster than systems such as SPLATT~\cite{smith2015splatt}, Tensor Toolbox~\cite{bader2008efficient}, and TensorFlow~\cite{abadi2016tensorflow}. For a fair comparison, we have included the time for assembling the output tensor in \taco.

\smartpara{Sparse Matrix Processing} First, we consider the task of computing the covariance matrix $X^T X$ (cf. Section~\ref{sec:la}), where $X$ is a synthetically generated input data matrix of varying dimensions and density.
We consider the following different versions of the generated code from \lang: 1) unoptimized, which is the uncurried representation of matrices, 2) curried, which uses the curried representation, and 3) fused, which additionally fuses the transpose and multiplication operators. 

As Figure~\ref{fig:exp:la} shows, using curried representation can provide asymptotic improvements over the na\"ive representation, thanks to the improved matrix multiplication operator \revision{(cf. Section~\ref{sec:curried})}. Furthermore, performing fusion can provide $2\times$ speedup on average. 
The usage of dense representation (by NumPy) can provide better implementations as the matrix becomes more dense; however, for smaller densities, sparse representations (by SciPy and \lang) can be up to two orders of magnitude faster. 
Finally, the most optimized version of the generated code by \lang is in average $3\times$ and $2\times$ faster than the COO and CSR represenations of SciPy, respectively, thanks to fusion and the efficient low-level code generated by \lang.

\smartpara{Sparse Tensor Processing} Next, we consider three higher-order tensor workloads on NELL-2, a real world dataset coming from the Never Ending Language-Learning project~\cite{carlson2010toward}.
Table~\ref{tbl:tensorexp} shows the performance comparison for these workloads. We observe that especially for a medium range of sparsity \system is faster than \taco (from $1.4\times$ to $23\times$). For sparser scenarios, \taco shows better performance (up to $1.3\times$), thanks to the DCSR format and its merge-based multiplications. A similar observation on hash/CSR formats has been made in~\cite{taco_format}.
% \begin{center}
% \begin{tabular}{c c c c}
% &
% \textbf{TTV} & 
% \textbf{TTM} & 
% \textbf{MTTKRP}  \\
% \textbf{LA Formulation}
% &
% $A_{ij}=\Sigma_k B_{ijk}c_k$ & 
% $A_{ijk}=\Sigma_k B_{ijl}C_{kl}$ & 
% $A_{ij}=\Sigma_{k,l} B_{ikl}C_{kj}D_{lj}$ \\
% \textbf{Einsum Notation}
% & !ijk,k->ij! &
% !ijl,kl->ijk! &
% !ijk,kj,lj->ij!
% \end{tabular}
% \end{center}

\subsection{Hybrid LA/DB Workload}
\label{sec:exp:indbml}
As the final set of experiments, we consider hybrid workloads that involve linear algebra and query processing. 
Figure~\ref{fig:exp:ladb} shows the experimental results for computing the covariance matrix.
We consider experiments that use \revision{1) nested, 2) relational, and 3) normalized matrix} input datasets.

\smartpara{Nested Data}
For nested data, we use our motivating biomedical example as the workload and variant data from 1000 genomes dataset as input \cite{1000g}. 
The experiment involves computing the covariance matrix of the join of !Genes! and !Variants! relations, by increasing the number of the elements of the former relation; this is 
synonymous to increasing the number of features in the covariant 
matrix by approximately 15, 30, 55, and 70.
We consider the following four versions of the generated code from \lang: 1) unoptimized code that uses uncurried representation for matrices, 2) curried version that uses curried representation for intermediate matrices, 3) a version that uses hash join for joining !Genes! and !Variants!, and 4) a version obtained by fusing intermediate dictionaries resulting from grouping and matrix transpose.
As our competitor, we only consider Trance~\cite{trance_vldb} for the query processing part\revision{, which implements an extension of \nrcplus with aggregation called \nrcagg and} uses Spark MLLib~\cite{mllib} for the linear algebra processing. 
This is because in-database machine learning frameworks such as IFAQ~\cite{ifaq-cgo}, LMFAO~\cite{Schleich:2019:LAE:3299869.3324961}, and Morpheus~\cite{chen2017towards,li2019enabling} do not support nested relations.

As Figure~\ref{fig:exp:bio} shows, we observe that using curried representation gives asymptotic improvements, and allows \lang to scale to larger inputs.
Furthermore, using hash join, gives an additional $3\times$ speedup. This speedup can be larger for larger !Genes! relations.
Performing fusion results in an additional $50\%$ speedup thanks to the removal of intermediate dictionaries and less loop traversals.
Finally, we observe around one order of magnitude performance improvement over Trance/MLLib thanks to the lack of need for unnesting, which is enabled by nested dictionaries provided by \lang.

\begin{table}[t]
\setlength{\tabcolsep}{0.27em}
% \small
\caption{Run time results of \system and \taco for TTV, TTM, and MTTKRP on Nell-2 dataset by varying the sparsity of the second and third operands. Both systems use a sparse representation for all tensor modes.}
\vspace{-0.2cm}
\label{tbl:tensorexp}
\begin{scriptsize}
\begin{tabular}{l l l l l c l l c l l c l l c l l}
\toprule
 & \multicolumn{1}{c}{Sparsity} & & \multicolumn{2}{c}{$2^{-11}$} & & \multicolumn{2}{c}{$2^{-9}$} & & \multicolumn{2}{c}{$2^{-7}$} & & \multicolumn{2}{c}{$2^{-5}$}  & & \multicolumn{2}{c}{$2^{-3}$} \\
 \cmidrule{2-2} \cmidrule{4-5} \cmidrule{7-8} \cmidrule{10-11} \cmidrule{13-14}  \cmidrule{16-17}
 \multicolumn{1}{c}{Kernel} & \multicolumn{1}{c}{LA Formulation} & & \multicolumn{1}{c}{\system} & \multicolumn{1}{c}{\taco}      & & \multicolumn{1}{c}{\system} & \multicolumn{1}{c}{\taco}  & & \multicolumn{1}{c}{\system} & \multicolumn{1}{c}{\taco} & & \multicolumn{1}{c}{\system} & \multicolumn{1}{c}{\taco} & & \multicolumn{1}{c}{\system} & \multicolumn{1}{c}{\taco} \\ \midrule
 TTV & $A_{ij}=\Sigma_k B_{ijk}c_k$ & & 621.8 & 466.3 & & 621.8 & 544.9 & & 632.0 & 866.2 & & 661.8 & 2088.1 & & 729.4 & 6742.7           \\
 TTM & $A_{ijk}=\Sigma_k B_{ijl}C_{kl}$ && 4534.2 & 5936.2 & & 4679.6 & 7851.6 & & 4764.2 & 15563.9 & & 5189.2 & 46153.7 & & 7146.6 & 169865.5     \\
 MTTKRP & $A_{ij}=\Sigma_{k,l} B_{ikl}C_{kj}D_{lj}$ & & 5.6 & 4.3 & & 18.4 & 17.3 & & 32.2 & 60.4 & & 103.2 & 388.1 & & 723.8 & 4371.1           \\
 \bottomrule
\end{tabular}
\end{scriptsize}
\vspace{-0.2cm}
\end{table}

\begin{figure}[t]
\begin{subfigure}{.49\textwidth}
\includegraphics[width=\columnwidth]{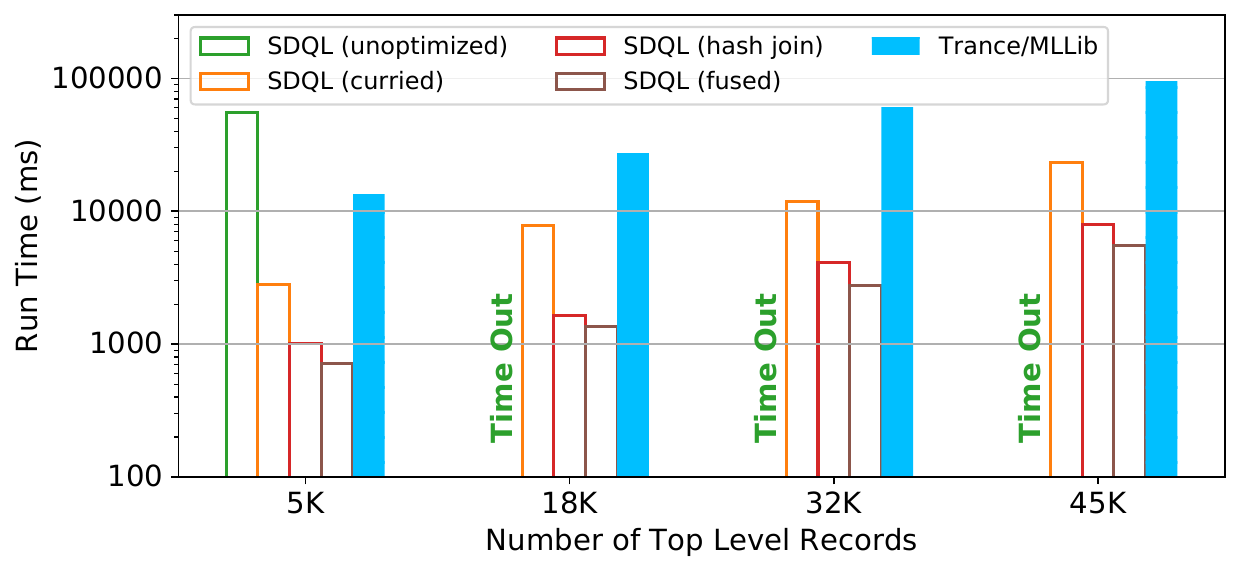}
\vspace{-0.6cm}
\caption{Biomedical query with different optimizations in \lang and Trance~\cite{trance_vldb}/MLLib.}
\label{fig:exp:bio}
\end{subfigure}
\hspace{.01\textwidth}
\begin{subfigure}{.48\textwidth}
\includegraphics[width=\columnwidth]{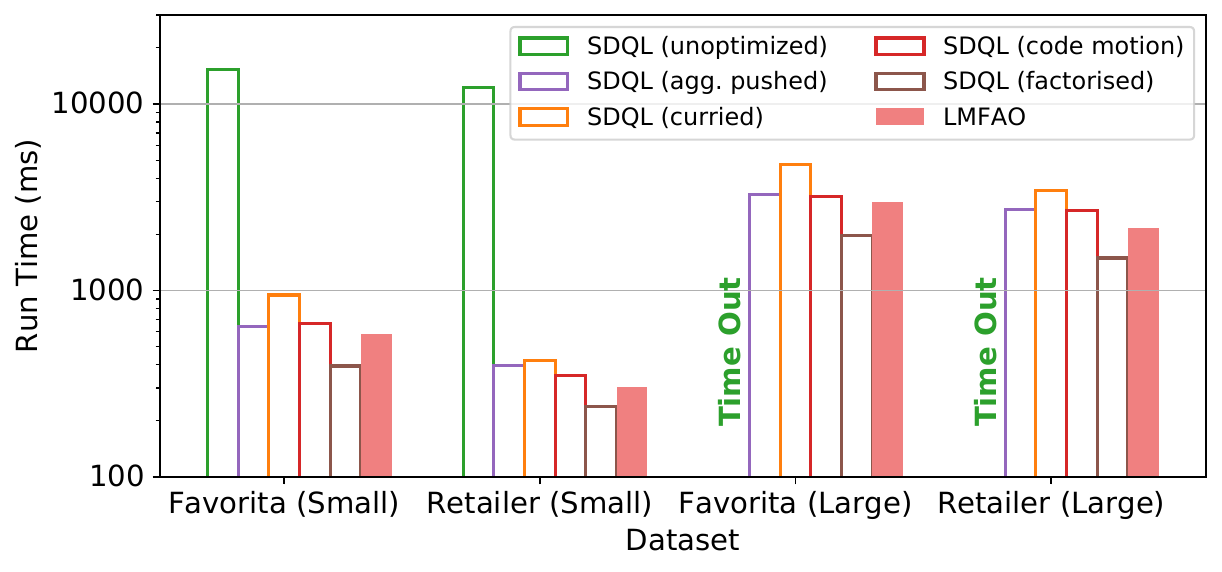}
\vspace{-0.6cm}
\caption{Retail forecasting using different optimizations in \lang and LMFAO~\cite{Schleich:2019:LAE:3299869.3324961}.}
\label{fig:exp:indb}
\end{subfigure}
\vspace{-0.4cm}
% \caption{Run time results for the biomedical query comparing different optimizations in \lang and Trance~\cite{trance_vldb}/MLLib~\cite{mllib}.}
\caption{Run time results for computing covariance matrix over nested and relational data.}
\label{fig:exp:ladb}
\vspace{-0.4cm}
\end{figure}

\smartpara{Relational Data} 
\revision{Next}, we compute the covariance matrix over the result of join of relational input. 
To do so, we use the semi-ring of the covariance matrix (cf. Section~\ref{sec:semiring_ext}).
We use two real-world relational datasets: 1)
\emph{Favorita}~\cite{favorita}, a publicly available Kaggle dataset, and 2)
\emph{Retailer}, a US retailer dataset~\cite{Schleich:2016:LLR:2882903.2882939}.
Both datasets are used in retail forecasting scenarios and consist of 6 and 5 relations, respectively. 
We only use five continuous attributes of these datasets.
We consider the following five versions of the generated code, where optimizations are applied accumulatively: 
1) unoptimized code that involves materializing the result of join before computing the aggregates, 
2) a version where all the aggregates are push down before the join computation, 
3) a curried version that uses a trie representation for input relations and intermediate results,
4) a version that applies loop-invariant code motion,
and 5) the most optimized version that performs loop factorization after all the previous optimizations.
As our competitor, we use LMFAO~\cite{Schleich:2019:LAE:3299869.3324961},
an in-DB ML framework that was shown to be up to two orders of magnitude faster than
\revision{Tensorflow~\cite{abadi2016tensorflow} and MADLib~\cite{hellerstein2012madlib}} for these two datasets.

\revision{Figure~\ref{fig:exp:indb} shows that first,}
pushing aggregates before join results in around one order of magnitude performance improvement, 
thanks to the removal of the intermediate large join. 
Second, using a curried representation degrades the performance, due to the fact that iterations over hash tables is more costly.
Third, code motion can leverage the trie-based iteration, and hoist invariant computations outside the loop \revision{to bring 30\% speed up in comparison with the curried version}. 
Finally, loop factorization leverages the distributivity rule for the semi-ring of covariance matrix,
and factorizes the costly multiplications outside the inner loops. 
On average, this optimization brings 60\% speed up in comparison with the previous version, and 40\% speed up over LMFAO.

\begin{figure}[t]
\includegraphics[width=0.48\columnwidth]{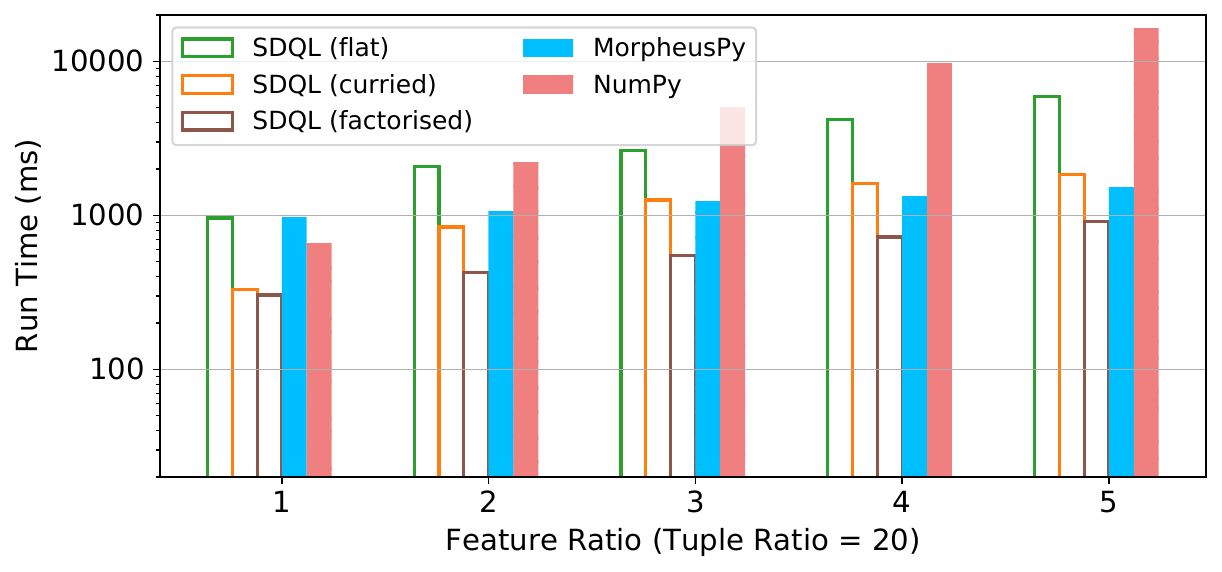}~\includegraphics[width=0.48\columnwidth]{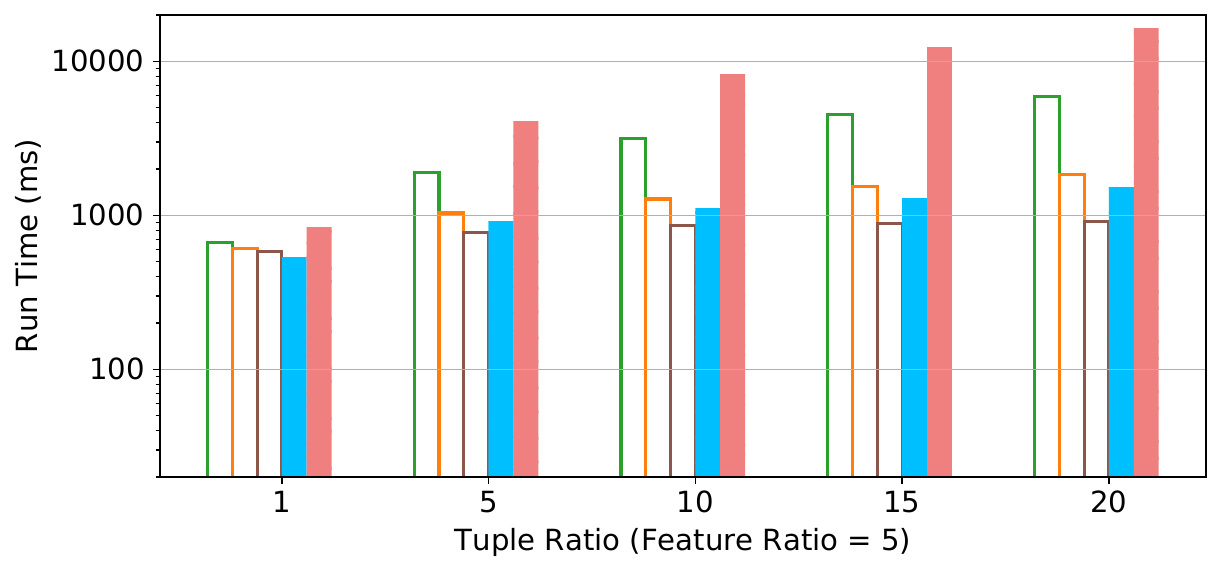}
\vspace{-0.4cm}
\caption{\revision{Run time of \lang, MorpheusPy, and NumPy for computing the covariance matrix over normalized matrix.
For both plots, $S$ has two features ($d_S=2$) and $R$ contains one million tuples ($n_R=1M$). In the left figure, $n_S=20M$ and $d_R\in\{2,4,6,8,10\}$. In the right figure, $d_R=10$ and $n_S\in\{1M,5M,10M,15M,20M\}$.}}
\label{fig:exp:indb_la}
\vspace{-0.3cm}
\end{figure}

\revision{
\smartpara{Normalized Matrix Data} 
Finally, we compute the covariance matrix over the join of relations represented as normalized matrices. 
We use the same semi-ring as the one for relational data. As the competitor, we consider NumPy and MorpheusPy~\cite{morpheuspy}, a Python-based implementation of Morpheus~\cite{chen2017towards}. 
The publicly available version of Morpheus only supports one primary-key foreign-key join of two relations~\cite{morpheus_bug}, i.e., $R \bowtie S$.
Figure~\ref{fig:exp:indb_la} shows the performance of Morpheus and \lang for computing the covariance matrix over such a join.
As in the original Morpheus paper~\cite{chen2017towards}, the join computation time for NumPy is not included. Also, the values for the primary key is the dense integer values between one and one million; 
thus all competitors use a dense representation for them. 
The number of tuples for $R$ is one million ($n_R=1M$), and for $S$ varies between millions ($n_S\in\{1M,5M,10M,15M,20M\}$).
The number of the features for $S$ is two ($d_S=2$), and for $R$ varies between two and ten ($d_R\in\{2,4,6,8,10\}$).

Figure~\ref{fig:exp:indb_la} shows that the NumPy-based implementation over the materialized join can have a better performance for relations with the same number of features. The factorized computation starts showing its benefits for larger feature ratios.
MorpheusPy is always better than the flat representation of \lang. This is thanks to vectorization, which shows its impacts further as the feature ratio increases.
Finally, we observe a superior performance for \lang once the curried representation and loop factorization are used.
% As the feature ratio increases, the vectorized implementation of MorpheusPy closes the performance gap.
As the tuple ratio increases, the speed up of \lang over MorpheusPy climbs up to $1.7\times$; this is because of loop factorization enabled by the curried representation for relation $S$.
MorpheusPy expresses aggregations and joins in terms of linear algebra operations using NumPy, which do not benefit from such optimizations.      
}

\section{Related Work}
In this section, we review the literature. Table~\ref{tbl:relwork} summarizes the differences between different data analytics approaches and \lang.

\smartpara{Relational Query Engines}
Just-in-time compilation of queries has been heavily investigated in the DB community~\cite{krikellas, 
Neumann11, dbtoaster, legobase_tods,dblablb, DBLP:journals/debu/ViglasBN14, crotty2015tupleware, Nagel:2014:CGE:2732977.2732984, karpathiotakis2015just, spark-sql,palkar2017weld,tahboub2018architect}.
% The PL community also have investigated compilation of functional
% collection programs~\cite{Kiselyov:2017:SFC:3009837.3009880,Mainland:2013:EVI:2500365.2500601,jfppushpull}.
As an alternative, vectorized query engines process blocks of data to remove interpretation overhead~\cite{monetdb-handwritten}.
None of these efforts have focused on handling hybrid DB/LA workloads as opposed to \lang.

\smartpara{Nested Data Models}
% There is a long history of query languages for nested collections in the literature.
Nested relational model~\cite{roth1988extended} and monad calculus~\cite{monad-calc-1, monad-calc-2, monad-comprehension, query-comprehension, query-comprehension-2, Buneman:1995:PPC:210500.210501} support complex data models but do not support aggregations and efficient equi-joins~\cite{gibbons_icfp18}. 
Monoid comprehensions solve the former issue~\cite{monoid-comprehension}, however, require an intermediate algebra to support equi-joins efficiently.
Kleisli~\cite{wong2000kleisli},
BQL~\cite{libkin1997query}, and Trance~\cite{trance_vldb} extend monad calculus with aggregations and bag semantics.
Representing flat relations as bags has been investigated in AGCA~\cite{dbtoaster}, FAQ~\cite{abo2016faq}, and HoTTSQL~\cite{chu2017hottsql}.
\lang extends all these approaches by allowing nested dictionaries and representing relations and intermediate group-by aggregates as dictionaries. Although monadic and monoid collection structures were observed, \lang is the first work that introduces semi-ring dictionaries.

\smartpara{Language-Integrated Queries} LINQ~\cite{linq} and Links~\cite{Cooper:2006:LWP:1777707.1777724} mainly aim to generate SQL or host language's code from nested functional queries. 
One of the main challenges for them is to resolve avalanche of queries during this translation, for which techniques such as query shredding has proved useful~\cite{ferry-1,cheney2014query}. 
Comprehensive Comprehensions (CompComp)~\cite{jones2007comprehensive} extend Haskell's list comprehensions with group-by and order-by. 
% Trance~\cite{trance_vldb} has used query shredding for resolving load balancing in distributed query processing. 
Rather than only serving as a frontend language and relying on the target language to perform optimizations, \lang takes an approach similar to Kleisli~\cite{wong2000kleisli}; it directly translates nested collections to low-level code, and enables more aggressive optimizations. 

\smartpara{Loop Fusion}
Functional languages use deforestation~\cite{deforestation,foldr-fusion-1,Svenningsson:2002:SFA:581478.581491,Coutts07streamfusion,erikfusion,semiringfusion} to remove unnecessary intermediate collections. This optimization is implemented by rewrite rule facilities of GHC~\cite{jones2001playing} in Haskell~\cite{foldr-fusion-1}, and also by using multi-stage programming in Scala~\cite{fold-based-fusion,Kiselyov:2017:SFC:3009837.3009880,jfppushpull}. Generalized stream fusion~\cite{Mainland:2013:EVI:2500365.2500601} combines deforestation with vectorization for Haskell.
Functional array processing languages such as APL~\cite{iverson1962programming}, SAC~\cite{Grelck2006}, Futhark~\cite{henriksen2017futhark}, and $\widetilde{\text{F}}$~\cite{shaikhha2019efficient} also need to support loop fusion. Such languages mainly use pull and push arrays~\cite{edsl-push,Claessen:2012:EAC:2103736.2103740,Svensson:2014:DPA:2636228.2636231,Axelsson:2010:DIF:2050135.2050143,kiselyov2018reconciling,dps_fhpc} to remove unnecessary intermediate arrays. 
Even though these work support fusion for lists of key-value pairs, they do not support dictionaries. Thus, they do not have efficient support for operators such as grouping and hash join.

\smartpara{Linear Algebra Languages}
DSLs such as Lift~\cite{Steuwer:2015:GPP:2784731.2784754}, 
Halide~\cite{ragan2013halide}, Diderot~\cite{chiw2012diderot}, and OptiML~\cite{sujeeth2011optiml} can
generate parallel code from their high-level programs, while DSLs such as Spiral~\cite{spiral}, LGen~\cite{spampinato2016basic,Spampinato:2018:PGS:3179541.3168812} exploit the memory hierarchy and make careful decisions on tiling and scheduling decisions.
These DSLs exploit the memory hierarchy by relying on searching algorithms for making tiling and scheduling decisions.
The generated output is a C function that includes intrinsics to enable SIMD vector extensions.
SPL~\cite{Xiong:2001:SLC:378795.378860} is a language that expresses recursion and mathematical formulas. 
TACO~\cite{Kjolstad:2017:TAC:3152284.3133901} generates efficient low-level code 
for compound linear algebra operations on dense and sparse matrices.
All these languages are limited to linear algebra workloads and do not support database workloads.

\begin{table}
\begin{scriptsize}
  \caption{Comparison of different data analytics approaches.
 $\supfull$ means that the property is supported, $\supnone$ means that it is absent in the work, and
 $\suphalf$ means that the property is partially supported. \revision{For the corresponding sets of operators supported by (nested) relational and linear algebra refer to Figures~\ref{fig:ra_ndql}-\ref{fig:la_ndql}.}}
 \label{tbl:relwork}
 \vspace{-0.4cm}
 \begin{tabular}{|l|c|c|c|c|c|c|c|c|c|c|c|c|c|c|c|}
 \hline 
 \multirow{2}{*}{}
  & \multicolumn{5}{c|}{Expressiveness} & \multicolumn{5}{c|}{Data Representation} & \multicolumn{5}{c|}{Specialization} \\
 \cline{2-16} 
  & \rot{Relational Algebra}  & \rot{Nested Rel. Calc.} & \rot{Group-by Aggregates} & \rot{Efficient Equi-Joins} & \rot{Linear Algebra} & \rot{Set \& Bag} & \rot{Dense Array} & \rot{Sparse Tensor} & \rot{Dictionary} & \rot{Semi-rings} & \rot{Loop Fusion} & \rot{Loop Hoisting} & \rot{Loop Memoization} & \rot{Code Generation} & \rot{Vectorization} \\
 \hline
 \system{} (This Paper) &
 \supfull & \supfull & \supfull & \supfull  & \supfull  & \supfull & \supfull & \supfull & \supfull & \supfull & \supfull & \supfull & \supfull & \supfull & \supnone \\ \hline
Query Compilers (HyPer) &
 \supfull & \supnone & \supfull & \supfull  & \supnone  & \supfull & \supfull & \supnone & \supfull & \supnone & \suphalf & \suphalf & \supnone & \supfull & \supnone \\ \hline
Vectorized Query Engines (Vectorwise) &
 \supfull & \supnone & \supfull & \supfull  & \supnone  & \supfull & \supfull & \supnone & \supfull & \supnone & \suphalf & \suphalf & \supnone & \supnone & \supfull \\ \hline
 Monad Calculus\revision{, \nrcplus} &
 \supfull & \supfull & \supnone & \supnone  & \supnone  & \supfull & \supnone & \supnone & \supnone & \supnone & \suphalf & \suphalf & \supnone & \supnone & \supnone \\ \hline
 Monoid Comprehension &
 \supfull & \supfull & \supfull & \supnone  & \supnone  & \supfull & \supnone & \supnone & \supnone & \supnone & \suphalf & \suphalf & \supnone & \supnone & \supnone \\ \hline
  Monad Calc. + Agg. (Kleisli, Trance) &
 \supfull & \supfull & \supfull & \supnone  & \suphalf  & \supfull & \supnone & \supnone & \supnone & \supnone & \suphalf & \suphalf & \supnone & \supfull & \supnone \\ \hline
 Lang. Integrated Queries (LINQ, CompComp) &
 \supfull & \supfull & \supfull & \supnone  & \supfull  & \supfull & \supnone & \supnone & \supnone & \supnone & \suphalf & \suphalf & \supnone & \supnone & \supnone \\ \hline
 Functional Lists (Generalized Stream Fusion) &
 \supfull & \supfull & \supfull & \supnone  & \supfull  & \supfull & \suphalf & \supnone & \supnone & \supnone & \supfull & \suphalf & \supnone & \supfull & \supfull \\ \hline
 Functional APL (Futhark, SAC) &
 \suphalf & \suphalf & \suphalf & \supnone  & \supfull  & \suphalf & \supfull & \supnone & \supnone & \supnone & \supfull & \suphalf & \suphalf & \supfull & \supfull \\ \hline
 Dense LA Library (NumPy) &
 \supnone & \supnone & \supnone & \supnone  & \supfull  & \supnone & \supfull & \supnone & \supnone & \supnone & \supnone & \supnone & \supnone & \supnone & \supfull \\ \hline
 Dense LA DSL (Lift,Halide,LGen) &
 \supnone & \supnone & \supnone & \supnone  & \supfull  & \supnone & \supfull & \supnone & \supnone & \supnone & \supfull & \suphalf & \supnone & \supfull & \supfull \\ \hline
 Sparse LA Library (SPLATT, SciPy) &
 \supnone & \supnone & \supnone & \supnone  & \supfull  & \supnone & \supfull & \suphalf & \supnone & \supnone & \supnone & \supnone & \supnone & \supnone & \suphalf \\ \hline
 Sparse LA DSL (TACO) &
 \supnone & \supnone & \supnone & \supnone  & \supfull  & \supnone & \supfull & \supfull & \supnone & \supnone & \suphalf & \suphalf & \supnone & \supfull & \supnone \\ \hline
 \revision{Sparse LA + Semi-rings (GraphBLAS)} &
 \supnone & \supnone & \supnone & \supnone  & \supfull  & \supnone & \supfull & \suphalf & \supnone & \supfull & \supnone & \supnone & \supnone & \supnone & \suphalf \\ \hline
 DB/LA by casting to LA (Morpheus) &
 \suphalf & \supnone & \supfull & \supfull  & \supfull  & \supfull & \supfull & \suphalf & \supnone & \supnone & \supnone & \supnone & \supnone & \supnone & \supfull \\ \hline
 DB/LA by casting to DB (LMFAO) &
 \supfull & \supnone & \supfull & \supfull  & \suphalf  & \supfull & \supfull & \suphalf & \supnone & \supfull & \suphalf & \suphalf & \supnone & \supfull & \supnone \\ \hline
 DB/LA by \revision{unified IR} (IFAQ) &
 \supfull & \supnone & \supfull & \supfull  & \supfull  & \supfull & \supnone & \supfull & \supfull & \supfull & \suphalf & \supfull & \suphalf & \supfull & \supnone \\ \hline
 \revision{DB/LA by combined IR (Raven)} &
 \supfull & \supnone & \supfull & \supfull  & \supfull  & \supfull & \supfull & \suphalf & \supnone & \supnone & \suphalf & \suphalf & \suphalf & \supfull & \supfull \\ \hline
 \end{tabular}
\end{scriptsize}
 \vspace{-0.5cm}
\end{table}

\revision{
\smartpara{Semi-Ring Languages}
The use of semi-rings for expressing graph problems as linear algebra is well-known~\cite{kepner2011graph}. 
This connection has been used for expressing path problems by solving matrix 
equations~\cite{path_tarjan,algebra_path,valiant1975general}. 
\lang requires extensions in order to express such problems (cf. Section~\ref{sec:langext}).
GraphBLAS~\cite{graph_blas} is a framework for expressing graph problems in terms of sparse linear algebra.
The functional languages has shown before an appropriate implementation choice for linear algebra languages with
various semi-ring instances~\cite{pilatus19ecoop,dolan2013fun}.
In the DB world,} K-relations~\cite{green2007provenance} use semi-rings~\cite{karvounarakis2012semiring} and semi-modules~\cite{amsterdamer2011provenance} for encoding provenance information for relational algebra with aggregations. 
\revision{The pvc-tables~\cite{probdb_agg} are a representation system that use this idea to encode aggregations in databases with uncertainties.
The closest work to ours is FAQ~\cite{abo2016faq}, which provides a unified declarative interface for LA and DB. 
However, none of the existing work support nested data models.

\smartpara{DB/LA Query Languages}
There has been a recent interest in the study on the expressive power of query languages for hybrid DB/LA tasks.
Matrix query languages~\cite{matlang_sigrec} such as MATLANG~\cite{matlang_tods} and its extensions have shown to 
be connected to different fragments of relational algebra with aggregates.
LARA~\cite{laradb} is a query language over associative tables (flat dictionaries), 
with more expressive power than MATLANG~\cite{lara_expr}.
Associative algebra~\cite{polystore_query_lang} defines a query language over associative arrays (flat dictionaries, 
and without the ability to map between dictionaries of different value types) expressive enough 
for both database and linear algebra workloads. 
All these query languages are declarative and can only serve as frontend query languages; they need to rely on the techniques offered by other formalisms (e.g., FAQ~\cite{abo2016faq}) for optimizations.
Furthermore, none of these languages support nested data like \lang.
}

\smartpara{DB/LA Frameworks}
Hybrid database and linear algebra workloads, such as training machine learning models over databases are increasingly gaining attention. 
Traditionally, these workloads are processed in two isolated environments: 1) the training data set is constructed using a database system or libraries such as Python Pandas, and then 2) the model is trained over the materialized dataset using frameworks such as scikit-learn~\cite{pedregosa2011scikit}, TensorFlow~\cite{abadi2016tensorflow}, PyTorch~\cite{paszke2017automatic}, etc.
% R~\cite{team2013r}, MLlib~\cite{mllib}, SystemML~\cite{ghoting2011systemml}, or XGBoost~\cite{chen2016xgboost}.
There has been some efforts on avoiding the separation of the environments by defining ML tasks as user-defined functions inside the database system such as 
MADlib~\cite{hellerstein2012madlib}, Bismarck~\cite{Feng:2012:TUA:2213836.2213874}, and GLADE PF-OLA~\cite{qin2015speculative}; however, the training process is still executed after the training dataset is materialized.

Alternative approaches avoid the materialization of the training dataset.
The current solutions are currently divided into \revision{four} categories. 
First, systems such as Morpheus~\cite{chen2017towards,li2019enabling} cast the in-DB ML task as a linear algebra problem on top of R~\cite{chen2017towards} and NumPy~\cite{li2019enabling}. 
\revision{An advantage of this system is that it benefits from efficient linear algebra frameworks (cf. Section~\ref{sec:exp:indbml}). 
However, one requires to encode database knowledge in terms of linear algebra rewrite rules and implement query evaluation techniques for them (e.g., trie-based evaluation as observed in Section~\ref{sec:exp:indbml}).}
The second category are systems such as F~\cite{fdb, Schleich:2016:LLR:2882903.2882939}, AC/DC~\cite{Khamis:2018:AIL:3209889.3209896}, and LMFAO~\cite{Schleich:2019:LAE:3299869.3324961} that cast the in-DB ML task as a \revision{batch of aggregate queries}.
\revision{The third approach involves defining an intermediate representation (IR) 
that \textit{combines} linear and relational algebra constructs together.
Raven~\cite{raven} and MatRel~\cite{matrel} are frameworks that provide such an IR.
For implementing cross-domain optimizations, this approach requires developing new 
transformation rules for different combinations of linear and relational algebra constructs, which can be tedious and error prone.}
The \revision{fourth} category \revision{resolves this issue by defining a unified intermediate} language that can express both workloads.
\revision{Lara~\cite{kunft2019intermediate} provides a two-level IR. The first level combines linear and relational algebra constructs.
The second level is based on monad-calculus and can perform cross-domain optimizations such as vertical loop fusion and selection push down.
IFAQ~\cite{ifaq-cgo,ifaq_ir} introduces a single dictionary-based DSL for expressing the entire data science pipelines.}
\lang also falls into \revision{the fourth} category, and additionally supports nested data\revision{, dense representations, and more loop optimizations (cf. Table~\ref{tbl:relwork})}. Furthermore, to the best of our knowledge, \lang is the only hybrid DB/LA framework for which type safety and the correctness of the optimizations are proved using denotational and operational semantics.

\section{Conclusion}
In this paper, we introduce a statically typed and functional language based on semi-ring dictionaries.
\lang is expressive enough for different data science use-cases with a better or competitive performance relative to specialized systems.
For example, the performance of \lang is competitive with the state-of-the-art in-memory database systems that are especially built for database workloads, and thus cannot efficiently handle other use-cases including sparse linear algebra, and in-database machine learning over different formats of data: nested, relational, and normalized matrix. 
This makes \lang a suitable intermediate language for data science pipelines typically expressed in several languages and executed using different systems.
For future, we plan to add the support for vectorization and parallelization.

\vspace{0.3cm}

\smartpara{Acknowledgements}
This project has received funding from the European Union's Horizon 2020 research and innovation programme under grant agreement No 682588. The authors also acknowledge the EPSRC grant EP/T022124/1 (QUINTON).

\bibliography{refs} 

\clearpage

\appendix

\section{Translation of Relational Algebra}
In this section, we explain the translation of relational operators to \lang, as shown in Figure~\ref{fig:ra_ndql}.

\smartpara{Selection}
Consider the translation of relation R in \lang, which is represented as \translate{R}. 
The selection operator, represented as $\sigma_p$(R), filters the elements that satisfy a predicate $p$.
For each element !x! of this relation, if the predicate is satisfied, we return the singleton set containing !x.key!.
Otherwise, we return an empty set.
As the body of the loop returns a set, this loop performs a set union, which results in a filtered relation.

\smartpara{Projection}
This operator projects a subset of attributes specified by the function $f$, and is represented as $\pi_f$(R). Similar to selection, we iterate over the 
elements of \translate{R}; at each iteration we return a singleton set 
with the applied projection 
function !f! on each row of relation (!f(x.key)!). 
% Similar to selection, we iterate over the elements of \translate{R}. 
% At each iteration, we return a singleton set with the application of the projection function !f! on each row of relation (!f(x.key)!).

\smartpara{Union} Set union is achieved by using the $+$ operation of the boolean semi-ring, i.e., boolean disjunction, on the values of elements with the same key. For the elements that exist only in one of the collections, the value associated with the element in the other collection is considered as !false!.

\smartpara{Intersection} Set intersection is achieved by iterating over the elements of the first collection. At each iteration, if an element with the same key exists in the other collection, a singleton set of that element key is returned, otherwise an empty set is returned.

\smartpara{Difference} This operator is symmetric to intersection with the difference that if the element key exists in the second collection, an empty set is returned. Otherwise, a singleton set is returned.

\smartpara{Cartesian Product}
For this operation, we use nested loops iterating over the elements !x! and !y! of relations R and S, respectively.
For each combination of tuples, we return a singleton set that has the combination of the tuples of these two relations as its element.

\smartpara{Inner Join}
This operator is expressed similarly to the Cartesian product operator.
For each combination of elements, if the selection predicate is satisfied the joined tuples are emitted.

\smartpara{Semi Join}
R left semijoin S is the set of all tuples in R for which there is a matched tuple in S. It can be simulated using a natural join followed by the projection over the attributes of R. The function concat joins two tuples from R and S, whereas !projR! extracts a record with attributes of R from the joined tuple:

\begin{lstlisting}
let RjS = sum(x in R) sum(y in S) 
  if(join(x.key, y.key)) then { concat(x.key, y.key) } in
sum(x in RjS) { projR(RjS) }
\end{lstlisting}

\smartpara{Anti Join}
R left antijoin S is similar to the semijoin, except that its result is only those tuples of R for which there is no matched tuple in S. It can be expressed by substracting R left semi-join S from R:

\begin{lstlisting}
let RjS = sum(x in R) sum(y in S) 
  if(join(x.key, y.key)) then { concat(x.key, y.key) } in
let RsjS = sum(x in RjS) { projR(RjS) } in
sum(x in R) if(not(RsjS(x.key))) then { x.key }
\end{lstlisting}

Right semi/anti join can be expressed similarly.

\smartpara{Outer Join}
R left outer join S is expressed as:

\begin{lstlisting}
let RjS = sum(x in R) sum(y in S) 
  if(join(x.key, y.key)) then { <r=x.key, s={y.key}> } in
let Rproj = sum(xy in RjS) { xy.key.r } in
let RpS = sum(x in R) if(not(Rproj(x.key))) then { <r=x.key, s={}> }
in RjS + RpS
\end{lstlisting}

The expression !RjS! corresponds to R inner join with S, Rproj corresponds to the projection over the attributes of R, and !RpS! corresponds to the elements of R that didn't join with any element from S padded with NULL, which is !{}! in this case. Finally, we compute the union of !RjS! and !RpS!. The right outer join and full outer joins are expressed similarly.

\section{Translation of Nested Relational Caculus}

\smartpara{Bag Construction}
The empty bag construction is expressed by an empty dictionary. 
The singleton bag construction is achieved by constructing a dictionary with the given element as its key, and the multiplicity of one as its value.

\smartpara{Flattening}
The flattening operation is only performed on an input parameter that is nested. 
Thus, the translated input is a dictionary where the key is also a dictionary.
In order to flatten the input dictionary, one has to union its keys; however, one needs to multiply the multiplicity with the keys to take into account the bag semantics, represented as !x.val * x.key!.

\smartpara{For-comprehensions}
Similar to the flattening operator, bag semantics 
requires that the translation of the body of the loop
be multiplied by the multiplicity.

\smartpara{Bag Union}
Similar to set union in relational algebra, bag union in \nrcplus is also achieved by addition on the translation of the operands.

\smartpara{Bag Product}
Similar to Cartesian Product in relational algebra, we iterate over each combination of the elements of the two inputs. 
The key of the result dictionary is a pair of the keys of inputs.
The value is the multiplication of the values of two inputs in order to take into account the bag semantics.

\section{Translation of Aggregations}

\smartpara{Scalar Aggregate}
This operator can be implemented by
iterating over the elements of the relation and computing the appropriate aggregate function !f! (cf. first and third rules of Figure~\ref{fig:agg_ndql}). 
As the relations have bag semantics, there could be duplicates of an element in the input relation, the multiplicity of which is shown by !x_v!; thus, the aggregate result for each element needs to be multiplied by !x_v!. 
The following example shows the translation of sum and count queries:

\begin{center}
\begin{tabular}{l c l}
\begin{lstlisting}[language=SQL,frame=none]
SELECT SUM(R.A) FROM R
\end{lstlisting}&\transto&
\begin{lstlisting}[frame=none]
sum(<r,r_v> in R) r_v * r.A
\end{lstlisting}\\
\begin{lstlisting}[language=SQL,frame=none]
SELECT COUNT(*) FROM R
\end{lstlisting}&\transto&
\begin{lstlisting}[frame=none]
sum(<r,r_v> in R) r_v
\end{lstlisting}
\end{tabular}
\end{center}

\smartpara{Group-by Aggregate} As opposed to its scalar variant, a group-by aggregate returns a single dictionary with the key specified by the grouping function !g!, and the value specified using the aggregate function !f! (cf. second and fourth rules of Figure~\ref{fig:agg_ndql}).
The following example shows the translation of group-by sum and group-by count queries:

\begin{tabular}{l c l}
\begin{lstlisting}[language=SQL,frame=none]
SELECT SUM(R.A) FROM R
GROUP BY R.B
\end{lstlisting}&\transto&
\begin{lstlisting}[frame=none]
let tmp = sum(<r,r_v> in R) { r.B -> r_v * r.A }
in sum(<x,x_v> in tmp) { <key=x, val=x_v> -> 1 }
\end{lstlisting}\\
\begin{lstlisting}[language=SQL,frame=none]
SELECT COUNT(*) FROM R
GROUP BY R.B
\end{lstlisting}&\transto&
\begin{lstlisting}[frame=none]
let tmp = sum(<r,r_v> in R) { r.B -> r_v }
in sum(<x,x_v> in tmp) { <key=x, val=x_v> -> 1 }
\end{lstlisting}
\end{tabular}

\smartpara{Nest} This operator performs grouping without aggregation. 
This means that the output is a nested relation, that is only supported by \nrcagg, not relational algebra. 
Similar to the group-by aggregate operator, at each iteration it returns a singleton dictionary with the key specified by the function !g!, and the value is a singleton dictionary with !x! as key and !x_v! as value (cf. last rule of Figure~\ref{fig:agg_ndql}).

\section{Translation of Linear Algebra}

\smartpara{Vector Addition} This operation is expressed as the addition of the translated dictionaries in \lang, which results in a dictionary where the values of the elements with the same key are summed. 

\smartpara{Scalar-Vector Multiplication} The multiplication of a scalar value with a vector is translated to the multiplication of the translated scalar \lang expression with the translated dictionary. The result expression is evaluated to a dictionary with the same keys as the translated dictionary, with the values multiplied by the scalar expression. 

\smartpara{Vector Hadamard Product} The Hadamard product of two vectors, or the element-wise multiplication, is achieved by iterating over the elements of the translated dictionary of the first vector, 
and constructing a singleton dictionary with its key, and the value multiplied by looking up the value associated with the translated dictionary of the second vector.
If an element with the same key (index) does not exist in the second dictionary, the singleton dictionary will have a zero value. 
When the result dictionary is constructed, singleton dictionaries are ignored when computing the union of the intermediate dictionaries.

\smartpara{Vector Dot Product} This operation is achieved by iterating over the elements of the translation of the first vector, and summing the multiplication of the associated value and the corresponding value of the second vector. 

\smartpara{Vector Summation} Finally, the summation of the elements of a vector is expressed by iterating over the elements of the vector and adding its values.

\smartpara{Matrix Transpose} The transposition of a matrix is expressed by iterating over the elements of the translated dictionary and constructing a dictionary where the key is a record with the same value, but with the row and column swapped.

\smartpara{Matrix Addition, Scalar-Matrix Multiplication, Matrix Hadamard Product} These operators are expressed similarly to the corresponding vector oporators.

\smartpara{Matrix-Matrix Multiplication} This operator is expressed by iterating over each combination of the elements of two matrices. At each iteration, if the column of the element from the first matrix is the same as the row of the element of the second matrix, a singleton dictionary is created, where the key is with the row of the first element, and column of the second element, and the value is the multiplication of both elements. Otherwise, and empty dictionary is created. 

\smartpara{Matrix-Vector Multiplication} This operator is expressed by iterating over the elements of the translated matrix, and constructing a singleton dictionary where the key is the row of this element, and the value
is the multiplication of this element with the element from the vector associated with its column.

\smartpara{Matrix Trace} The trace of a matrix is the result of summation of the elements of a diagonal of a matrix. 
This can be expressed by iterating over the elements of a matrix and adding the value of that element if its row and column are identical, and otherwise adding zero.

\section{Translation of Curried Linear Algebra}
Figure~\ref{fig:cla_ndql_full} shows the translation of matrix operator, in the case of using curried representation for them.

\begin{figure*}[t]
\begin{tabular}{|l|r c l|c|}
\hline
\textbf{Name} & \multicolumn{3}{|l|}{\textbf{Translation}} & \textbf{Einsum} \\ \hline
Transpose&\translate{$M_1^T$} &=& !sum(row in! \translate{$M_1$} !) sum(x in row.val)!& !ij->ji!\\
&&&\tab!{ x.key -> { row.key -> x.val } }!&\\ \hline
Had. Prod.&\translate{$M_1 \circ M_2$} &=& !sum(row in! \translate{$M_1$} !) { row.key -> ! & !ij,ij->ij! \\
&&&\tab!sum(x in row.val) { x.key ->! & \\ 
&&&\tab\tab !x.val*!\translate{$M_2$}!(row.key)(x.key) } }! &\\ \hline
Matrix-Matrix&\translate{$M_1 \times M_2$} &=& !sum(row in! \translate{$M_1$} !) { row.key -> ! & !ij,jk->ik! \\
Multiplication&&&\tab!sum(x in row.val) sum(y in! \translate{$M_2$}!(x.key))! &\\ 
&&&\tab\tab! { y.key -> x.val * y.val } }!&\\ \hline
Mat-Vec. Mul.&\translate{$M \cdot V$} &=& !sum(row in! \translate{$M$} !) { row.key ->! & !ij,j->i!\\
&&&\tab!sum(x in row.val) x.val * !\translate{$V$}!(x.key) }!&\\\hline
Trace&\translate{$Trace(M)$} &=& !sum(row in! \translate{$M$} !) row.val(r.key)! & !ii->!\\ \hline
\end{tabular}
\caption{Translation of curried matrix operations to \lang.}
% \caption{Translation of matrix-matrix multiplication for curried matrices to \lang.}
\label{fig:cla_ndql_full}
\end{figure*}

\section{Correctness of loop optimizations}

\begin{proposition}
The horizontal loop fusion rules of Figure~\ref{fig:opt_rules} are sound.
\end{proposition}

\begin{proof}[Proof]
\dsbegin!let y1=sum(x in e1) f1(x) in let y2=sum(x in e1) f2(x) in f3(y1,y2)!\dsend{} \tab \\
= \dsbegin!let y2=sum(x in e1) f2(x) in f3(y1,y2)!\dsgend{$\gamma'$} \tab ($\gamma'$ = $\gamma$[\dsbegin!sum(x in e1) f1(x)!\dsend{}/ y1]) \\
= \dsbegin!let y2=sum(x in e1) f2(x) in f3(y1,y2)!\dsgend{$\gamma'$} ($\gamma'$ = $\gamma$[$\sum\limits_{k \in X}$\densem{\code{f1}}<$k,a_k$>/ y1], \densem{\code{e1}}=$\sum\limits_{k \in X}a_k\mytimes k$) \\
= \dsbegin!f3(y1,y2)!\dsgend{$\gamma'$} ($\gamma'$ = $\gamma$[$\sum\limits_{k \in X}$\densem{\code{f1}}<$k,a_k$>/ y1,\dsbegin!sum(x in e1) f2(x)!\dsend{} /y2], \densem{\code{e1}}=$\sum\limits_{k \in X}a_k\mytimes k$) \\
= \dsbegin!f3(y1,y2)!\dsgend{$\gamma'$} ($\gamma'$ = $\gamma$[$\sum\limits_{k \in X}$\densem{\code{f1}}<$k,a_k$>/ y1,$\sum\limits_{k \in X}$\densem{\code{f2}}<$k,a_k$> /y2], \densem{\code{e1}}=$\sum\limits_{k \in X}a_k\mytimes k$) \\
= \dsbegin!f3!\dsend{}$\big(\sum\limits_{k \in X}$\densem{\code{f1}}<$k,a_k$>, $\sum\limits_{k \in X}$\densem{\code{f2}}<$k,a_k$>$\big)$ \tab (\densem{\code{e1}}=$\sum\limits_{k \in X}a_k\mytimes k)$ \\
= \dsbegin!f3(tmp.y1, tmp.y2)!\dsgend{$\gamma'$} \tab ($\gamma'$ = $\gamma$[<$\sum\limits_{k \in X}$\densem{\code{f1}}<$k,a_k$>,\densem{\code{f2}}<$k,a_k$> > /tmp], \densem{\code{e1}}=$\sum\limits_{k \in X}a_k\mytimes k)$ \\
= \dsbegin!let tmp = sum(x in e1) <y1=f1(x), y2=f2(x)> in f3(tmp.y1, tmp.y2)!\dsend{}
\end{proof}

\begin{proposition}
The rewrite rule for loop-invariant code motion in Figure~\ref{fig:opt_rules} is sound.
\end{proposition}

\begin{proof}[Proof]
\dsbegin!sum(x in e1) let y = e2 in f(x, y)!\dsend{} \\
= $\sum\limits_{k \in X}$\dsbegin!let y = e2 in f(x, y)!\dsgend{$\gamma'$} \tab ($\gamma'$=$\gamma$[<$k,a_k$>/ x], \densem{\code{e1}}=$\sum\limits_{k \in X}a_k\mytimes k$)\\
= $\sum\limits_{k \in X}$\dsbegin!f(x, y)!\dsgend{$\gamma'$} \tab ($\gamma'$=$\gamma$[<$k,a_k$>/ x, \dsbegin!e2!\dsend{} / y], \densem{\code{e1}}=$\sum\limits_{k \in X}a_k\mytimes k$)\\
= \dsbegin!sum(x in e1) f(x, y)!\dsgend{$\gamma'$} \tab ($\gamma'$=$\gamma$[\dsbegin!e2!\dsend{} / y]) \\
= \dsbegin!let y = e2 in sum(x in e1) f(x, y)!\dsend{}
\end{proof}

\begin{proposition}
The rewrite rules for loop memoization in Figure~\ref{fig:opt_rules} are sound.
\end{proposition}

\begin{proof}[Proof]
We prove the second rule, and the first rule is proved similarly.
\\
\dsbegin!sum(x in e1) if(p(x) == e2) then f(x)!\dsend{} \tab \\
= $\sum\limits_{k \in X}$\dsbegin!if(p(x) == e2) then f(x)!\dsgend{$\gamma'$} \tab ($\gamma'$=$\gamma$[<$k,a_k$>/ x], \densem{\code{e1}}=$\sum\limits_{k \in X}a_k\mytimes k$) \tab (!x! $\notin$ FVs of !f!, !p!, !e2!) \\
= $\sum\limits_{k \in X}$\dsbegin!f!\dsend{$\gamma'$}($a_k, k$) $*$ (\dsbegin!p!\dsend{$\gamma'$}($a_k, k$) == \dsbegin !e2! \dsend{$\gamma'$}) \tab ($\gamma'$=$\gamma$[<$k,a_k$>/ x], \densem{\code{e1}}=$\sum\limits_{k \in X}a_k\mytimes k$) \\
= $\sum\limits_{k \in X}$\dsbegin!f!\dsend{$\gamma'$}($a_k, k$) \tab ($\gamma'$=$\gamma$[<$k,a_k$>/ x], \densem{\code{e1}}=$\sum\limits_{k \in X}a_k\mytimes k$, \dsbegin!p!\dsend{}($a_k, k$) == \dsbegin !e2! \dsend{}) \\
= $\pi_{\text{\dsbegin \code{e2} \dsend{$\gamma'$}}}\big(\sum\limits_{k \in X}$\dsbegin!f!\dsend{$\gamma'$}($a_k, k$) $\mytimes$ 
\dsbegin !p! \dsend{$\gamma'$}($a_k, k$)$\big)$ \tab 
($\gamma'$=$\gamma$[<$k,a_k$>/ x], \densem{\code{e1}}=$\sum\limits_{k \in X}a_k\mytimes k$) \\
= $\pi_{\text{\dsbegin \code{e2} \dsend{}}}\big($ \dsbegin !sum(x in e1) {p(x)->f(x)}! \dsend{} $\big)$ \\ 
= \dsbegin!let tmp=sum(x in e1) {p(x)->f(x)} in tmp(e2)!\dsend{}
\\
\end{proof}

\begin{figure*}[t]
\centering
\setlength{\tabcolsep}{0.3em}
\centering
\begin{tabular}{|l c l|}
\hline
\multicolumn{3}{|l|}{\grammarcomment{Evaluation contexts}} \\
$E$ & \mbox{::=} & !sum(x in $E$) e! $\mid$ !$E$(e)! $\mid$ !v($E$)!  %\\

$\mid$
!let x = $E$ in e!  $\mid$ !if($E$) then e else e!\\
& $\mid$ & 
!{ v -> v, ..., $E$ -> e, ... }! %\\
% & $\mid$ & 
$\mid$
!{ v -> v, ..., v -> $E$, ... }! \\
& $\mid$ & !< a_1 = v, ..., a_i = $E$, ... >!  $\mid$ !$E$.a!
$\mid$ !$E$ * e!  $\mid$ !v * $E$! $\mid$ !$E$ + e!  $\mid$ !v + $E$! \\
& $\mid$ & !promote$_{\texttt{S},\texttt{S}}$($E$)! $\mid$  $[]$ \\ \hline
\multicolumn{3}{|l|}{\grammarcomment{Values}} \\ 
!v! & \mbox{::=} & !{ v -> v, ... }! $\mid$ !< a = v, ... >! 
$\mid$ !n! $\mid$ !r! $\mid$ !false! $\mid$ !true! $\mid$ \zero{T}\\
\hline
\end{tabular}
\vspace{0.2cm}
\begin{tabular}{c}
\\\hline
!sum(x in{k0->v0,...})e2! \evalsto !e2![!<key=k0,val=v0>!/!x!]!+sum(x<-{k1->v1,...})e2!
\end{tabular}

\begin{tabular}{c}
!v1,v2:S!
\\\hline
!(v1)+(v2)! \evalsto !(v1+v2)!
\end{tabular}
\begin{tabular}{c}
!v1,v2:S!
\\\hline
!(v1)*(v2)! \evalsto !(v1*v2)!
\end{tabular}
\begin{tabular}{c}
\\\hline
!promote$_{\texttt{S1},\texttt{S2}}$(v)! \evalsto !v!
\end{tabular}

\begin{tabular}{c}
!e2!: !T!\\\hline
!sum(x in{})e2! \evalsto \zero{T}
\end{tabular}
\hspace{0.3cm}
\begin{tabular}{c}
\\\hline
!let x=v in e2! \evalsto !e2![!v!/!x!]
\end{tabular}
\hspace{0.3cm}
\begin{tabular}{c}
\\\hline
!<a_0=e_0,...>.a_i! \evalsto !e_i!
\end{tabular}

\begin{tabular}{c}
\\\hline
!{k_0->v0_0, ...} + {k_0->v1_0, ...}! \evalsto !{k_0->v0_0+v1_0, ...}!
\end{tabular}

\begin{tabular}{c}
\\\hline
!<a_0=v0_0, ...> + <a_0=v1_0, ...>! \evalsto !<a_0=v0_0+v1_0, ...>!
\end{tabular}

\begin{tabular}{c}
!v1! : !T3!
\\\hline
!{ }$_{\texttt{T1},\texttt{T2}}$ * v1! \evalsto !{ }$_{\texttt{T1},\texttt{T2}\otimes\texttt{T3}}$!
\end{tabular}
\hspace{0.5cm}
\begin{tabular}{c}
\\\hline
!{k_0->v_0, ...} * v1! \evalsto !{k_0->(v_0*v1), ...}!
\end{tabular}

\begin{tabular}{c}
\\\hline
!<a_0=v_0, ...> * v1! \evalsto !<a_0=(v_0*v1), ...>!
\end{tabular}

\begin{tabular}{c}
!v1!: !S! \\\hline
!v1 * { }$_{\texttt{T1},\texttt{T2}}$! \evalsto !{ }$_{\texttt{T1},\texttt{T2}}$!
\end{tabular}
\hspace{0.5cm}
\begin{tabular}{c}
!v1!: !S! \\\hline
!v1 * {k_0->v_0, ...}! \evalsto !{k_0->(v1*v_0), ...}!
\end{tabular}

\begin{tabular}{c}
!v1!: !S! \\\hline
!v1 * <a_0=v_0, ...>! \evalsto !<a_0=(v1*v_0), ...>!
\end{tabular}

\begin{tabular}{c}
$\exists$!j!. !k_j! = !k1!\\\hline
!{ k_0 -> v_0, ... }(k1)! \evalsto !v_j!
\end{tabular}
\hspace{0.5cm}
\begin{tabular}{c}
$\nexists$!j!. !k_j! = !k1! \quad $\forall$!i!. !v_i!: !T!\\\hline
!{ k_0 -> v_0, ... }(k1)! \evalsto \zero{T}
\end{tabular}

\begin{tabular}{c}
\\\hline
!if(true) then e1 else e2! \evalsto !e1!
\end{tabular}
\hspace{0.5cm}
\begin{tabular}{c}
\\\hline
!if(false) then e1 else e2! \evalsto !e2!
\end{tabular}
\caption{Reduction rules for \lang.}
\label{fig:evalsem}
\end{figure*}

\section{Operational Semantics}

We now give a standard call-by-value small-step operational semantics to \lang. The syntax for evaluation context and values as well as reduction rules are shown in Figure~\ref{fig:evalsem}. All our types form a semi-ring with zero denoted by \zero{T}. \zero{T} is a macro, defined by induction on !T! as follows. \zero{S} is the constant 0 of the scalar type !S!. 
\zero{\texttt{< a:T, ... >}}= \texttt{< a:\zero{T}, ... >}. 
\zero{T1 -> T2}=!{ }$_{\texttt{T1},\texttt{T2}}$!. For construction of records and dictionaries with multiple arguments, the evaluation order is from left to right. Next, we introduce some lemmas.

\begin{lemma}[Confluence]
Let $\Gamma\vdash$ !e!: !T!. If !e! \evalsto !e1! and !e! \evalsto !e2!, there exists !e!' such that !e1! \evalsto$^*$ !e!' and !e2!\evalsto$^*$ !e!'.  
\end{lemma}

\begin{proof}[Proof Sketch]
By inspection, the only non deterministic cases are dictionary addition and !sum! (that requires ranging over a dictionary). Technically, our dictionaries are unordered. This allows + on semi-ring dictionaries to be commutative. 
% A formal way to express that is to define a congruence $\equiv$ on terms such that for every permutation $\sigma:\{1,..,n\}\to\{1,..,n\}$, !{ k1 -> v1, ..., kn ->vn }!$\equiv$ 
% !{ k!$\sigma$!(1) -> v!$\sigma$!(1), ..., k!$\sigma$!(n) ->v!$\sigma$!(n) }!. The complication is then that the evaluation semantics cannot be made deterministic for reduction involving dictionaries without assuming an order on the keys. This can be dealt with by noting that the evaluation has the confluence property. Then every term can be shown to be weakly terminating, and hence strongly terminating. This slightly complicates the exposition of the operational semantics so we will simply assume in this section that the keys are ordered.
\end{proof}

\begin{lemma}[Type Preservation]
\label{lem:tp}
	If $\Gamma\vdash$ !e: T! and !e! $\to$ !e'! then $\Gamma\vdash$ !e': T!.
\end{lemma}

\begin{proof}[Proof Sketch]
 By induction on the structure of !e! and case analysis on each reduction rule.
\end{proof}

\begin{lemma}[Fundamental lemma]
\label{lemma:prog}
For every !x1:T1,...,xn:Tn! $\vdash$ !e: T! and every value !v1:T1,...,vn:Tn!, !e[v1/x1,...,vn/xn]! reduces to a value.
\end{lemma}

\begin{proof}[Proof Sketch]
 By induction on the structure of !e!, then case analysis on each typing rule. As usual, the quantification is for all !n! and not for fixed !n!. 
\end{proof}

\begin{theorem}
Every closed and well-typed term !e! reduces to a unique value.
\end{theorem}

\begin{proof}[Proof Sketch]
By choosing $\Gamma=\emptyset$ in Lemma~\ref{lemma:prog}.
\end{proof}

\revision{\section{Soundness of the Denotational semantics}}

\begin{proof}[Proof of the substitution lemma]
We only show the non standard cases.

\begin{itemize}
    \item Case of dictionary creation:\\
        \dsbegin !{ki -> vi}! \dsend{}![!\dsbegin !e!\dsend{}!/x]! \\
        = ($\sum_i$ \dsbegin !ki!\dsend{} $\mytimes$ \dsbegin !vi! \dsend{})![!\dsbegin !e!\dsend{}!/x]! \\
        = $\sum_i$ \dsbegin !ki!\dsend{}![!\dsbegin !e!\dsend{}!/x]! $\mytimes$ \dsbegin !vi! \dsend{}[\dsbegin !e!\dsend{}!/x]! \\
        = $\sum_i$ \dsbegin !ki[e/x]!\dsend{} $\mytimes$ \dsbegin !vi[ e/x]!\dsend{} \quad\quad\text{by I.H.} \\
        = \dsbegin !{ki[e/x] -> vi[e/x]}! \dsend{} \\
        = \dsbegin !{ki-> vi}[e/x]! \dsend{}
    \item Case of !sum! introduction:\\
    \dsbegin !sum (x in e1)e2! \dsend{}![!\dsbegin !e!\dsend{}!/y]! \\
    = $\sum_{x\in X}$ \dsbegin !e2! \dsend{$\gamma'$} \quad ($\gamma''=\gamma$[<ak,k>/x, \dsbegin !e!\dsend{}/y], \dsbegin !e1!\dsend{$\gamma'$}= $\sum_k$ ak$\mytimes$ k)\\
    = $\sum_{x\in X}$ \dsbegin !e2! \dsend{$\gamma'$}[\dsbegin !e!\dsend{}/y] \quad ($\gamma'=\gamma$[<ak,k>/x], \dsbegin !e1!\dsend{$\gamma'$}= $\sum_k$ ak$\mytimes$ k)\\
    = $\sum_{x\in X}$ \dsbegin !e2[e/y]! \dsend{$\gamma'$} \quad ($\gamma'=\gamma$[<ak,k>/x], \dsbegin !e1[e/y]!\dsend{}= $\sum_k$ ak$\mytimes$ k) \quad \text{by I.H.}\\
    = \dsbegin !(sum (x in e1[e/y])e2[e/y])! \dsend{} \\
    = \dsbegin !(sum (x in e1)e2)[e/y]! \dsend{}
\end{itemize}
\end{proof}

\begin{proof}[Proof of the soundness theorem]
Most rules follow from the S-semi-module structure of types, or standard denotational semantics in sets and functions.
The only non standard case is !sum!, but the result follows from associativity of addition, and 0 being the unit of addition.
\end{proof}

\section{Correctness of Optimizations using Operational Semantics}
We prove correct the optimizations of Figure~\ref{fig:opt_rules}. As is usual, we denote by $\to^*$ the transitive reflexive closure of $\to$. We say a rule !e!$\transto$!e'! is sound (w.r.t. the evaluation semantics) if !e! and !e'! have the same operational semantics, i.e. !e! $\to^*$ !v! iff !e'! $\to^*$ !v!.

\begin{proposition}[Correctness of Vertical Loop Fusion]
% If $\vdash$ !e: T! and !e! $\transto$ !e'! using the vertical loop fusion rules of Figure~\ref{fig:opt_rules}, then !e! $\to^*$ !v! iff !e'! $\to^*$ !v!.
The vertical loop fusion rules of Figure~\ref{fig:opt_rules} are sound.
\end{proposition}

\begin{proof}[Proof Sketch]
The correctness of the first rule can be proved by performing induction on the value of the dictionary !d={k$_1$->v$_1$, ..., k$_{n+1}$->v$_{n+1}$}! where !e1! $\to^*$ !d!. The correctness of the base case !d={}! is obvious. For the induction step, one has to consider different cases based on whether !f1(k$_{n+1}$)! is equivalent to !f1(k$_i$)! for $i \leq n$. If this is the case the proof is straightforward. If this is not the case, there will be two further cases. Assuming !f1(k$_{n+1}$)!$\to^*$!k'$_{p+1}$!, either !f2(k'$_{p+1}$)! is equivalent to !f2(k'$_j$)! for some $j \leq p$.In each case, both LHS and RHS are evaluated to the same value.

The correctness of the second rule can be proved by simply computing the result of the evaluation of both the LHS and RHS for an arbitrary dictionary value for !e1!.
\end{proof}

\begin{proposition}[Correctness of Horizontal Loop Fusion]
% If $\vdash$ !e: T! and !e! $\transto$ !e'! using the vertical loop fusion rules of Figure~\ref{fig:opt_rules}, then !e! $\to^*$ !v! iff !e'! $\to^*$ !v!.
The horizontal loop fusion rules of Figure~\ref{fig:opt_rules} are sound.
\end{proposition}

\begin{proof}[Proof Sketch]
Straightforward by induction on the value of dictionary !d! which is the result of evaluating !e1!.
\end{proof}

\begin{proposition}[Correctness of Loop Factorization]
\label{theorem:loopfact}
The loop factorization rules of Figure~\ref{fig:opt_rules} are sound.
\end{proposition}

\begin{proof}[Proof Sketch]
By induction on the values of the dictionary !d! which is the result of evaluating !e1!. For the inductive step, we use the distributive law of the semi-ring structure.
\end{proof}

\begin{proposition}[Correctness of Loop-Invariant Code Motion]
The rewrite rule for loop-invariant code motion in Figure~\ref{fig:opt_rules} is sound.
\end{proposition}

\begin{proof}[Proof Sketch]
% We can rewrite the expression on the LHS of this rewrite rule as follows:

% !sum(x in e1) (let y = e2 in 1)*f(x, y)!

% Based on Proposition~\ref{theorem:loopfact}, this expression can be rewritten as follows:

% !(let y = e2 in 1)*sum(x in e1) f(x, y)!

% Finally, we can rewrite the let binding as follows:

% !let y = e2 in sum(x in e1) f(x, y)!

!e1! reduces to a value !{ k1 -> v1, ..., kn -> vn }!.
The LHS reduces to $\sum_i$  !(let y = e2 in 1)*f(x, y)[ki,vi/x]!, where $\sum_i g_i$ is a shorthand for $g_1 + ... + g_n$.
Assuming !e2! reduces to a value !v!,
the first element of the summation reduces to !f(x, y)[k1,v1/x, v/y]!.
This term then reduces to a value !f1!. 
Similarly, for each !i!,  !f(x, y)[ki,vi/x, v/y]! reduces to !fi!.
Hence, the LHS eventually reduces to  $\sum_i$ !fi!.
In the RHS, !e2! reduces first to the value !v!. Then the RHS reduces to !sum(x in e1) f(x, y)[v/y]!.
We then conclude as before. !e1! reduces to a value !{ k1 -> v1, ..., kn -> vn }! and the RHS reduces to $\sum$ !fi!.
In summary, what makes this optimization correct is that substituting !x! then !y! is the same as substituting !y! then !x!.
\end{proof}

\begin{proposition}[Correctness of Loop Memoization]
The rewrite rule for loop memoization in Figure~\ref{fig:opt_rules} is sound.
\end{proposition}

\begin{proof}[Proof Sketch]
By induction on the dictionary !d! which is the result of evaluating !e1!.
\end{proof}

\end{document}